\documentclass[12pt, a4paper]{amsart}

\usepackage{amsmath,amssymb,amsfonts}
\usepackage{bbold}
\usepackage{epic,eepic}
\usepackage{enumerate}

\usepackage{amsthm}
\usepackage[all]{xy}
\usepackage{caption}
\usepackage{stmaryrd} 
\usepackage[dvipdfmx]{graphicx,color} 
\usepackage{tikz}
\usetikzlibrary{trees}

\usepackage{geometry}
\usepackage{listings} 

\theoremstyle{plain}
\newtheorem{thm}{Theorem}[section]
\newtheorem{prop}[thm]{Proposition}
\newtheorem{lem}[thm]{Lemma}
\newtheorem{cor}[thm]{Corollary}

\theoremstyle{definition}
\newtheorem{defn}[thm]{Definition}
\newtheorem{exmp}[thm]{Example}
\newtheorem{Note}[thm]{Note}
\newtheorem{rem}[thm]{Remark}
\newtheorem{asm}[thm]{Assumption}

\numberwithin{figure}{section}
\newenvironment{diagram}[1][\fps@figure]{ 
  
  \begin{figure}[#1]
}{
  \end{figure}
}

\numberwithin{table}{section}

\DeclareMathOperator*{\esssup}{ess\,sup}

\newcommand{\lspace} {
  \vspace{0.8\baselineskip}
}

\newcommand{\abs}[1]{
  \lvert  #1 \rvert
}

\newcommand{\norm}[1]{
  \|  #1 \|
}

\newcommand\Prefix[3]{\vphantom{#3}#1#2#3}  

\newdir{ >}{{}*!/-5pt/@{>}}

\definecolor{arrowred}{rgb}{0,0,0} 
\newcommand{\newword}[1]{\textbf{\textit{#1}}}

\newcommand{\textred}[1]{#1}

\numberwithin{equation}{section}

\newcommand{\Prob}{\mathbf{Prob}}
\newcommand{\mpProb}{\mathbf{mpProb}}

\newcommand{\Unc}{\mathbf{Unc}}
\newcommand{\mpUnc}{\mathbf{mpUnc}}
\newcommand{\USeq}{\mathbf{USec}}

\newcommand{\Mble}{\mathbf{Mble}} 
\newcommand{\Ban}{\mathbf{Ban}} 

\mathchardef\mhyphen="2D  

\def \JAFEE		{0} 

\if \JAFEE 1
\title{Uncertainty Spaces}
\else 
\title{Hierarchical Structure of Uncertainty}
\fi 

\thanks{This work was supported by JSPS KAKENHI Grant Number 24K04941.}

\author[T. Adachi]{Takanori Adachi}
\address{Graduate School of Management,
         Tokyo Metropolitan University,
         1-4-1 Marunouchi, Chiyoda-ku, Tokyo 100-0005, Japan}
\email{Takanori Adachi <taka.adachi@tmu.ac.jp>}

\date{\today}



\def \MkLonger  {0} 
\def \WithProof {1} 
\def \InclCMFun {1} 
\def \InclPM    {1} 
\def \InclSC    {0} 

\if \JAFEE 1
\def \InclPM    {0} 
\def \InclSC    {0} 
\def \InclCMFun {0} 
\def \MkLonger  {0} 
\def \WithProof {0} 
\fi 

\if \MkLonger 1 
\def \WithProof {1} 
\fi 

\begin{document}

\maketitle

\if \InclSC  1
\input{whatIwanttodo}
\else 
We introduce a new concept called \textit{uncertainty spaces} which is an extended concept of probability spaces.
Then, we express $n$-layer uncertainty
by
a hierarchically constructed sequence of uncertainty spaces, called a \textit{U-sequence}.
We use U-sequences for providing examples that illustrate Ellsberg's paradox.
We will use the category theory to get a bird's eye view of the hierarchical structure of uncertainty.
We discuss maps between uncertainty spaces and maps between U-sequences,
seeing that they form categories of uncertainty spaces and the category of U-sequences, respectively.
We construct an endofunctor 
of the category of measurable spaces
in order to embed a given U-sequence into it.
Then,
by the iterative application of the endofunctor,
we construct the \textit{universal uncertainty space}
which may be able to serve as a basis for multi-layer uncertainty theory.

\fi 

\if \MkLonger  1
\tableofcontents
\fi 

\section{Introduction}
\label{sec:intro}

The experience of unknown events such as financial crises and infectious disease crises, 
which are black swans that appear suddenly, has revealed the limitations of measuring risk under a fixed probability measure. 
In order to solve this problem, the importance of so-called ambiguity, which allows the probability measure itself to change, 
has long been recognized in the financial world. 
On the other hand, 
there have been many studies on subjective probability measures in the field of economics, 
including Savage's work 
(\cite{savage1954}, \cite{izhakian_2017}, \cite{izhakian_2020}), 
but even in those cases, 
the studies are based on the two levels of uncertainty: 
risk when a conventional probability measure (probability distribution) is known, 
and ambiguity due to the fact that the subjective probability measure (prior) can be taken arbitrarily in a certain space.

In this study, we express $n$-layer uncertainty, 
which we call \textit{hierarchical uncertainty} 
by introducing a new concept called \textit{uncertainty spaces} 
which is an extended concept of probability spaces 
and \textit{U-sequence} which is a hierarchically constructed sequence of uncertainty spaces.  
By defining a value function at each level of uncertainty, 
it becomes possible to characterize preference relations at each level.
The U-sequences are used as examples to illustrate Ellsberg's paradox. 
Specifically, three U-sequences will be prepared, 
and each U-sequence will be used to check whether the paradox can be explained with it.
In particular, 
the third example introduces a third level of uncertainty in addition to risk and ambiguity, 
and calculates a value function at that level.
Although this is a simple example, 
it may be helpful in considering how to deal with uncertainties that mankind has not yet faced, 
for example, those brought about by the development of artificial intelligence.

We will use the category theory to get a bird's eye view of the hierarchical structure of uncertainty.
%
%
We discuss maps between uncertainty spaces and maps between U-sequences,
seeing that they form categories of uncertainty spaces and the category of U-sequences, respectively.
\if \InclCMFun 1 
We introduce the lift-up functor 
$\mathfrak{L}$
using one of the categories of uncertainty spaces, $\mpUnc$ 
whose arrows are measure preserving maps in a sense.
$\mathfrak{L}$
is used to define CM-functors that will be considered as envelopes of U-sequences.
By the iterative application of the CM-functor,
we construct a universal uncertainty space
that has a potential to be the basis of multi-level uncertainty theory 
because it has the uncertainty spaces of all levels as its projections.
\fi 

\lspace

The structure of this article is described below.
In Section \ref{sec:choquet},
we review the Choquet capacity and the Choquet integral.
After introducing the new concept of uncertainty spaces in 
Section \ref{sec:uncertainSp},
we introduce U-sequence, a sequence of uncertainty space in 
Section \ref{sec:hieraUncertain},
and give three different U-sequences as examples to illustrate Ellsberg's paradox. 
In particular, the third U-sequence represents another level of uncertainty after the conventional risk and ambiguity.
In sections \ref{sec:catUnc} and \ref{sec:uCat}, 
we organize the concepts discussed so far using category theory.
The key to recognize concept in category theoretical framework is how to define the arrows between objects.
In Section \ref{sec:catUnc}, 
we define two types of arrows between uncertainty spaces, 
one based on absolute continuity and the other based on measure preservation, 
seeing that they form categories of uncertainty spaces $\Unc$ and $\mpUnc$. 
Then, we discuss the relationship between these two types of arrows and the structure of U-sequences.
We also discuss
their possible embedding in higher-order uncertainty spaces.
In Section \ref{sec:uCat}, 
we define the arrows between U-sequences with the help of the concept of entropic risk measure
and also use them to introduce the category of U-sequences,
$\USeq^G$.
We briefly look at 
how they are treated among U-sequences as seen in the interpretation of Ellsberg's paradox earlier.
\if \InclCMFun 1 
One of the categories of uncertainty spaces, $\mpUnc$ whose arrows are measure preserving maps in a sense,
is used for introducing the lift-up functor
$
\mathfrak{L}
	:
\mpUnc
	\to
\Mble
$
in Section \ref{sec:functorS}.
We use 
$
\mathfrak{L}
$
for defining a endofunctor of 
$\Mble$
called a CM-functor
in order to embed a given set of uncertainty spaces into it.
After developing n-layer uncertainty analysis through uncertainty spaces,
we construct the universal uncertainty space
in Section \ref{sec:uniUncSp}
as a limit of the sequence generated by iterative applications of a CM-functor.
This universal uncertainty space may be able to serve as a basis for multi-layer uncertainty theory
because it has as its projections to the uncertainty spaces of all levels.
In Section \ref{sec:GiryMonad},
we confirm a sufficient condition for the functor
$\mathfrak{S}$
to be a probability monad developed by
\cite{lawvere_1962}
and
\cite{giry_1982}
though the resulting condition is not new.
\fi 
Finally, in 
Section \ref{sec:concl}, 
we discuss the significance of considering hierarchies of uncertainty in modern society.

\if \MkLonger  1
\section{Review of Choquet integrals}
\else 
\section{Preliminaries}
\fi 
\label{sec:choquet}

In
\cite{ellsberg_1961},
Ellsberg claimed a paradox showing that people's subjective probability may not be additive\footnote{
See Section \ref{sec:Ellsberg} for the detail.
}.
Since we are talking about uncertainty derived by subjective probabilities, 
we need to develop our theory based on non-additive probabilities, i.e. capacities.
Choquet extended the usual procedure of calculating expectations to non-additive probabilities, 
which is now called Choquet integrals
(\cite{choquet_1954}). 

In this section,
we review 
the concepts of capacities and Choquet integrals
based on
\cite{schmeidler_1986}.
\if \MkLonger 0
We omit the proofs in this section
because some of them are straightforward and others are described in
\cite{schmeidler_1986}.
\fi 

First, we introduce 
characteristic maps
$
\mathbb{1}_X(A)
$
of subsets.

\begin{defn}
Let
$X$
be a set.
For
$A \subset X$,
$
\mathbb{1}_X(A)
:
	X
\to
	\mathbb{R}
$
is the map defined by
for
$
x \in X
$,
\begin{equation}
\mathbb{1}_X(A)(x)
	:=
\begin{cases}
1
	&\textrm{if}
\; x \in A,
	\\
0
	&\textrm{if}
\; x \in X \setminus A.
\end{cases}
\end{equation}
\end{defn}

\begin{prop}
\label{prop:lemOne}
Let
$
h : X \to Y
$
be a map,
and
$
B \subset Y
$.
Then,
\if \MkLonger  1
\begin{equation}
\mathbb{1}_Y(B)
	\circ
h
	=
\mathbb{1}_X(
	h^{-1}(B)
).
\end{equation}
\else 
$
\mathbb{1}_Y(B)
	\circ
h
	=
\mathbb{1}_X(
	h^{-1}(B)
).
$
\fi 
\end{prop}
\if \MkLonger 1
\begin{proof}
Let
$x \in X$.
Then, we have
\begin{align*}
(\mathbb{1}_Y(B)
	\circ
h)
(x)
		&=
\mathbb{1}_Y(B)
(
h
(x)
)
		\\&=
\begin{cases}
1
	&\textrm{if}
\; h(x) \in B,
	\\
0
	&\textrm{if}
\; h(x) \in Y \setminus B.
\end{cases}
		\\&=
\begin{cases}
1
	&\textrm{if}
\; x \in h^{-1}(B),
	\\
0
	&\textrm{if}
\; x \in X \setminus h^{-1}(B).
\end{cases}
		\\&=
\mathbb{1}_X(
	h^{-1}(B)
)(x).
\end{align*}
\end{proof}
\fi 

Throughout the rest of this section,
$
\mathfrak{X}
 = (X, \Sigma_X)
$
is a measurable space,
and consider
$\mathbb{R}$
as a measurable space equipped with the usual Lebesgue measure.

\begin{prop}
\label{prop:charFunMeas}
Let 
$
f : 
\mathfrak{X}
\to \mathbb{R}
$ 
be a step function
defined by
\if \MkLonger  1
\begin{equation}
f
	:=
\sum_{i=1}^{\infty}
a_i
\mathbb{1}_X(A_i),
\end{equation}
\else 
$
f
	:=
\sum_{i=1}^{\infty}
a_i
\mathbb{1}_X(A_i),
$
\fi 
where
$a_i \in \mathbb{R}$
and
$A_i \in \Sigma_X$
for each $i = 1, 2, \cdots$,
and
$\{ A_i\}_{i=1}^{\infty}$
are mutually disjoint.
Then,
$f$ is measurable.
\end{prop}
\if \MkLonger 1
\begin{proof}
It is obvious since 
$
f^{-1}(U)
	=
\cup
\{
	A_i
\mid
	a_i \in U
\}
	\in
\Sigma_X
$
for every $U  \subset \mathbb{R}$.

\end{proof}
\fi 

\begin{defn}
\label{def:capacity}
A 
\newword{(Choquet) capacity}
on $X$
is a map
$
u
	:
\Sigma_X \to \mathbb{R}
$
satisfying the following conditions.
\begin{enumerate}
\item
$u(\emptyset) = 0$,
\quad
$u(X) = 1$.

\item
$u$
is an increasing map in the sense that
for every pair 
$A, B $
in $\Sigma_X$,
$
A \subset B
$
implies
$
u(A)
	\le
u(B)
$.

\end{enumerate}

\end{defn}
A capacity is sometimes called a
\newword{subjective probability}.
A probability measure on 
$\mathfrak{X}$
is a capacity on 
$\mathfrak{X}$.
But a capacity does not require the additivity like probability measures,
that is, the equation
$
u(A \cup B)
	=
u(A) + u(B)
$
for
$
A, B \in \Sigma_X
$
may not hold even if
$A \cap B = \emptyset$.

\begin{exmp}
\label{exmp:capacity}
Let
$u$
be a capacity on 
$\mathfrak{X}$
and
$
h : [0, 1] \to [0, 1]
$
be a non-decreasing function such that
$
h(0) = 0
$
and
$
h(1) = 1 .
$
Then,
$
h \circ u
$
is a capacity on
$\mathfrak{X}$
as well.
Especially,
For a probability measure
$\mathbb{P}$
on 
$\mathfrak{X}$,
$
h \circ \mathbb{P}
$
is a capacity on 
$\mathfrak{X}$.
\end{exmp}

\begin{defn}
\label{defn:Linfty}
Let
$
L^{\infty}(\mathfrak{X})
$
be the set of all bounded real-valued measurable functions on 
$\mathfrak{X}$.
It is a Banach space with the norm\footnote{
Usually, the
$L^{\infty}$
norm is defined by
$
\esssup
$
in order to exclude null-sets.
However, in our usage the null-sets are not pre-defined since
probability measure or capacity changes all the time.
So, we simply use $\sup$.
}
\begin{equation}
\label{eq:LinftyNorm}
\norm{f}_{\infty}
	:=
\sup\{ \abs{f(x)}
	\mid
x \in X
\} 
\end{equation}
for
$f \in L^{\infty}(\mathfrak{X})$.
Then, we can think
$L^{\infty}$
as a contravariant functor
\begin{equation*}
\xymatrix{
	\Mble
		\ar @{->} [r]^{L^{\infty}}
&
	\Ban
\\
	\mathfrak{X}
		\ar @{->} [d]_{h}
&
	L^{\infty}(
		\mathfrak{X}
	)
		\ar @{}_-{\mathrel{\rotatebox[origin=c]{-180}{$\in$} }} @<+6pt> [r]
&
	L^{\infty}(h)(f) :=
	f \circ h
\\
	\mathfrak{Y}
&
	L^{\infty}(
		\mathfrak{Y}
	)
		\ar @{->} [u]_{L^{\infty}(h)}
		\ar @{}_-{\mathrel{\rotatebox[origin=c]{-180}{$\in$} }} @<+6pt> [r]
&
	f
		\ar @{|->} [u]
}
\end{equation*}
where
$\Ban$ is the category of Banach spaces whose arrows are bounded linear maps.

\end{defn}

\begin{defn}
\label{defn:ChoquetInt}
Let
$u$
be
a capacity 
 on 
$\mathfrak{X}$
and 
$f \in
L^{\infty}(
	\mathfrak{X}
)
$.
\begin{enumerate}
\item
A function
$
f^{u}
	:
\mathbb{R} \to [0,1]
$
is
defined by
for $r \in \mathbb{R}$,
\begin{equation}
\label{eq:distUnu}
f^{u}(r)
	:=
\begin{cases}
u(\{f \ge r\})
	&\textrm{if}\;
r \ge 0,
	\\
u(\{f \ge r\}) - 1
	&\textrm{if}\;
r < 0
\end{cases}
\end{equation}
where
$
\{f \ge r\}
$
denotes the subset of $X$
specified by
$
\{
x \in X
	\mid
f(x) \ge r
\}
$,
or
$
f^{-1}([r, \infty))
$.

\item
The 
\newword{Choquet integral}
of
$f$
with respect to 
$u$
is 
the real value
$
I_{
	\mathfrak{X}
}^{u}
(f)
$
calculated by
\begin{equation}
\label{eq:choquetIntegral}
I_{
	\mathfrak{X}
}^{u}
(f)
	:=
\int_{-\infty}^{\infty}
f^{u}(r)
\, dr .
\end{equation}

\item
Two functions
$f$
and
$g$
in
$
L^{\infty}(
	\mathfrak{X}
)
$
are said to be 
\newword{comonotonic}
if
for all
$x$
and
$y$
in
$X$,
\begin{equation}
\label{eq:comonotonic}
(f(x) - f(y))
(g(x) - g(y))
	\ge 0 .
\end{equation}

\end{enumerate}
\end{defn}

\begin{lem}
\label{lem:stepFunChoquetInt}
Let 
$
f : 
\mathfrak{X}
 \to \mathbb{R}
$ 
be a finite step function
that can be written as
\if \WithProof 0
$
f = \sum_{i=1}^n
a_i
\mathbb{1}_X(A_i),
$
\else 
\begin{equation}
\label{eq:stepFunOne}
f = \sum_{i=1}^n
a_i
\mathbb{1}_X(A_i),
\end{equation}
\fi 
where
$A_1, \cdots, A_n$
are mutually disjoint members of
$\Sigma_X$,
and
$a_i$'s
are real numbers satisfying
$
a_1 \ge a_2 \ge \cdots \ge  a_n
$.
Then, for a capacity
$u$ on 
$
\mathfrak{X}
$,
we have
\if \WithProof 0
$
I_{
	\mathfrak{X}
}^{u}(f)
	=
\sum_{i=1}^n
	(a_i - a_{i+1})
	\,
	u\big(
		\bigcup_{j=1}^i
		A_j
	\big),
$
\else 
\begin{equation}
\label{eq:stepFunChoquetInt}
I_{
	\mathfrak{X}
}^{u}(f)
	=
\sum_{i=1}^n
	(a_i - a_{i+1})
	\,
	u\big(
		\bigcup_{j=1}^i
		A_j
	\big),
\end{equation}
\fi 
where
$
a_{n+1} := 0
$.
\end{lem}
\if \MkLonger 1
\begin{proof}
%
First, note that for 
$r \in \mathbb{R}$
\begin{equation}
\{ f \ge r\}
	= \begin{cases}
	\emptyset
		&\textrm{if} \;
		r > a_1,
\\
	\bigcup_{j=1}^i A_j
		&\textrm{if} \;
	a_i \ge r > a_{i+1} \quad(i = 1, 2, \cdots, n-1),
\\
	X
		&\textrm{if} \;
	r \le a_n .
\end{cases}
\end{equation}
Now let
$k \in \{0, 1, \cdots, n\}$
be a number such that
$
a_k \ge 0 \ge a_{k+1}
$.
Then, we have
\begin{align*}
I_{
	\mathfrak{X}
}^{u}(f)
	=&
\int_0^{\infty}
	u(f \ge r)
\, dr
	+
\int_{-\infty}^0
	\big( u(f \ge r) - 1 \big) 
\, dr
	\\=&
\sum_{i=1}^{k-1}
	(a_{i} - a_{i+1}) \,
	u\big( \bigcup_{j=1}^{i} A_j \big)
	+
(a_k - 0) \,
u\big( \bigcup_{j=1}^k A_j \big)
	\\&+
(0 - a_{k+1}) \,
\Big(
u\big( \bigcup_{j=1}^k A_j \big)
	- 1
\Big)
+
\sum_{i=k+1}^n
	(a_{i} - a_{i+1}) \,
\Big(
	u\big( \bigcup_{j=1}^{i} A_j \big)
	- 1
\Big)
	\\=&
\sum_{i=1}^n
	(a_{i} - a_{i+1}) \,
	u\big( \bigcup_{j=1}^{i} A_j \big)
+
	a_{k+1}
-
\sum_{i=k+1}^n
	(a_{i} - a_{i+1}) 
	\\=&
\sum_{i=1}^n
	(a_{i} - a_{i+1}) \,
	u\big( \bigcup_{j=1}^{i} A_j \big) .
\end{align*}
\end{proof}
\fi 

Note that
as a special case of Lemma \ref{lem:stepFunChoquetInt},
we have
\if \WithProof 0
$
I_{
	\mathfrak{X}
}^u(
	\mathbb{1}_X(A)
)
		=
u(A) .
$
\else 
\begin{equation}
\label{eq:ChoquetOneA}
I_{
	\mathfrak{X}
}^u(
	\mathbb{1}_X(A)
)
		=
u(A) .
\end{equation}
\fi 

\begin{lem}
\label{lem:stepFunComonotonic}
Let
$f$
and
$g$
be two finite measurable step functions on
	$\mathfrak{X}$.
Then,
$f$ and $g$ are comonotonic
if and only if
there exist disjoint members
$
A_1, \cdots, A_n 
	\in
\Sigma_X
$
and
real numbers
$
a_1, \cdots, a_n
$
and
$
b_1, \cdots, b_n
$
such that
$
a_1 \ge a_2 \ge \cdots \ge a_n
$,
$
b_1 \ge b_2 \ge \cdots \ge b_n
$
and
\begin{equation}
\label{eq:fgRep}
f = 
\sum_{i=1}^n
	a_i
	\mathbb{1}_X(A_i) ,
		\quad
g = 
\sum_{i=1}^n
	b_i
	\mathbb{1}_X(A_i) .
\end{equation}
\end{lem}
\if \MkLonger 1
\begin{proof}
Suppose that
$f$ and $g$
are represented as
(\ref{eq:fgRep}).
Then, the domains of 
$f$ and $g$
are subsets of
$
\bigcup_{i=1}^n A_i
$.
Now, for
$x \in 
\bigcup_{i=1}^n A_i
$,
let
$i(x)$
be the unique number such that
$x \in A_{i(x)}$.
Then,
for
$x , y \in X$
\begin{equation}
(f(x) - f(y))
(g(x) - g(y))
	=
(a_{i(x)} - a_{i(y)})
(b_{i(x)} - b_{i(y)}) .
\end{equation}
On the other hand,
$i(x)$
and
$i(y)$
vary
$1$ through $n$
depending on $x$ and $y$.
Therefore, 
$f$
and
$g$
are comonotonic
iff
for every pair
$i, j \in \{1, \cdots, n\}$,
$a_i - a_j$
and
$b_i - b_j$
have same sign,
which proves the lemma.

\end{proof}
\fi 

\begin{cor}
\label{cor:ChoquetComono}
Let
$u$ be a capacity on 
$\mathfrak{X}$,
and
$c \in \mathbb{R}$ be a constant.
\begin{enumerate}
\item
$
I_{
	\mathfrak{X}
}^{u}(a \mathbb{1}_X(A))
	=
a \, u(A) 
$
for
$a \in \mathbb{R}$
and
$A \in \Sigma_X$.

\item
$
I_{
	\mathfrak{X}
}^{u}
(c)
	=
c .
$

\item
A constant function and 
any function in $L^{\infty}(
	\mathfrak{X}
)$
are comonotonic.

\item
For a function
$
f \in L^{\infty}(
	\mathfrak{X}
)
$
and
non-negative numbers
$
a, b \in \mathbb{R}
$,
$af$
and
$bf$
are comonotonic.

\end{enumerate}
\end{cor}
\if \MkLonger 1
\begin{proof}
\begin{enumerate}
\item
Immediate from Lemma \ref{lem:stepFunChoquetInt}.

\item
By (1),
$
I_X^{u}
(c)
	=
I_X^{u}
(c \, \mathbb{1}_X(X))
	=
c \, u(X)
	=
c .
$

\item
Let
$
f \in L^{\infty}(X)
$
be any bounded function
and
$
g :=
c \, \mathbb{1}_X(X)
$
be a constant function.
Then for every pair $x, y \in X$,
$
(f(x) - f(y))(g(x) - g(x))
	=
(f(x) - f(y))(c - c)
	= 0 \ge 0.
$

\item
For 
$x, y \in X$,
$
(
a\,f(x)
	-
a\,f(y)
)
(
b\,f(x)
	-
b\,f(y)
)
		=
ab
(f(x) - f(y))^2
		\ge
0 .
$

\end{enumerate}
\end{proof}
\fi 

\if \JAFEE 1
Let us introduce a general method of discretizing bounded functions.
\else 
Let us introduce a general method of discretizing bounded functions,
which we will use in several proofs.
\fi 

\begin{Note}
\label{note:discretizeU}
Let
$f$
be an element of
$L^{\infty}(\mathfrak{X})$
with
$
\norm{f}_{\infty}
	<
M
$.
For
$n \in \mathbb{N}$,
define
two finite step functions
$
\underline{f}_n
$
and
$
\bar{f}_n
$
by
\begin{equation}
\underline{f}_n
	:=
\sum_{k=-2^n+1}^{2^n}
	2^{-n} (k-1)M
	\,
	\mathbb{1}_X(A_k) ,
		\quad
\bar{f}_n
	:=
\sum_{k=-2^n+1}^{2^n}
	2^{-n} kM
	\,
	\mathbb{1}_X(A_k),
\end{equation}
where 
$
A_k
	:=
\{
	x \in X
\mid
2^{-n} (k-1)M
	<
f(x)
	\le
2^{-n} kM 
\} .
$
Then,
for $x \in X$
\begin{equation}
\label{eq:discUlimit}
\bar{f}_{n}(x)
	-
\underline{f}_{n}(x)
	=
2^{-n}M
	\to
0
	\quad
(n \to \infty) .
\end{equation}
Actually, for every $n \in \mathbb{N}$, we have
\if \WithProof 0
$
\underline{f}_n(x)
	\le
\underline{f}_{n+1}(x)
	\le
f(x)
	\le
\bar{f}_{n+1}(x)
	\le
\bar{f}_{n}(x) .
$
\else 
\begin{equation}
\label{eq:discUmono}
\underline{f}_n(x)
	\le
\underline{f}_{n+1}(x)
	\le
f(x)
	\le
\bar{f}_{n+1}(x)
	\le
\bar{f}_{n}(x) .
\end{equation}
\fi 
Moreover, for every pair of $x, y \in X$,
$
f(x) \le f(y)
$
implies
\if \WithProof 0
$
\underline{f}_n(x)
	\le
\underline{f}_n(y)
	\quad \textrm{and} \quad
\bar{f}_n(x)
	\le
\bar{f}_n(y) .
$
\else 
\begin{equation}
\label{eq:discUcomono}
\underline{f}_n(x)
	\le
\underline{f}_n(y)
	\quad \textrm{and} \quad
\bar{f}_n(x)
	\le
\bar{f}_n(y) .
\end{equation}
\fi 

\end{Note}

\begin{thm}
\label{thm:choquetThm}
Let
$u$ 
be a capacity on 
$\mathfrak{X}$
and 
$f, g \in
L^{\infty}(
	\mathfrak{X}
)
$.
\begin{enumerate}
\item
$
I_{
	\mathfrak{X}
}^{u}
$
is 
\newword{monotonic}, i.e.,
$f \ge g$ 
implies
$
I_{
	\mathfrak{X}
}^{u}
(f)
	\ge
I_{
	\mathfrak{X}
}^{u}
(g)
$.

\item
$
I_{
	\mathfrak{X}
}^{u}
$
is
\newword{comonotonically additive}, i.e.,
$
I_{
	\mathfrak{X}
}^{u}(f+g)
	=
I_{
	\mathfrak{X}
}^{u}(f)
	+
I_{
	\mathfrak{X}
}^{u}(g)
$
if
$f$
and
$g$
are comonotonic.

\item
$
I_{
	\mathfrak{X}
}^{u}
$
is 
\newword{positively homogeneous},
i.e.,
$
I_{
	\mathfrak{X}
}^{u}
(\lambda f)
	=
\lambda
I_{
	\mathfrak{X}
}^{u}
(f)
$
for every
$\lambda > 0$.

\end{enumerate}
\end{thm}
\if \MkLonger  1
\begin{proof}
\begin{enumerate}
\item
By 
(\ref{eq:distUnu}),
$f \ge g$ on
$X$
implies
$
f^{u}(r)
	\ge
g^{u}(r)
$
for every
$r \in \mathbb{R}$.
Then, 
we have
$
I_{
	\mathfrak{X}
}^{u}(f)
	\ge
I_{
	\mathfrak{X}
}^{u}(g)
$
by (\ref{eq:choquetIntegral}).

\item
First, we show the case both $f$ and $g$ are comonotonic finite step functions.
By Lemma \ref{lem:stepFunComonotonic},
we can assume
\begin{equation}
f = 
\sum_{i=1}^n
	a_i
	\mathbb{1}_{
		\mathfrak{X}
	}(A_i) ,
		\quad
g = 
\sum_{i=1}^n
	b_i
	\mathbb{1}_{
		\mathfrak{X}
	}(A_i) ,
\end{equation}
where
$\{ A_i\}_{i=1}^n$
are mutually disjoint members of
$\Sigma_X$,
$
a_1 \ge a_2 \ge \cdots \ge a_n
$
and
$
b_1 \ge b_2 \ge \cdots \ge b_n
$.
Then,
by Lemma \ref{lem:stepFunChoquetInt},
we have
\begin{align*}
I_{
	\mathfrak{X}
}^{u}(f)
	+
I_{
	\mathfrak{X}
}^{u}(g)
		&=
\sum_{i=1}^n
	(a_i - a_{i+1})
	\,
	u\big(
		\bigcup_{j=1}^i
		A_j
	\big)
		+
\sum_{i=1}^n
	(b_i - b_{i+1})
	\,
	u\big(
		\bigcup_{j=1}^i
		A_j
	\big)
		\\&=
\sum_{i=1}^n
	((a_i + b_i) - (a_{i+1} + b_{i+1}))
	\,
	u\big(
		\bigcup_{j=1}^i
		A_j
	\big)
		\\&=
I_{
	\mathfrak{X}
}^{u}(f+g) .
\end{align*}

Now let's proceed to the case when
$f$ and $g$ are arbitrary comonotonic pair in
$L^{\infty}(
	\mathfrak{X}
)$.
Assume that
both are bounded like 
$
\norm{f}_{\infty} < M
$
and 
$
\norm{g}_{\infty} < M'
$.

Using the technique described in Note \ref{note:discretizeU},
we construct step functions
$\underline{f}_n$,
$\bar{f}_n$,
$\underline{g}_n$
and
$\bar{g}_n$
for $n \in \mathbb{N}$.
Then,
by the monotonicity of 
$I_{
	\mathfrak{X}
}^{u}$
and
(\ref{eq:discUmono})
implies
\begin{equation}
\label{eq:IeqMono}
I_{
	\mathfrak{X}
}^{u}
(\underline{f}_n)
	\le
I_{
	\mathfrak{X}
}^{u}
(f)
	\le
I_{
	\mathfrak{X}
}^{u}
(\bar{f}_n)
	\quad \textrm{and} \quad
I_{
	\mathfrak{X}
}^{u}
(\underline{g}_n)
	\le
I_{
	\mathfrak{X}
}^{u}
(g)
	\le
I_{
	\mathfrak{X}
}^{u}
(\bar{g}_n) .
\end{equation}

On the other hand,
since
$f$
and
$g$
are comonotonic,
by (\ref{eq:discUcomono})
shows that pairs
$\underline{f}_n$
and
$\underline{g}_n$,
and also
$\bar{f}_n$
and
$\bar{g}_n$
are comonotonic, respectively.
Therefore, we obtain
\begin{equation}
\label{eq:IeqIpI}
I_{
	\mathfrak{X}
}^{u}
(
\underline{f}_n
	+
\underline{g}_n
)
	=
I_{
	\mathfrak{X}
}^{u}
(\underline{f}_n)
	+
I_{
	\mathfrak{X}
}^{u}
(\underline{g}_n)
	\quad \textrm{and} \quad
I_{
	\mathfrak{X}
}^{u}
(
\bar{f}_n
	+
\bar{g}_n
)
	=
I_{
	\mathfrak{X}
}^{u}
(\bar{f}_n)
	+
I_{
	\mathfrak{X}
}^{u}
(\bar{g}_n) .
\end{equation}
By
(\ref{eq:IeqMono})
and
(\ref{eq:IeqIpI}),
we have
\begin{equation}
\label{eq:IIlimone}
I_{
	\mathfrak{X}
}^{u}
(
\underline{f}_n
	+
\underline{g}_n
)
	\le
I_{
	\mathfrak{X}
}^{u}(f)
	+
I_{
	\mathfrak{X}
}^{u}(g)
	\le
I_{
	\mathfrak{X}
}^{u}
(
\bar{f}_n
	+
\bar{g}_n
)
\end{equation}
Moreover, 
by
(\ref{eq:IeqMono})
and the monotonicity of
$
I_{
	\mathfrak{X}
}^{u}
$, we also have
\begin{equation}
\label{eq:IIlimtwo}
I_{
	\mathfrak{X}
}^{u}
(
\underline{f}_n
	+
\underline{g}_n
)
	\le
I_{
	\mathfrak{X}
}^{u}(f + g)
	\le
I_{
	\mathfrak{X}
}^{u}
(
\bar{f}_n
	+
\bar{g}_n
)
\end{equation}
Considering
\begin{align*}
	&
I_{
	\mathfrak{X}
}^{u}
(
\bar{f}_n
	+
\bar{g}_n
)
		-
I_{
	\mathfrak{X}
}^{u}
(
\underline{f}_n
	+
\underline{g}_n
)
		\\&=
(
I_{
	\mathfrak{X}
}^{u}
(
\bar{f}_n
)
	-
I_{
	\mathfrak{X}
}^{u}
(
\underline{f}_n
)
)
	+
(
I_{
	\mathfrak{X}
}^{u}
(
\bar{g}_n
)
	-
I_{
	\mathfrak{X}
}^{u}
(
\underline{g}_n
)
)
		& (\textrm{by (\ref{eq:IeqIpI})}) 
		\\&=
(
I_{
	\mathfrak{X}
}^{u}
(
\underline{f}_n
	+
2^{-n} M
)
	-
I_{
	\mathfrak{X}
}^{u}
(
\underline{f}_n
)
)
	+
(
I_{
	\mathfrak{X}
}^{u}
(
\underline{g}_n
	+
2^{-n} M'
)
	-
I_{
	\mathfrak{X}
}^{u}
(
\underline{g}_n
)
)
		& (\textrm{by (\ref{eq:discUlimit})}) 
		\\&=
2^{-n}(M + M')
\mathbb{1}_X(X)
		& (\textrm{by Corollary \ref{cor:ChoquetComono} (3)}) 
		\\& \to
0
\cdot
\mathbb{1}_X(X)
	\quad
(n \to \infty)
		& (\textrm{pointwise limit}) 
\end{align*}
and
taking limits of
(\ref{eq:IIlimone})
and
(\ref{eq:IIlimtwo}),
we get
$
I_{
	\mathfrak{X}
}^{u}(f)
	+
I_{
	\mathfrak{X}
}^{u}(g)
	=
I_{
	\mathfrak{X}
}^{u}(f + g).
$

\item
For a positive integer
$n$,
\begin{align*}
I_{
	\mathfrak{X}
}^u(f)
		&=
I_{
	\mathfrak{X}
}^u
\Big(
\frac{1}{n}f
	+
\frac{n-1}{n}f
\Big)
		\\&=
I_{
	\mathfrak{X}
}^u
\Big(
	\frac{1}{n}f
\Big)
	+
I_{\mathfrak{X}}^u
\Big(
	\frac{n-1}{n}f
\Big)
		&(\textrm{by 
			(2) and
			Corollary \ref{cor:ChoquetComono} (4)
		})
		\\&=
I_{\mathfrak{X}}^u
\Big(
	\frac{1}{n}f
\Big)
	+
I_{\mathfrak{X}}^u
\Big(
	\frac{1}{n}f
		+
	\frac{n-2}{n}f
\Big)
		\\&=
2
I_{\mathfrak{X}}^u
\Big(
	\frac{1}{n}f
\Big)
	+
I_{\mathfrak{X}}^u
\Big(
	\frac{n-2}{n}f
\Big)
		&(\textrm{by 
			(2) and
			Corollary \ref{cor:ChoquetComono} (4)
		})
		\\&=
\cdots
		\\&=
n
I_{\mathfrak{X}}^u
\Big(
	\frac{1}{n}f
\Big) .
\end{align*}
Hence,
\begin{equation}
\label{eq:comoPropPHone}
I_{\mathfrak{X}}^u
\Big(
	\frac{1}{n}f
\Big)
		=
\frac{1}{n}
I_{\mathfrak{X}}^u
(f) .
\end{equation}
Similarly,
for a non-negative integer
$m$,
\begin{align*}
I_{\mathfrak{X}}^u
(m f)
		&=
I_{\mathfrak{X}}^u
(
	f + (m-1) f
)
		\\&=
I_{\mathfrak{X}}^u
(f)
	+
I_{\mathfrak{X}}^u
(
	(m-1)f
)
		&(\textrm{by 
			(2) and
			Corollary \ref{cor:ChoquetComono} (4)
		})
		\\&=
\cdots
		\\&=
m
I_{\mathfrak{X}}^u
(f) .
\end{align*}
Therefore,
for any positive rational number $a$, 
\begin{equation}
\label{eq:posHomo}
I_{\mathfrak{X}}^u
(af)
	=
a 
I_{\mathfrak{X}}^u
(f) .
\end{equation}
Then, the monotonicity which is also a continuity implies, 
in turn, 
the equality 
(\ref{eq:posHomo})
for an arbitrary non-negative number 
$a \in \mathbb{R}$.

\end{enumerate}
\end{proof}
\fi 

\begin{prop}
\label{prop:additiveNu}
Let
$u$ 
be a capacity on $
	\mathfrak{X}
$.
If $u$ is additive, i.e.,
$
u(A \cup B)
	=
u(A)
	+
u(B)
$
for any disjoint pair 
$A$
and
$B$
in
$
\Sigma_X
$.
Then,
for
$f \in L^{\infty}(
	\mathfrak{X}
)$,
\if \WithProof 0
$
I_{\mathfrak{X}}^u
(f)
	=
\int_X f \, d u,
$
\else 
\begin{equation}
\label{eq:additiveNu}
I_{\mathfrak{X}}^u
(f)
	=
\int_X f \, d u,
\end{equation}
\fi 
where the right hand side is the usual Lebesgue integral.

Especially, $
I_{\mathfrak{X}}^u
$ is linear, that is,
for every
$a, b \in \mathbb{R}$
and
$
f, g \in L^{\infty}(
	\mathfrak{X}
)
$,
\begin{equation}
\label{eq:IXuLinear}
I_{\mathfrak{X}}^u
(a f + b g)
		=
a 
I_{\mathfrak{X}}^u
(f)
	+
b 
I_{\mathfrak{X}}^u
(g) .
\end{equation}
\end{prop}
\if \MkLonger  1
\begin{proof}
For a step function 
$f$
defined in 
(\ref{eq:stepFunOne}),
we have
\begin{align*}
I_{\mathfrak{X}}^u
(f)
			&=
\sum_{i=1}^n
	(a_i - a_{i+1})
	\,
	u\big(
		\bigcup_{j=1}^i
		A_j
	\big)
			& (\textrm{by (\ref{eq:stepFunChoquetInt})}) 
			\\&=
\sum_{i=1}^n
	(a_i - a_{i+1})
	\sum_{j=1}^i
		u(A_j)
			& (\textrm{by additivity}) 
			\\&=
\sum_{j=1}^n
	\sum_{i=j}^n
		(a_i - a_{i+1})
	\,
	u(A_j)
			\\&=
\sum_{j=1}^n
	a_j
	\,
	u(A_j) .
\end{align*}
Applying this result to 
(\ref{eq:discUmono})
and the monotonicity of $
I_{\mathfrak{X}}^u
$,
we obtain
for general
$f \in L^{\infty}(
	\mathfrak{X}
)$,
\begin{equation}
\label{eq:monoLebesgue}
\sum_{k=-2^n+1}^{2^n}
	2^{-n} (k-1)M
	\,
	u(A_k)
		\le
I_{\mathfrak{X}}^u
(f)
		\le
\sum_{k=-2^n+1}^{2^n}
	2^{-n} kM
	\,
	u(A_k) .
\end{equation}
Both sides of
(\ref{eq:monoLebesgue})
converge to
$
\int_X f \, du
$
as
$n \to \infty$
by the definition of Lebesgue integral,
which completes the proof.

\end{proof}
\fi 

By Proposition \ref{prop:additiveNu},
we write
$
\int_X f \, d u
	:=
I_{\mathfrak{X}}^u
(f)
$
even when 
$u$
is not additive.

\section{Uncertainty spaces}
\label{sec:uncertainSp}

In this section,
we introduce 
a concept of uncertainty space
as an augmented probability space.
We also introduce a Choquet expectation maps
which can be seen as maps between random variables in different uncertainty levels.

\subsection{Uncertainty space $\mathfrak{X}$}
\label{sec:uncSp}

\begin{defn}
\label{defn:uncertaintySp}
Let
$(X, \Sigma_X)$
be a measurable space.
\begin{enumerate}
\item
$
\mathbb{I}
	=
(\mathbb{I}, \Sigma_{\mathbb{I}})
	:=
([0,1], \mathcal{B}([0,1]) .
$

\item
An
\newword{uncertainty space}
is a triple
$
\mathfrak{X}
	=
(X, \Sigma_X, U_X)
$,
where
$U_X$ is a non-empty set of capacities on the measurable space
$(X, \Sigma_X)$.

From now on,
$
\mathfrak{X}
	=
(X, \Sigma_X, U_X)
$
is an uncertainty space.

\item
For 
$A \in \Sigma_X$,
$
\varepsilon_{\mathfrak{X}}(A)
	:
U_X
	\to
\mathbb{I}
$
is the map
defined
\footnote{
By using the lambda calculus notation, we can write
$
\varepsilon_{\mathfrak{X}}
	:=
\lambda A \in \Sigma_X
	\,.\,
\lambda u \in U_X
	\,.\,
u(A)
$.
}
by
for
$
u
\in
U_X
$,
\if \WithProof 0
$
\varepsilon_{\mathfrak{X}}(A)
(u)
		:=
u(A) .
$
\else 
\begin{equation}
\varepsilon_{\mathfrak{X}}(A)
(u)
		:=
u(A) .
\end{equation}
\fi 

\item
$
\Sigma_{
	\mathfrak{X}
}
$
is the smallest
$\sigma$-algebra
on
$U_X$
that makes
$
\varepsilon_{\mathfrak{X}}(A)
$
a measurable map to
$
(
\mathbb{I}, 
\Sigma_{\mathbb{I}}
)
$
for all
$A \in \Sigma_{\mathfrak{X}}$.

\end{enumerate}
\end{defn}


\begin{prop}
\label{prop:sigmaGX}
For
an uncertainty space
$
\mathfrak{X}
	=
(X, \Sigma_X, U_X)
$,
\begin{equation}
\label{eq:sigmaGX}
\Sigma_{
	\mathfrak{X}
}
		=
\sigma \big\{
	\{
		u
			\in
		U_X
	\mid
		u(A)
			\in
		B
	\}
\mid
	A \in \Sigma_X,
	\,
	B \in \Sigma_{\mathbb{I}}
\big\} .
\end{equation}
\end{prop}
\if \WithProof 1
\begin{proof}
Since
$
\Sigma_{\mathfrak{X}}
	:= \sigma\{
	(\varepsilon_{\mathfrak{X}}(A))^{-1}(B)
\mid
	A \in \Sigma_X, \;
	B \in
	\Sigma_{\mathbb{I}}
\},
$
(\ref{eq:sigmaGX})
comes from
the fact that
\begin{equation}
\label{eq:epsXAinvU}
(\varepsilon_{\mathfrak{X}}(A))^{-1}
(B)
		=
\{
	u \in U_X
\mid
	\varepsilon_{\mathfrak{X}}(A)
	(u)
		\in
	B
\}
		=
\{
	u \in U_X
\mid
	u
	(A)
		\in
	B
\} .
\end{equation}
\end{proof}
\fi 

\subsection{Choquet expectation map $\xi_{\mathfrak{X}}$}
\label{sec:xiX}

Throughout this section,
$
\mathfrak{X}
	=
(X, \Sigma_X, U_X)
$
is an uncertainty space.

\begin{defn}
\label{defn:xiXf}
Let
$
f 
	\in L^{\infty}(X) .
$
Define a map
$
\xi_{\mathfrak{X}}(f)
	:
U_X
	\to
\mathbb{R}
$
by
for
$
u
	\in
U_X,
$
\begin{equation}
\label{eq:defXiX}
\xi_{\mathfrak{X}}(f)
(u)
	:=
I_X^{u}
(f) .
\end{equation}
\end{defn}

\begin{prop}
\label{prop:xiXstepF}
Let
$
f := \sum_{i=1}^n
a_i
\mathbb{1}_X(A_i)
$
be a finite step function,
where
$\{ A_i\}_{i=1}^n$
are mutually disjoint members of
$\Sigma_X$,
and
$
a_1 \ge a_2 \ge \cdots \ge  a_n
$
be a sequence of real numbers.
Then,
we have
\begin{equation}
\label{eq:xiXstepF}
\xi_{\mathfrak{X}}(f)
		=
\sum_{i=1}^n
	(a_i - a_{i+1})
	\,
	\varepsilon_{\mathfrak{X}}
	\Big(
		\bigcup_{j=1}^i
			A_j
	\Big) .
\end{equation}
Especially,
\begin{equation}
\label{eq:etaOneEpsilon}
\xi_{\mathfrak{X}}
	\circ
\mathbb{1}_X
		=
\varepsilon_{\mathfrak{X}} .
\end{equation}
\end{prop}
\if \WithProof 1
\begin{proof}
For
$u \in U_X$,
we have 
\if \MkLonger  1
\begin{align*}
\xi_{\mathfrak{X}}(f)
(u)
		&=
I_X^u
(f)
		\\&=
\sum_{i=1}^n
	(a_i - a_{i+1})
	\,
	u
	\Big(
		\bigcup_{j=1}^i
			A_j
	\Big)
		&(\textrm{by Lemma \ref{lem:stepFunChoquetInt}})
		\\&=
\sum_{i=1}^n
	(a_i - a_{i+1})
	\,
	\varepsilon_{\mathfrak{X}}
	\Big(
		\bigcup_{j=1}^i
			A_j
	\Big)
	(u) .
\end{align*}
\else 
by Lemma \ref{lem:stepFunChoquetInt},
\[
\xi_{\mathfrak{X}}(f)
(u)
		=
I_X^u
(f)
		=
\sum_{i=1}^n
	(a_i - a_{i+1})
	\,
	u
	\Big(
		\bigcup_{j=1}^i
			A_j
	\Big)
		=
\sum_{i=1}^n
	(a_i - a_{i+1})
	\,
	\varepsilon_{\mathfrak{X}}
	\Big(
		\bigcup_{j=1}^i
			A_j
	\Big)
	(u) .
\]
\fi 
Therefore, we obtain
(\ref{prop:xiXstepF}).
\end{proof}
\fi 

Note that in
(\ref{eq:etaOneEpsilon}),
we treat
$
\mathbb{1}_X(A)
$
and
$
\varepsilon_{\mathfrak{X}}(A)
$
as maps like
$
\mathbb{1}_X
	:
\Sigma_X \to \mathbb{R}
$
and
$
\varepsilon_{\mathfrak{X}}
	:
\Sigma_X
	\to
(
U_X
	\to
\mathbb{R}
)
$.

\if \MkLonger  0
The following lemma is well-known.
\fi 

\begin{lem}
\label{lem:monotonicMeasurability}
Let 
$f : X \to \mathbb{R}$
be a function,
$
\underline{f}_n,
\bar{f}_n
	\in
L^{\infty}(X)
$
be measurable functions
such that
$
\underline{f}_n(x)
	\le
\underline{f}_{n+1}(x)
	\le
f(x)
	\le
\bar{f}_{n+1}(x)
	\le
\bar{f}_n(x)
$
for every $n \in \mathbb{N}$
and
$x \in X$,
and
$
\bar{f}_{n}(x)
	-
\underline{f}_n(x)
	\to
0
$
as
$
n
	\to
\infty
$.
Then,
$
f
	\in L^{\infty}(X) .
$
\end{lem}
\if \MkLonger  1
\begin{proof}
For every
$x \in X$,
since
$
\{ 
	\bar{f}_n(x)
\}_{n \ge 1}
$
is a decreasing sequence
and
$
\lim_{n \to \infty}
\bar{f}_n(x)
	=
f(x)
$,
we have 
\[
f(x)
		=
\lim_{n \to \infty}
\bar{f}_n(x)
		=
\inf_{n \ge 1}
\bar{f}_n(x).
\]
Therefore, 
for any
$c \in \mathbb{R}$,
\begin{align*}
x \in f^{-1}((-\infty, c])
		&\Leftrightarrow
f(x) \le c
		\\&\Leftrightarrow
\inf_{n \ge 1}
\bar{f}_n(x)
	\le c
		\\&\Leftrightarrow
\forall \varepsilon > 0
	\,.\,
\exists n_0 \ge 1
	\,.\,
\forall n_1 \ge n_0
	\,.\,
\bar{f}_{n_1}(x) < f(x) + \varepsilon \le c + \varepsilon
		\\&\Leftrightarrow
\forall \varepsilon > 0
	\,.\,
\exists n_0 \ge 1
	\,.\,
\bar{f}_{n_0}(x) < f(x) + \varepsilon \le c + \varepsilon
		\\&\Leftrightarrow
\forall j \ge 1
	\,.\,
\exists n \ge 1
	\,.\,
\bar{f}_{n}(x) \le c + \frac{1}{j}
		\\&\Leftrightarrow
x \in
\bigcap_{j \ge 1}
\bigcup_{n \ge 1}
	\bar{f}_n^{-1}\big(
		(-\infty, c + \frac{1}{j}]
	\big)
\end{align*}
Hence
\[
f^{-1}((-\infty, c])
		=
\bigcap_{j \ge 1}
\bigcup_{n \ge 1}
	\bar{f}_n^{-1}\big(
		(-\infty, c + \frac{1}{j}]
	\big)
		\in
\Sigma_{X} .
\]
Therefore,
$f$
is measurable.
The boundedness of $f$ is clear.

\end{proof}
\fi 

\begin{prop}
For
a measurable space
$
(X, \Sigma_X)
$
and
$
f \in L^{\infty}(X) ,
$
$
\xi_{\mathfrak{X}}(f)
$
is bounded measurable.
That is,
\if \WithProof 0
$
\xi_{\mathfrak{X}}
	:
L^{\infty}(X)
	\to
L^{\infty}(U_X) ,
$
\else 
\begin{equation}
\label{eq:xiXmap}
\xi_{\mathfrak{X}}
	:
L^{\infty}(X)
	\to
L^{\infty}(U_X) ,
\end{equation}
\fi 
where
$
L^{\infty}(U_X)
	:=
L^{\infty}(U_X, \Sigma_{\mathfrak{X}}) .
$
\end{prop}
\if \WithProof 1
\begin{proof}
First we show
that
$
\xi_{\mathfrak{X}}(f)
$
is measurable
in the case when
$f$ is a finite step function
\begin{equation}
\label{eq:stepFunTwo}
f = \sum_{i=1}^n
a_i
\mathbb{1}_X(A_i),
\end{equation}
where
$\{ A_i\}_{i=1}^n$
are mutually disjoint members of
$\Sigma_X$,
and
$
a_1 \ge a_2 \ge \cdots \ge  a_n
$
be a sequence of real numbers,
by induction on $n$.

When $n=1$, 
$
f = 
a_1
\mathbb{1}_X(A_1) .
$
Then for
$
u
	\in
U_X
$,
\if \MkLonger  1
\begin{align*}
\xi_{\mathfrak{X}}(f)(u)
		&=
\xi_{\mathfrak{X}}(
a_1
\mathbb{1}_X(A_1)
)(u)
		\\&=
I_X^u(
	a_1
	\mathbb{1}_X(A_1)
)
		\\&=
a_1
u(A_1)
		&(\textrm{by Lemma \ref{lem:stepFunChoquetInt}})
		\\&=
a_1
\varepsilon_{\mathfrak{X}}(A_1)(u) .
\end{align*}
\else 
we have, by Lemma \ref{lem:stepFunChoquetInt},
\begin{align*}
\xi_{\mathfrak{X}}(f)(u)
		=
\xi_{\mathfrak{X}}(
a_1
\mathbb{1}_X(A_1)
)(u)
		=
I_X^u(
	a_1
	\mathbb{1}_X(A_1)
)
		=
a_1
u(A_1)
		=
a_1
\varepsilon_{\mathfrak{X}}(A_1)(u) .
\end{align*}
\fi 
Therefore,
$
\label{eq:xiEqEpsilon}
\xi_{\mathfrak{X}}(
a_1
\mathbb{1}_X(A_1)
)
	=
a_1
\varepsilon_{\mathfrak{X}}(A_1),
$
which is a measurable map.

Assume that for every 
$m \le n$,
$
\xi_x\big(
	\sum_{i=1}^m
	a_i
	\mathbb{1}_X(A_i)
\big)
$
is measurable.
We will prove that
$\xi_{\mathfrak{X}}(f)$
is measurable
when
$
f = \sum_{i=1}^{n+1}
a_i
\mathbb{1}_X(A_i) .
$
Now let
\begin{equation}
g := 
\sum_{i=1}^{n}
(a_i - a_{n+1})
\mathbb{1}_X(A_i),
	\quad
h := 
\sum_{i=1}^{n+1}
a_{n+1}
\mathbb{1}_X(A_i)
	=
a_{n+1}
\mathbb{1}_X \big(
	\bigcup_{i=1}^{n+1}
		A_i
\big)
\end{equation}
Then, 
$f = g + h$
and
by
Lemma \ref{lem:stepFunComonotonic},
$g$ and $h$ are comonotonic.
So we have
\if \MkLonger  1
\begin{align*}
\xi_{\mathfrak{X}}(f)(u)
		&=
I_X^u(f)
		\\&=
I_X^u(g+h)
		\\&=
I_X^u(g)
	+
I_X^u(h)
		&(\textrm{by Theorem \ref{thm:choquetThm} (2)})
		\\&=
\xi_{\mathfrak{X}}(g)(u)
	+
\xi_{\mathfrak{X}}(h)(u)
		\\&=
(\xi_{\mathfrak{X}}(g) + \xi_{\mathfrak{X}}(h)) (u) .
\end{align*}
\else 
by
Theorem \ref{thm:choquetThm} (2),
\[
\xi_{\mathfrak{X}}(f)(u)
		=
I_X^u(f)
		=
I_X^u(g+h)
		=
I_X^u(g)
	+
I_X^u(h)
		=
\xi_{\mathfrak{X}}(g)(u)
	+
\xi_{\mathfrak{X}}(h)(u)
		=
(\xi_{\mathfrak{X}}(g) + \xi_{\mathfrak{X}}(h)) (u) .
\]
\fi 
Therefore,
$
\xi_{\mathfrak{X}}(f)
	=
\xi_{\mathfrak{X}}(g) + \xi_{\mathfrak{X}}(h) 
$
in which the first term is an $m=n$ case
and the second term is an $m=1$ case,
and both are measurable by the assumption.
Hence 
$
\xi_{\mathfrak{X}}(f)
$
is measurable.
For general
$f \in L^{\infty}(X)$,
use the approximation technique in
Note \ref{note:discretizeU}
and
Lemma \ref{lem:monotonicMeasurability}.

Next, we will show that
$
\xi_{\mathfrak{X}}(f)
$
is bounded.
Since 
$f$
is bounded, there exists a positive
$M$
such that
$
\norm{f}_{\infty}
	\le
M.
$
Then,
for any
$u \in U_X$
by the monotonicity  of
$I_X^u$, we have
\if \MkLonger  1
\begin{equation*}
I_X^u(-M)
	\le
I_X^u(f)
	\le
I_X^u(M) .
\end{equation*}
\else 
$
I_X^u(-M)
	\le
I_X^u(f)
	\le
I_X^u(M) .
$
\fi 
Hence, by
Corollary \ref{cor:ChoquetComono} (2)
and the definition of
$\xi_{\mathfrak{X}}(f)$,
we obtain
\if \MkLonger  1
\[
-M
	\le
\xi_{\mathfrak{X}}(f)(u)
	\le
M .
\]
\else 
$
-M
	\le
\xi_{\mathfrak{X}}(f)(u)
	\le
M .
$
\fi 
Therefore,
$\xi_{\mathfrak{X}}(f)$
is bounded.

\end{proof}
\fi 

By Theorem \ref{thm:choquetThm},
we have a following remark.

\begin{rem}
\label{rem:xiX}
Let
$\mathfrak{X} = (X, \Sigma_X, U_X)$
be an uncertainty space
and 
$f, g \in
L^{\infty}(X)
$.
The map
$
\xi_{\mathfrak{X}}
	:
L^{\infty}(X)
	\to
L^{\infty}(U_X)
$
has the following properties.
\begin{enumerate}
\item
$
\xi_{\mathfrak{X}}
$
is 
\newword{monotonic}, i.e.,
$f \ge g$ 
implies
$
\xi_{\mathfrak{X}}
(f)
	\ge
\xi_{\mathfrak{X}}
(g)
$.

\item
$
\xi_{\mathfrak{X}}
$
is
\newword{comonotonically additive}, i.e.,
$
\xi_{\mathfrak{X}}(f+g)
	=
\xi_{\mathfrak{X}}(f)
	+
\xi_{\mathfrak{X}}(g)
$
if
$f$
and
$g$
are comonotonic.

\item
$
\xi_{\mathfrak{X}}
$
is 
\newword{positively homogeneous},
i.e.,
$
\xi_{\mathfrak{X}}
(\lambda f)
	=
\lambda
\xi_{\mathfrak{X}}
(f)
$
for every
$\lambda > 0$.

\end{enumerate}

\end{rem}

But, in general,
the map
$\xi_{\mathfrak{X}}$
is not linear.
So it does not belong to the category
$\Ban$
whose arrows are bounded linear
unless
$
U_X
$
contains only additive capacities.

The following is an example in which
$\xi_{\mathfrak{X}}$
does not preserve comonotonicity.

\begin{exmp}
\label{exmp:comonoCE}
Let
$A_1, A_2, A_3 \in \Sigma_X$
be mutually disjoint non-empty subsets of $X$ such that
$
A_1 \cup A_2 \cup A_3 = X
$.
Define two measurable functions
$
f, g \in L^{\infty}{X}
$
by
\begin{align}
f(x)
		&:=
11 \cdot
\mathbb{1}_X(A_1)
	+
1 \cdot
\mathbb{1}_X(A_2),
	+
0 \cdot
\mathbb{1}_X(A_3),
\label{eq:exmpComonoF}
		\\
g(x)
		&:=
11 \cdot
\mathbb{1}_X(A_1)
	+
10 \cdot
\mathbb{1}_X(A_2),
	+
0 \cdot
\mathbb{1}_X(A_3).
\label{eq:exmpComonoG}
\end{align}
Then, by
Lemma \ref{lem:stepFunComonotonic},
$f$
and
$g$
are comonotonic.
For any
$u \in
U_X
$,
we have
\begin{align*}
\xi_{\mathfrak{X}}(f)
(u)
		&=
(11 - 1) u(A_1)
	+
(1 - 0) u(A_1 \cup A_2)
	+
0 \cdot u(X) ,
		\\
\xi_{\mathfrak{X}}(g)
(u)
		&=
(11 - 10) u(A_1)
	+
(10 - 0) u(A_1 \cup A_2)
	+
0 \cdot u(X) .
\end{align*}

Here, suppose we have two additive capacities
$
u_1, u_2 \in U_X
$
satisfying
\begin{align}
u_1(A_1) &= \frac{1}{3},
		\quad
u_1(A_2) = \frac{1}{3},
		\quad
u_1(A_3) = \frac{1}{3},
		\label{eq:exmpComonoCapOne}
			\\
u_2(A_1) &= \frac{1}{2},
		\quad
u_2(A_2) = \frac{1}{8},
		\quad
u_2(A_3) = \frac{3}{8} .
		\label{eq:exmpComonoCapTwo}
\end{align}
Then,
\if \MkLonger  1
\begin{align*}
\xi_{\mathfrak{X}}(f)(u_1) - \xi_{\mathfrak{X}}(f)(u_2)
			&=
10 \cdot
\big(
	u_1(A_1) - u_2(A_2)
\big)
		+
1  \cdot
\big(
	u_1(A_1 \cup A_2) - u_2(A_1 \cup A_2)
\big) 
			\\&=
10 \cdot
\big(
	\frac{1}{3} - \frac{1}{2}
\big)
		+
1  \cdot
\big(
	\frac{2}{3} - \frac{5}{8}
\big) 
			=
-\frac{13}{8} ,
			\\
\xi_{\mathfrak{X}}(g)(u_1) - \xi_{\mathfrak{X}}(g)(u_2)
			&=
1  \cdot
\big(
	u_1(A_1) - u_2(A_2)
\big)
		+
10  \cdot 
\big(
	u_1(A_1 \cup A_2) - u_2(A_1 \cup A_2)
\big) 
			\\&=
1 \cdot
\big(
	\frac{1}{3} - \frac{1}{2}
\big)
		+
10  \cdot
\big(
	\frac{2}{3} - \frac{5}{8}
\big)
			=
\frac{1}{4} .
\end{align*}
\else 
\begin{align*}
\xi_{\mathfrak{X}}(f)(u_1) - \xi_{\mathfrak{X}}(f)(u_2)
			&=
10 \cdot
\big(
	u_1(A_1) - u_2(A_2)
\big)
		+
1  \cdot
\big(
	u_1(A_1 \cup A_2) - u_2(A_1 \cup A_2)
\big) 
			=
-\frac{13}{8} ,
			\\
\xi_{\mathfrak{X}}(g)(u_1) - \xi_{\mathfrak{X}}(g)(u_2)
			&=
1  \cdot
\big(
	u_1(A_1) - u_2(A_2)
\big)
		+
10  \cdot 
\big(
	u_1(A_1 \cup A_2) - u_2(A_1 \cup A_2)
\big) 
			=
\frac{1}{4} .
\end{align*}
\fi 
Hence,
$
\big(
	\xi_{\mathfrak{X}}(f)(u_1)
		-
	\xi_{\mathfrak{X}}(f)(u_2)
\big)
\big(
	\xi_{\mathfrak{X}}(g)(u_1)
		-
	\xi_{\mathfrak{X}}(g)(u_2)
\big)
		=
-\frac{13}{32} .
$

Therefore,
$
\xi_{\mathfrak{X}}(f)
$
and
$
\xi_{\mathfrak{X}}(g)
$
in
$
L^{\infty}(
	U_X
)
$
are not comonotonic.
\end{exmp}

\section{Hierarchical uncertainty}
\label{sec:hieraUncertain}

In this section,
we express $n$-layer uncertainty,
which we call
\textit{hierarchical uncertainty}, 
by introducing a new concept called U-sequences
which are
inductively defined uncertainty spaces
introduced in
Section \ref{sec:uncertainSp}.
We analyze them mainly in relation to decision theory
such as multi-layer Choquet expectations.

Especially,
we demonstrate the second and the third level uncertainty hierarchies
with a concrete example representing the Ellsberg's paradox.

\subsection{U-sequences}
\label{sec:highLayerValFun}

\begin{defn}
\label{defn:Usequence}
A 
\newword{U-sequence}
$
\mathbb{X}
$
is a sequence
of uncertainty spaces
$
\{
	\mathfrak{X_n}
		=
	(X_n, \Sigma_{X_n}, U_{X_n})
\}_{n \in \mathbb{N}}
$
such that
for 
$n \in \mathbb{N}$,
\if \WithProof 0
$
X_{n+1}
	=
U_{X_n}
$
and
$
\Sigma_{X_{n+1}}
	=
\Sigma_{\mathfrak{X_{n}}}.
$
\else 
\begin{equation}
\label{eq:seqUS}
X_{n+1}
	=
U_{X_n},
		\quad
\Sigma_{X_{n+1}}
	=
\Sigma_{\mathfrak{X_{n}}}.
\end{equation}
\fi 
\end{defn}

\begin{Note}
\label{Note:UseqTerm}
Let
$
\mathfrak{T}
	=
(
	\{*\},
	\{ \emptyset, \{*\}\},
	\{ \bar{*}\}
)
$
is the uncertainty space
where
$
\bar{*}
	:
\{ \emptyset, \{*\}\}
	\to
[0,1]
$
is a capacity defined by
\if \WithProof 0
$
\bar{*}(\emptyset)
	= 0
$
and
$
\bar{*}(\{*\})
	= 1 .
$
\else 
\begin{equation}
\label{eq:UStermDef}
\bar{*}(\emptyset)
	= 0,
\quad
\bar{*}(\{*\})
	= 1 .
\end{equation}
\fi 
Then,
for a U-sequence
$
\mathbb{X}
	=
\{
	\mathfrak{X}_n
\}_{n \in \mathbb{N}}
$,
if
$
	\mathfrak{X}_{n_0}
=
	\mathfrak{T}
$ with some
$n_0 \in \mathbb{N}$,
we have
$
	\mathfrak{X}_{m}
=
	\mathfrak{T}
$
for every
$m \ge n_0$.
We call
$\mathfrak{T}$
the \newword{terminal uncertainty space}.
\end{Note}

For an uncertainty space
$
	\mathfrak{X}
		=
	(X, \Sigma_{X}, U_X)
$,
we write
$
L^{\infty}(\mathfrak{X})
$
for
$
L^{\infty}((X, \Sigma_X))
$.
Then,
U-sequence
$
\{
	\mathfrak{X_n}
		=
	(X_n, \Sigma_{X_n}, U_{X_n})
\}_{n \in \mathbb{N}}
$
yields the following sequence of Banach spaces.

\begin{equation*}
\label{eq:FunSSeq}
\xymatrix@C=30 pt@R=20 pt{
	L^{\infty}(\mathfrak{X}_0)
		\ar @{->}^{
			\xi_{
			\mathfrak{X}_0
			}
		} [r]
&
	L^{\infty}(\mathfrak{X}_1)
		\ar @{->}^{
			\xi_{
			\mathfrak{X}_1
			}
		} [r]
&
	L^{\infty}(\mathfrak{X}_2)
		\ar @{->}^{
			\xi_{
			\mathfrak{X}_2
			}
		} [r]
&
	L^{\infty}(\mathfrak{X}_3)
		\ar @{->}^-{
			\xi_{
			\mathfrak{X}_3
			}
		} [r]
&
	\cdots
}
\end{equation*}

\begin{defn}
\label{defn:etaN}
Let
$
\mathbb{X}
		=
\{
	\mathfrak{X_n}
		=
	(X_n, \Sigma_{X_n}, U_{X_n})
\}_{n \in \mathbb{N}}
$
be a U-sequence.
\begin{enumerate}
\item
Maps
$
\xi_{\mathbb{X}}^{m,n}
		:
L^{\infty}(\mathfrak{X}_m)
		\to
L^{\infty}(\mathfrak{X}_n)
$
for
$
m, n \in \mathbb{N}, 
(m \le n)
$
are defined inductively by
\begin{equation}
\label{eq:xiN}
\xi_{\mathbb{X}}^{m,m}
	:=
\mathrm{Id}_{
	L^{\infty}(\mathfrak{X}_m)
} ,
		\quad
\xi_{\mathbb{X}}^{m,n+1}
	:=
\xi_{
	\mathfrak{X}_n
}
	\circ
\xi_{\mathbb{X}}^{m,n}  .
\end{equation}

\item
Let
$f
\in
L^{\infty}(\mathfrak{X}_0)
$
be an act,
and
$
\mathfrak{u}
: \mathbb{R} \to \mathbb{R}
$
be a utility function.
Then,
the \newword{value function}
of $f$,
$
V_n(f)
	:
X_n
	\to
\mathbb{R}
$
from the $n$th uncertainty layer's point of view
is defined by
\if \WithProof 0
$
V_n(f)
	:=
\xi_{\mathbb{X}}^{0,n}(\mathfrak{u} \circ f) .
$
\else 
\begin{equation}
\label{eq:VnF}
V_n(f)
	:=
\xi_{\mathbb{X}}^{0,n}(\mathfrak{u} \circ f) .
\end{equation}
\fi 

\end{enumerate}
\end{defn}

\begin{prop}
\label{prop:VnInd}
For
$n = 1, 2, \cdots$,
\begin{equation}
\label{eq:VnInd}
V_{n+1}
		=
\xi_{
	\mathfrak{X}_n
}
	\circ
V_n .
\end{equation}
\end{prop}
\if \WithProof 1
\begin{proof}
For any
$
f \in L^{\infty}(\mathfrak{X}_0)
$,
\[
V_{n+1}(f)
		=
\xi_{\mathbb{X}}^{n+1}
(\mathfrak{u} \circ f)
		=
\big(
\xi_{
	\mathfrak{X}_n
}
	\circ
\xi_{\mathbb{X}}^n
\big)
(\mathfrak{u} \circ f)
		=
\xi_{
	\mathfrak{X}_n
}\big(
	V_n(f)
\big) 
		=
\big(
	\xi_{
		\mathfrak{X}_n
	}
		\circ
	V_n
\big)
(f) .
\]

\end{proof}
\fi 

Let
$
	\mathfrak{X}
		=
	(X, \Sigma_{X}, U_X)
$
be an uncertainty space
throughout the rest of this section,

\subsection{Savage's axioms}
\label{sec:SavageAxioms}

This section
recalls 
a part of the axioms introduced by 
\cite{savage1954}.

\begin{defn}
\begin{enumerate}
\item
A random variable
$
f
	\in
L^{\infty}(\mathfrak{X})
$
is called an
\newword{act} (on $X$).

\item
For
$r \in \mathbb{R}$,
the act
$r
	\in
L^{\infty}(\mathfrak{X})
$
is a random variable defined by
$
r(x) := r
$
for every
$x \in X$.

\item
Let
$
f, g
	\in
L^{\infty}(\mathfrak{X})
$
be two acts
and
$
A \in \Sigma_X
$.
Then 
the act
$
(A; f, g)
	\in
L^{\infty}(\mathfrak{X})
$
is defined by
for
$x \in X$,
\begin{equation}
\label{eq:condActs}
(A; f, g)(x)
	:=
\begin{cases}
f(x)
		& \textrm{if} \; x \in A,
\\
g(x)
		& \textrm{if} \; x \in X \setminus A.
\end{cases}
\end{equation}

\end{enumerate}
\end{defn}

\begin{defn}
{[Savage's axioms]}
A
\newword{(Savage) preference order}
on $X$
is a binary relation
$
\succsim_X
$
over
$
L^{\infty}(\mathfrak{X})
$
satisfying seven axioms, 
\textbf{P1}
-
\textbf{P7},
called
\newword{Savage's axioms}.
The followings are the first five axioms,
in which
$
f 
	\succ_X
g
$
means
$
g 
	\succsim_X
f 
$
is untrue.
\begin{enumerate}.
\item[\textbf{P1}]
$
\succsim_X
$
is a weak order.
i.e.
for every
$f, g, h
	\in
L^{\infty}(\mathfrak{X})
$,
	\begin{enumerate}
	\item
	$
	f \succsim_X g
	$
	or
	$
	g \succsim_X f
	$,

	\item
	$
	f \succsim_X g
	$
	and
	$
	g \succsim_X h
	$
	imply
	$
	f \succsim_X h
	$.
	
	\end{enumerate}

\lspace
\item[\textbf{P2}]
For every
$f, g, h, h'
\in 
L^{\infty}(\mathfrak{X})
$
and every
$A \in \Sigma_X$,
\[
(A; f, h)
	\succsim_X
(A; g, h)
	\quad \textrm{iff} \quad
(A; f, h')
	\succsim_X
(A; g, h') .
\]

\lspace
\item[\textbf{P3}]
For every
$f
\in 
L^{\infty}(\mathfrak{X})
$,
$
A \in \Sigma_X
\, (A \ne \emptyset)
$
and
$
r, s \in \mathbb{R}
$,
\[
r
	\succsim_X
s
	\quad \textrm{iff} \quad
(A; r, f)
	\succsim_X
(A; s, f) .
\]

\lspace
\item[\textbf{P4}]
For every
$
A, B
	\in
\Sigma_X
$
and every
$
r, s, t, u \in \mathbb{R}
$
with
$
r \succ_X s
$
and
$
t \succ_X u
$,
\[
(A; r, s) 
	\succsim_X
(B; r, s)
	\quad \textrm{iff} \quad
(A; t, u)
	\succsim_X
(B; t, u) .
\]

\lspace
\item[\textbf{P5}]
There are
$
f, g
\in 
L^{\infty}(\mathfrak{X})
$
such that
$
f
	\succsim_X
g
$.


\end{enumerate}

\end{defn}

\begin{Note}
\label{Note:prefReflexive}
The axiom
\textbf{P1}
implies the relation
$\succsim_X$
is reflexive,
i.e.
$f \succsim_X f$
for any
$
f \in 
L^{\infty}(\mathfrak{X})
$.
\end{Note}
\if \WithProof 1
\begin{proof}
Substitute $g$ by $f$ in the condition (a) of \textbf{P1}.
Then, we have
$
f \succsim_X f
	\; \textrm{or} \;
f \succsim_X f
$,
which means
$
f \succsim_X f
$.

\end{proof}
\fi 


\subsection{Ellsberg's single-urn paradox}
\label{sec:Ellsberg}

In this section, we describe 
the Ellsberg's single-urn paradox (\cite{ellsberg_1961})
in our language developed in 
Section \ref{sec:uncertainSp}.

Suppose we have an urn filled with balls of three colors, red, blue and yellow.
Total number of balls is $3N$ for a fixed positive integer $N$.
The number of red balls in the urn is $N$.
But, we have no information about the ratio of numbers of blue and yellow balls.
In other words, we have an uncertainty space 
$
\mathfrak{X}_0
	=
(X_0, \Sigma_{X_0}, U_{X_0})
$
defined by
\begin{equation}
\label{eq:EllsUrn}
X_0 := \{ R, B, Y \},
	\quad
\Sigma_{X_0} := 2^{X_0},
	\quad
U_{X_0}
	:=
\{
	u_k
\mid
	k = 0, \cdots, 2N
\}
\end{equation}
where
$u_k$
is a capacity on $X$ defined by
with a fixed 
$\alpha \in \mathbb{R}
\;
(\alpha \ge 1)
$,
\begin{align*}
u_k(\{R\})
	&:=
\frac{N}{3N}
	=
\frac{1}{3},
	\quad
u_k(\{B\}) 
	:=
\frac{2}{3}
\big(
	\frac{k}{2N}
\big)^{\alpha},
	\quad
u_k(\{Y\}) 
	:=
\frac{2}{3}
	\big(
		1 - 
		\frac{k}{2N}
	\big)^{\alpha},
	\\
	\quad
u_k(\{B, Y\})
	&:=
1 - u_k(\{R\})
	=
\frac{2}{3},
	\\
u_k(\{R, B\}) 
	&:= 
u_k(\{R\})
	+
 u_k(\{B\})
	=
\frac{1}{3} 
	+
\frac{2}{3}
\big(
	\frac{k}{2N}
\big)^{\alpha},
	\\
u_k(\{R, Y\}) 
	&:= 
u_k(\{R\})
	+
 u_k(\{Y\})
	=
\frac{1}{3} 
	+
\frac{2}{3}
\big(
	1 -
	\frac{k}{2N}
\big)^{\alpha} .
\end{align*}

\begin{Note}
\label{Note:ellsUrn}
\begin{enumerate}
\item
The capacity
$u_k$
becomes an additive probability measure
if $\alpha = 1$.

\item
$
\Sigma_{
	\mathfrak{X}_0
}
	=
2^{
	U_{X_0}
} .
$
\end{enumerate}
\end{Note}
\if \WithProof 1
\begin{proof}
\begin{enumerate}
\item
Obvious.

\item
All we need to show is that
$
\{u_k\}
	\in
\Sigma_{
	\mathfrak{X}_0
}
$
for every
$
u_k
	\in
U_{X_0}
$.
But, 
there exists 
$\delta > 0$
such that
\[
\Big\{
	u
		\in
	U_{X_0}
\mid
	u(\{B\})
		\in
	\big(
		\frac{2}{3}
		\big(
			1 - 
			\frac{k}{2N}
		\big)^{\alpha}
			-
		\delta,
	\,
		\frac{2}{3}
		\big(
			1 - 
			\frac{k}{2N}
		\big)^{\alpha}
			+
		\delta
	\big)
\Big\}
		=
\{u_k\} .
\]
Therefore, by Proposition \ref{prop:sigmaGX},
$
\{u_k\}
	\in
\Sigma_{
	\mathfrak{X}_0
}
$.
\end{enumerate}
\end{proof}
\fi 

Next, we consider the following four acts 
$
f_1, f_2, f_3, f_4
	\in
L^{\infty}(\mathfrak{X}_0)
$
on $X_0$.
\[
\textred{f_1} := \mathbb{1}_{X_0}(\{R\}) ,
    \;
\textred{f_2} := \mathbb{1}_{X_0}(\{B\}) ,
    \;
\textred{f_3} := \mathbb{1}_{X_0}(\{B, Y\}) ,
    \;
\textred{f_4} := \mathbb{1}_{X_0}(\{R, Y\}) .
\]

Do you prefer to bet on which acts?
The modal response here is to prefer the known probability over the unknown probability,
that is,
\begin{equation}
\label{eq:EllsbExp}
f_1 \succ_{X_0} f_2
	\quad \textrm{and} \quad
f_3 \succ_{X_0} f_4 .
\end{equation}

On the other hand,
\if \MkLonger  1
\[
f_1 = 
(
\{R, B\} ;
	f_1,
	0
)
	\quad \textrm{and} \quad
f_2 = 
(
\{R, B\} ;
	f_2,
	0
) .
\]
\else 
$
f_1 = 
(
\{R, B\} ;
	f_1,
	0
)
$
and
$
f_2 = 
(
\{R, B\} ;
	f_2,
	0
) .
$
\fi 
Hence, by the axiom \textbf{P2},
\[
f_1 \succsim_{X_0} f_2
	\; \Leftrightarrow \;
(
\{R, B\} ;
	f_1,
	0
)
	\succsim_{X_0}
(
\{R, B\} ;
	f_2,
	0
)
	\; \Leftrightarrow \;
(
\{R, B\} ;
	f_1,
	1
)
	\succsim_{X_0}
(
\{R, B\} ;
	f_2,
	1
) .
\]
But, easily we can see
\if \MkLonger  1
\[
(
\{R, B\} ;
	f_1,
	1
)
		= f_4
	\quad \textrm{and} \quad
(
\{R, B\} ;
	f_2,
	1
) 
		=f_3 .
\]
\else 
$
(
\{R, B\} ;
	f_1,
	1
)
		= f_4
$
and
$
(
\{R, B\} ;
	f_2,
	1
) 
		=f_3 .
$
\fi 
Therefore,
\[
f_1 \succsim_{X_0} f_2
	\quad \textrm{iff} \quad
f_4 \succsim_{X_0} f_3,
\]
which violates to
(\ref{eq:EllsbExp}).
This observation is called the
\newword{Ellsberg single-urn paradox}.

\subsection{Second layer analysis}
\label{sec:secLayerAna}
From now on,
we are at the position that the preference relation
described in 
Section \ref{sec:SavageAxioms}
should be represented by
\begin{equation}
\label{eq:prefRelvsExpUtil}
f_1 \succsim_{X_0} f_2
		\quad \Leftrightarrow \quad
I_{X_0}^u (\mathfrak{u} \circ f_1)
	\ge
I_{X_0}^u (\mathfrak{u} \circ f_2)
\end{equation}
with a non-decreasing function $\mathfrak{u}$, called a \newword{utility function}.

In this section,
we overcome the Ellsberg' single-urn paradox
by introducing non-additive capacity,
which is actually a well-known fact.

A 
utility function
we use here
is a function
$\mathfrak{u} : \mathbb{R} \to \mathbb{R}$
satisfying
\begin{equation}
\label{eq:utilFun}
1 > \mathfrak{u}(1) > \mathfrak{u}(0) = 0 .
\end{equation}
For example,
$
\mathfrak{u}(x) :=
1 - e^{-x}
$
is a possible utility function
satisfying
(\ref{eq:utilFun}).

Now, for
$
u \in 
X_1
	:=
U_{X_0}
$,
we have
\if \MkLonger  1
\begin{align*}
V_1(f)(u)
		&=
\xi_{\mathfrak{X}_0}
(\mathfrak{u} \circ f)
(u)
		\\&=
I_{X_0}^u
(\mathfrak{u} \circ f)
		\\&=
\int_0^{1}
	dz
		\,
	u(\{
		\mathfrak{u} \circ f
			\ge
		z
	\})
\end{align*}
\else 
\[
V_1(f)(u)
		=
\xi_{\mathfrak{X}_0}
(\mathfrak{u} \circ f)
(u)
		=
I_{X_0}^u
(\mathfrak{u} \circ f)
		=
\int_0^{1}
	dz
		\,
	u(\{
		\mathfrak{u} \circ f
			\ge
		z
	\})
\]
\fi 
Here,
\if \MkLonger  1
\begin{align*}
\{ \mathfrak{u} \circ f \ge z \}
		&=
\{ x \in X_0 \mid \mathfrak{u}(f(x)) \ge z \}
		\\&=
\begin{cases}
X_0
	& \textrm{if} \;
z \le 0,
	\\
\{ f = 1 \}
	& \textrm{if} \;
0 < z \le \mathfrak{u}(1),
	\\
\emptyset
	& \textrm{if} \;
z > \mathfrak{u}(1) .
\end{cases}
\end{align*}
\else 
\[
\{ \mathfrak{u} \circ f \ge z \}
		=
\{ x \in X_0 \mid \mathfrak{u}(f(x)) \ge z \}
		=
\begin{cases}
X_0
	& \textrm{if} \;
z \le 0,
	\\
\{ f = 1 \}
	& \textrm{if} \;
0 < z \le \mathfrak{u}(1),
	\\
\emptyset
	& \textrm{if} \;
z > \mathfrak{u}(1) .
\end{cases}
\]
\fi 
Hence, we obtain
\begin{equation*}
V_1(f)(u)
	=
\int_0^{\mathfrak{u}(1)}
	dz \,
	u(\{f = 1\})
	=
\mathfrak{u}(1)
\,
u(\{f = 1\}) 
	=
\mathfrak{u}(1)
\,
\varepsilon_{\mathfrak{X}_0}
(\{f = 1\}) 
(u) .
\end{equation*}
Therefore,
\begin{equation}
\label{eq:Vone}
V_1(f)
	=
\xi_{\mathfrak{X}_0}(\mathfrak{u} \circ f)
	=
\mathfrak{u}(1)
\,
\varepsilon_{\mathfrak{X}_0}
(\{f = 1\})  .
\end{equation}

Next,
for
a capacity
$v$
over
$
X_1
$,
we have
\if \MkLonger  1
\begin{align}
V_2(f)(v)
		&=
\xi_{\mathbb{X}}^2(\mathfrak{u} \circ f)
(v)
		\nonumber
		\\&=
\xi_{
	\mathfrak{X}_1
}
\big(
	\xi_{\mathfrak{X}_0}(
		\mathfrak{u} \circ f
	)
\big)
(v)
		\nonumber
		\\&=
I_{
	X_1
}^v
\big(
	\mathfrak{u}(1) \,
	\varepsilon_{\mathfrak{X}_0}
	(\{f = 1\}
\big)
		\nonumber
		\\&=
\mathfrak{u}(1)
I_{
	X_1
}^v
\big(
	\varepsilon_{\mathfrak{X}_0}
	(\{f = 1\}
\big)
		\nonumber
\mathfrak{u}(1)
\int_0^{1}
	dz \,
	v
\Big(
\big\{
	u \in 
	X_1
\mid
	\varepsilon_{\mathfrak{X}_0}
	(\{f = 1\}) 
	(u)
\ge
	z
\big\}
\Big)
		\nonumber
		\\&=
\mathfrak{u}(1)
\int_0^{1}
	dz \,
	v
\Big(
\big\{
	u \in 
	X_1
\mid
	u
	(\{f = 1\}) 
\ge
	z
\big\}
\Big) .
		\label{eq:VtwoFv}
\end{align}
\else 
\begin{align}
V_2(f)(v)
		&=
\xi_{\mathbb{X}}^2(\mathfrak{u} \circ f)
(v)
		=
\xi_{
	\mathfrak{X}_1
}
\big(
	\xi_{\mathfrak{X}_0}(
		\mathfrak{u} \circ f
	)
\big)
(v)
		=
I_{
	X_1
}^v
\big(
	\mathfrak{u}(1) \,
	\varepsilon_{\mathfrak{X}_0}
	(\{f = 1\}
\big)
		\nonumber
		\\&=
\mathfrak{u}(1)
I_{
	X_1
}^v
\big(
	\varepsilon_{\mathfrak{X}_0}
	(\{f = 1\}
\big)
		\nonumber
\mathfrak{u}(1)
\int_0^{1}
	dz \,
	v
\Big(
\big\{
	u \in 
	X_1
\mid
	\varepsilon_{\mathfrak{X}_0}
	(\{f = 1\}) 
	(u)
\ge
	z
\big\}
\Big)
		\nonumber
		\\&=
\mathfrak{u}(1)
\int_0^{1}
	dz \,
	v
\Big(
\big\{
	u \in 
	X_1
\mid
	u
	(\{f = 1\}) 
\ge
	z
\big\}
\Big) .
		\label{eq:VtwoFv}
\end{align}
\fi 

Let us calculate value functions for acts
$f_1, f_2, f_3, f_4$
by using (\ref{eq:VtwoFv}).
\begin{align*}
		&
u_k(\{ f_1 = 1 \}) \ge z
		 \; \Leftrightarrow \; 
u_k(\{R\}) \ge z
		\; \Leftrightarrow \; 
z
		\le
\frac{1}{3} ,
		\\&
u_k(\{ f_2 = 1 \}) \ge z
		 \; \Leftrightarrow \; 
u_k(\{B\}) \ge z
		 \; \Leftrightarrow \; 
z
		\le
\frac{2}{3}
\big(
	\frac{k}{2N}
\big)^{\alpha}
		 \; \Leftrightarrow \; 
k
	\ge
2N
\big(
	\frac{3}{2}
	z
\big)^{
	\frac{1}{\alpha}
} ,
		\\&
u_k(\{ f_3 = 1 \}) \ge z
		 \; \Leftrightarrow \; 
u_k(\{B, Y\}) \ge z
		 \; \Leftrightarrow \; 
z
		\le 
\frac{2}{3}, 
		\\&
u_k(\{ f_4 = 1 \}) \ge z
		 \; \Leftrightarrow \; 
u_k(\{R, Y\}) \ge z
		 \; \Leftrightarrow \; 
z
		\le 
\frac{1}{3}
	+
\frac{2}{3}
\big(
	1 - 
	\frac{k}{2N}
\big)^{\alpha} 
		 \; \Leftrightarrow \; 
k
	\le
2N 
\Big(
	1 -
	\big(
		\frac{3z - 1}{2}
	\big)^{
		\frac{1}{\alpha}
	}
\Big) .
\end{align*}

Here, assume that $v$ is an additive probability measure.
Then, we obtain
\if \MkLonger  1
\begin{align}
V_2(f_1)(v)
		&=
\mathfrak{u}(1)
\int_0^{1}
	dz \,
	v
\Big(
\big\{
	u_k \in 
	X_1
\mid
	z \le
	\frac{1}{3}
\big\}
\Big) 
		\nonumber
		\\&=
\mathfrak{u}(1)
\int_0^{
	\frac{1}{3}
}
	dz \,
	v(
		X_1
	)
		\nonumber
		\\&=
\frac{1}{3}
\mathfrak{u}(1) ,
		\label{eq:VtwoFone}
		\\
V_2(f_2)(v)
		&=
\mathfrak{u}(1)
\int_0^{1}
	dz \,
	v
\Big(
\big\{
	u_k \in 
	X_1
\mid
z
		\le
\frac{2}{3}
\big(
	\frac{k}{2N}
\big)^{\alpha}
\big\}
\Big) 
		\nonumber
		\\&=
\mathfrak{u}(1)
\sum_{k=0}^{2N}
\int_0^{
	\frac{2}{3}
	\big(
		\frac{k}{2N}
	\big)^{\alpha}
}
dz \,
v
(
\{
	u_k
\}
)
		&(\textrm{since $v$ is additive})
		\nonumber
		\\&=
\frac{2}{3}
\mathfrak{u}(1)
\sum_{k=0}^{2N}
	\big(
		\frac{k}{2N}
	\big)^{\alpha}
v
(
\{
	u_k
\}
) ,
		\label{eq:VtwoFtwo}
\end{align}
\begin{align}
V_2(f_3)(v)
		&=
\mathfrak{u}(1)
\int_0^{1}
	dz \,
	v
\Big(
\big\{
	u_k \in 
	X_1
\mid
	z \le
	\frac{2}{3}
\big\}
\Big) 
		\nonumber
		\\&=
\mathfrak{u}(1)
\int_0^{
	\frac{2}{3}
}
	dz \,
	v(
		X_1
	)
		\nonumber
		\\&=
\frac{2}{3}
\mathfrak{u}(1) ,
		\label{eq:VtwoFthree}
		\\
V_2(f_4)(v)
		&=
\mathfrak{u}(1)
\int_0^{1}
	dz \,
	v
\Big(
\big\{
	u_k \in 
	X_1
\mid
z
		\le
\frac{1}{3}
	+
\frac{2}{3}
\big(
	1 - 
	\frac{k}{2N}
\big)^{\alpha} 
\big\}
\Big) 
		\nonumber
		\\&=
\mathfrak{u}(1)
\sum_{k=0}^{2N}
\int_0^{
	\frac{1}{3}
		+
	\frac{2}{3}
	\big(
		1 - 
		\frac{k}{2N}
	\big)^{\alpha} 
}
dz \,
v
(
\{
	u_k
\}
)
		&(\textrm{since $v$ is additive})
		\nonumber
		\\&=
\mathfrak{u}(1)
\sum_{k=0}^{2N}
\Big(
	\frac{1}{3}
		+
	\frac{2}{3}
	\big(
		1 - 
		\frac{k}{2N}
	\big)^{\alpha} 
\Big)
v
(
\{
	u_k
\}
)
		\nonumber
		\\&=
\mathfrak{u}(1)
\Big(
	\frac{1}{3}
		+
	\frac{2}{3}
\sum_{k=0}^{2N}
	\big(
		1 - 
		\frac{k}{2N}
	\big)^{\alpha} 
		\,
v
(
\{
	u_k
\}
)
\Big) .
		\label{eq:VtwoFfour}
\end{align}
\else 
since $v$ is additive,
\begin{align}
V_2(f_1)(v)
		&=
\mathfrak{u}(1)
\int_0^{1}
	dz \,
	v
\Big(
\big\{
	u_k \in 
	X_1
\mid
	z \le
	\frac{1}{3}
\big\}
\Big) 
		=
\mathfrak{u}(1)
\int_0^{
	\frac{1}{3}
}
	dz \,
	v(
		X_1
	)
		=
\frac{1}{3}
\mathfrak{u}(1) ,
		\label{eq:VtwoFone}
		\\
V_2(f_2)(v)
		&=
\mathfrak{u}(1)
\int_0^{1}
	dz \,
	v
\Big(
\big\{
	u_k \in 
	X_1
\mid
z
		\le
\frac{2}{3}
\big(
	\frac{k}{2N}
\big)^{\alpha}
\big\}
\Big) 
		=
\mathfrak{u}(1)
\sum_{k=0}^{2N}
\int_0^{
	\frac{2}{3}
	\big(
		\frac{k}{2N}
	\big)^{\alpha}
}
dz \,
v
(
\{
	u_k
\}
)
		\nonumber
		\\&=
\frac{2}{3}
\mathfrak{u}(1)
\sum_{k=0}^{2N}
	\big(
		\frac{k}{2N}
	\big)^{\alpha}
v
(
\{
	u_k
\}
) ,
		\label{eq:VtwoFtwo}
		\\
V_2(f_3)(v)
		&=
\mathfrak{u}(1)
\int_0^{1}
	dz \,
	v
\Big(
\big\{
	u_k \in 
	X_1
\mid
	z \le
	\frac{2}{3}
\big\}
\Big) 
		=
\mathfrak{u}(1)
\int_0^{
	\frac{2}{3}
}
	dz \,
	v(
		X_1
	)
		=
\frac{2}{3}
\mathfrak{u}(1) ,
		\label{eq:VtwoFthree}
		\\
V_2(f_4)(v)
		&=
\mathfrak{u}(1)
\int_0^{1}
	dz \,
	v
\Big(
\big\{
	u_k \in 
	X_1
\mid
z
		\le
\frac{1}{3}
	+
\frac{2}{3}
\big(
	1 - 
	\frac{k}{2N}
\big)^{\alpha} 
\big\}
\Big) 
		\nonumber
		\\&=
\mathfrak{u}(1)
\sum_{k=0}^{2N}
\int_0^{
	\frac{1}{3}
		+
	\frac{2}{3}
	\big(
		1 - 
		\frac{k}{2N}
	\big)^{\alpha} 
}
dz \,
v
(
\{
	u_k
\}
)
		=
\mathfrak{u}(1)
\sum_{k=0}^{2N}
\Big(
	\frac{1}{3}
		+
	\frac{2}{3}
	\big(
		1 - 
		\frac{k}{2N}
	\big)^{\alpha} 
\Big)
v
(
\{
	u_k
\}
)
		\nonumber
		\\&=
\mathfrak{u}(1)
\Big(
	\frac{1}{3}
		+
	\frac{2}{3}
\sum_{k=0}^{2N}
	\big(
		1 - 
		\frac{k}{2N}
	\big)^{\alpha} 
		\,
v
(
\{
	u_k
\}
)
\Big) .
		\label{eq:VtwoFfour}
\end{align}
\fi 

Now,
let us consider
an uncertainty space
$
\mathfrak{X}_1
	=
(X_1, \Sigma_{X_1}, U_{X_1})
$
defined by
\begin{equation}
\label{eq:EllXOne}
\Sigma_{X_1}
	:=
2^{X_1},
	\quad	
U_{X_1}
	:=
\{
	v^{\mathrm{u}}
\}
\end{equation}
where
$v^{\mathrm{u}}$
is the uniform distribution over
$
X_1
$,
that is,
$
v^{\mathrm{u}}
	:
\Sigma_{X_1}
	\to
[0,1]
$
is
an additive capacity
defined by
for
$
A \in
	\Sigma_{X_1}
$,
\begin{equation}
\label{eq:vStar}
v^{\mathrm{u}}(A)
	:=
\frac{
	\#A
}{
	2N+1
} .
\end{equation}
Here
$\#A$
is the cardinality of $A$.
Then, we have
\begin{align}
V_2(f_1)(v^{\mathrm{u}})
		&=
\frac{1}{3}
\mathfrak{u}(1) ,
		\label{eq:VtwoFoneVS}
		\\
V_2(f_2)(v^{\mathrm{u}})
		&=
\frac{2}{3}
\mathfrak{u}(1)
\frac{1}{
	2N + 1
}
\sum_{k=0}^{2N}
	\big(\frac{k}{2N}\big)^{\alpha} ,
		\label{eq:VtwoFtwoVS}
		\\
V_2(f_3)(v^{\mathrm{u}})
		&=
\frac{2}{3}
\mathfrak{u}(1) ,
		\label{eq:VtwoFthreeVS}
		\\
V_2(f_4)(v^{\mathrm{u}})
		&=
\mathfrak{u}(1)
\Big(
	\frac{1}{3}
		+
	\frac{2}{
		3
		(2N+1)
	}
\sum_{k=0}^{2N}
	\big(\frac{k}{2N}\big)^{\alpha}
\Big) .
		\label{eq:VtwoFfourVS}
\end{align}

Now, if 
$\alpha = 1$,
then,
\begin{equation}
\label{eq:VstarAlOne}
V_2(f_1)(v^{\mathrm{u}})
		=
\frac{1}{3} \mathfrak{u}(1)
		=
V_2(f_2)(v^{\mathrm{u}}),
			\quad \textrm{and} \quad
V_2(f_3)(v^{\mathrm{u}})
		=
\frac{2}{3} \mathfrak{u}(1)
		=
V_2(f_4)(v^{\mathrm{u}}) .
\end{equation}
Hence we fail to represent
(\ref{eq:EllsbExp})
with them.

However, if
$\alpha > 1$,
then
\begin{equation}
\label{eq:VstarAlOnePlus}
V_2(f_1)(v^{\mathrm{u}})
		=
\frac{1}{3} \mathfrak{u}(1)
		>
V_2(f_2)(v^{\mathrm{u}})
			\quad \textrm{and} \quad
V_2(f_3)(v^{\mathrm{u}})
		=
\frac{2}{3} \mathfrak{u}(1)
		>
V_2(f_4)(v^{\mathrm{u}}) 
\end{equation}
since
$
\big(\frac{k}{2N}\big)^{\alpha} 
	<
\frac{k}{2N} ,
$
which supports
(\ref{eq:EllsbExp}).
%
Note that
the above discussion is on the U-sequence
$
\mathbb{X}
	=
\{
	\mathfrak{X}_n
\}_{n \in \mathbb{N}}
$
where
$
	\mathfrak{X}_n
=
\mathfrak{T}
$
for
$
n \ge = 2
$.

Next, we slightly modify the U-sequence
$\mathbb{X}$
to
another U-sequence
$
\mathbb{Y} = 
\{
	\mathfrak{Y}_n
\}_{n \in \mathbb{N}}
$
such that
\begin{equation}
\label{eq:UseqY}
\mathfrak{Y}_n
	=
(Y_n, \Sigma_{Y_n}, U_{Y_n})
	:=
\begin{cases}
	(X_1, \Sigma_{X_1}, \{ v^{\mathrm{b}}\})
& \textrm{if} \;
	n = 1 ,
\\
	(X_n, \Sigma_{X_n}, U_{X_n})
& \textrm{otherwise} ,
\end{cases}
\end{equation}
where
$
v^{\mathrm{b}}
	:
\Sigma_{Y_1}
	\to
[0, 1]
$
is a symmetric binary distribution defined by
for
$
A \in \Sigma_{Y_1}
$
\begin{equation}
\label{eq:vB}
v^{\mathrm{b}}(A)
	:=
2^{-2N}
\sum_{u_k \in A}
	\binom{2N}{k} .
\end{equation}
Then, 
by
		(\ref{eq:VtwoFone}),
		(\ref{eq:VtwoFtwo}),
		(\ref{eq:VtwoFthree})
and
		(\ref{eq:VtwoFfour}),
we have
\begin{align}
V_2(f_1)(v^{\mathrm{b}})
		&=
\frac{1}{3}
\mathfrak{u}(1) ,
		\label{eq:VtwoFoneVN}
		\\
V_2(f_2)(v^{\mathrm{b}})
		&=
\frac{2}{
	3
	\cdot
	2^{2N}
}
\mathfrak{u}(1)
\sum_{k=0}^{2N}
	\big(\frac{k}{2N}\big)^{\alpha} 
	\binom{2N}{k}
		\label{eq:VtwoFtwoVN}
		\\
V_2(f_3)(v^{\mathrm{b}})
		&=
\frac{2}{3}
\mathfrak{u}(1) ,
		\label{eq:VtwoFthreeVN}
		\\
V_2(f_4)(v^{\mathrm{b}})
		&=
\mathfrak{u}(1)
\Big(
	\frac{1}{3}
		+
	\frac{2}{
		3
		\cdot
		2^{2N}
	}
\sum_{k=0}^{2N}
	\big(\frac{2N-k}{2N}\big)^{\alpha}
	\binom{2N}{k}
\Big) .
		\label{eq:VtwoFfourVN}
\end{align}
Now if 
$\alpha=1$,
then
\if \MkLonger  1
\begin{align*}
V_2(f_2)(v^{\mathrm{b}})
		&=
\frac{2}{
	3
	\cdot
	2^{2N}
}
\mathfrak{u}(1)
\sum_{k=0}^{2N}
	\big(\frac{k}{2N}\big)
	\binom{2N}{k}
		\\&=
\frac{2}{
	3
	\cdot
	2^{2N}
}
\mathfrak{u}(1)
\sum_{k=1}^{2N}
	\binom{2N-1}{k-1}
		\\&=
\frac{2}{
	3
	\cdot
	2^{2N}
}
\mathfrak{u}(1)
\sum_{k=0}^{2N-1}
	\binom{2N-1}{k}
		\\&=
\frac{
	2
	\cdot
	2^{2N-1}
}{
	3
	\cdot
	2^{2N}
}
\mathfrak{u}(1)
		\\&=
\frac{
	1
}{
	3
}
\mathfrak{u}(1) ,
		\\
V_2(f_4)(v^{\mathrm{b}})
		&=
\mathfrak{u}(1)
\Big(
	\frac{1}{3}
		+
	\frac{2}{
		3
		\cdot
		2^{2N}
	}
\sum_{k=0}^{2N}
	\big(\frac{2N-k}{2N}\big)
	\binom{2N}{k}
\Big)
		\\&=
\mathfrak{u}(1)
\Big(
	\frac{1}{3}
		+
	\frac{2}{
		3
		\cdot
		2^{2N}
	}
\sum_{k=0}^{2N-1}
	\binom{2N-1}{k}
\Big)
		\\&=
\mathfrak{u}(1)
\Big(
	\frac{1}{3}
		+
	\frac{
		2
		\cdot
		2^{2N-1}
	}{
		3
		\cdot
		2^{2N}
	}
\Big)
		\\&=
\frac{2}{3}
\mathfrak{u}(1) .
\end{align*}
Therefore,
\fi 
\begin{equation}
\label{eq:VstarAlOneB}
V_2(f_1)(v^{\mathrm{u}})
		=
\frac{1}{3} \mathfrak{u}(1)
		=
V_2(f_2)(v^{\mathrm{u}})
			\quad \textrm{and} \quad
V_2(f_3)(v^{\mathrm{u}})
		=
\frac{2}{3} \mathfrak{u}(1)
		=
V_2(f_4)(v^{\mathrm{u}}) .
\end{equation}
Hence we fail to represent
(\ref{eq:EllsbExp})
with them.

However, if
$\alpha > 1$,
then
\begin{equation}
\label{eq:VstarAlOnePlusB}
V_2(f_1)(v^{\mathrm{u}})
		=
\frac{1}{3} \mathfrak{u}(1)
		>
V_2(f_2)(v^{\mathrm{u}})
			\quad \textrm{and} \quad
V_2(f_3)(v^{\mathrm{u}})
		=
\frac{2}{3} \mathfrak{u}(1)
		>
V_2(f_4)(v^{\mathrm{u}}) .
\end{equation}

\subsection{Third layer analysis}
\label{sec:thiLayerAna}

So far,
we have assumed that the probability distribution on
$
X_1 = U_{X_0}
	=
\{u_k \mid k = 0,1,\cdots, 2N\}
$
is predetermined
such as 
a uniform distribution
or 
a symmetric binary distribution.
However, 
the situation changes when we consider, 
for example, 
the process 
to produce the urns in question
in a factory.
Suppose that 
in the factory,
$N$ red balls are supplied from one large tank filled with red balls each time, 
while 
$2N$ blue and yellow balls are supplied together from another large tank
filled with blue and yellow balls whose filling ratio is unknown,
Then, each urn is filled with 3N balls in total.
As a result, 
the blue and yellow balls are filled for each urn 
according to the binomial distribution
that is determined by the unknown ratio of blue and yellow balls in the second tank.
Let us think this situation by introducing another U-sequence
$
\mathbb{Z}
	=
\{
	\mathfrak{Z}_n
		=
	(Z_n, \Sigma_{Z_n}, U_{Z_n})
\}_{n \in \mathbb{N}}
$
such that
$
	\mathfrak{Z}_0
=
	\mathfrak{X}_0
$
but
\begin{equation}
\label{eq:SSX}
	\mathfrak{Z}_1
=
	(Z_1, \Sigma_{Z_1}, U_{Z_1})
=
	(X_1, \Sigma_{X_1}, 
\{
	v_p
\mid
	p \in [0,1]
\}
)
\end{equation}
where
$
v_p
$
is
an additive capacity
defined  by
for
$A \in \Sigma_{
	Z_1
}$,
\begin{equation}
\label{eq:vK}
v_p(A)
	:=
\sum_{u_k \in A}
	\binom{2N}{k}
	p^k
	(1-p)^{2N-k} .
\end{equation}

\begin{Note}
\label{Note:StwoXiso}
Two measurable spaces
$
(
	Z_2,
	\Sigma_{
		Z_2
	}
)
	=
(
	U_{Z_1},
	\Sigma_{
		\mathfrak{Z}_1
	}
)
$
and
$
(
	[0,1],
	\mathcal{B}([0,1])
)
$
are isomorphic.
Especially, the function
$
w : [0,1] \to Z_2
$
defined by
$w(p) := v_p$
for
$p \in [0,1]$
is an isomorphism.
\end{Note}
\if \WithProof 1
\begin{proof}
It is obvious that the map
$v$
is 1-1 and onto.
So all we need to show is that
both
$v$ and $v^{-1}$ are measurable maps.
To this end, we will prove
\begin{equation}
\label{eq:vMeas}
\{
	p
\mid
	v_p(A) \in B
\}
	\in
\mathcal{B}([0,1])
	\quad
(
A \in \Sigma_{Z_1},
 \;
B \in \mathcal{B}([0,1])
)
\end{equation}
and
\begin{equation}
\label{eq:vInvMeas}
\{
	v_p
\mid
	p < s < q
\}
	\in
\Sigma_{Z_2}
	\quad
(
	0 \le p < q \le 1
)
\end{equation}
one by one.

Let
$g(p) := v_p(A)
	=
\sum_{u_k \in A}
	\binom{2N}{k}
	p^k
	(1-p)^{2N-k}
$.
Then,
since $A$ is a finite set,
the function 
$g : [0,1] \to [0.1]$
is polynomial, 
which is measurable.
Therefore,
\[
\{
	p
\mid
	v_p(A) \in B
\}
		=
\{
	p
\mid
	g(p) \in B
\}
		=
g^{-1}(B)
		\in
\mathcal{B}([0,1]).
\]

Next,
we will check the shape of the graph
of 
\[
v_p(\{u_k\})
	=
\binom{2N}{k}
p^k
(1-p)^{2N-k}
\]
as a function of 
$p \in [0,1]$,
where $k \in \{0, 1, \cdots, 2N\}$ is fixed.
We have
\begin{align*}
		&
v_0(\{u_k\})
	=
v_1(\{u_k\})
	=
0 ,
		\\&
\frac{
	\partial
}{
	\partial
	p
}
v_p(\{u_k\})
		=
\binom{2N}{k}
p^{k-1}
(1-p)^{2N-k-1}
(k - 2Np) ,
		\\&
\frac{
	\partial
}{
	\partial
	p
}
v_p(\{u_k\})
		\mid_{p=0}
		=
\frac{
	\partial
}{
	\partial
	p
}
v_p(\{u_k\})
		\mid_{p=1}
		=
0.
\end{align*}
Therefore,
$
v_p(\{u_k\})
$
is increasing on
$
[0, \frac{k}{2N}]
$,
and
is decreasing on
$
[\frac{k}{2N}, 1]
$.

Now let 
$p, q \in [0,1] \; (p < q)$
be fixed.
For any
$
r \in 
(p, q) \cap \mathbb{Q}
$
and
$
k = 0, 1, \cdots, 2N
$,
we will define
$
\delta_{r,k} > 0
$
and
$
B_{r,k}
	\in
\mathcal{B}([0,1])
$
by the following procedure.
\begin{enumerate}
\item [Case $r < \frac{k}{2N}$:]
Pick
$
\delta_{r,k}
 > 0$
such as
$
p
	<
r - \delta_{r,k}
	<
r
	<
r + \delta_{r,k}
	<
\frac{k}{2N} .
$
Define
$B_{r,k}$
by
\begin{equation}
\label{eq:isoCaseOne}
B_{r,k}
	:=
[
	v_{
		r - \delta_{r,k}
	} (\{u_k\}) ,
	v_{
		r + \delta_{r,k}
	} (\{u_k\})
] .
\end{equation}

\item [Case $r = \frac{k}{2N}$:]
Pick
$
\delta_{r,k}
 > 0$
such as
$
p
	<
r - \delta_{r,k}
	<
r
	=
\frac{k}{2N}
	<
r + \delta_{r,k}
	<
q .
$
Define
$B_{r,k}$
by
\begin{equation}
\label{eq:isoCaseTwo}
B_{r,k}
	:=
[
\max\big(
	v_{
		r - \delta_{r,k}
	} (\{u_k\}) ,
	v_{
		r + \delta_{r,k}
	} (\{u_k\})
\big),
v_r(\{u_k\})
] .
\end{equation}

\item [Case $r > \frac{k}{2N}$:]
Pick
$
\delta_{r,k}
 > 0$
such as
$
\frac{k}{2N} 
	<
r - \delta_{r,k}
	<
r
	<
r + \delta_{r,k}
	<
q .
$
Define
$B_{r,k}$
by
\begin{equation}
\label{eq:isoCaseThree}
B_{r,k}
	:=
[
	v_{
		r + \delta_{r,k}
	} (\{u_k\}) ,
	v_{
		r - \delta_{r,k}
	} (\{u_k\})
] .
\end{equation}

\end{enumerate}
Then in these three cases, we have
\begin{equation}
\label{eq:isoCaseZero}
\{
	v_s
\mid
	v_s(\{u_k\})
		\in
	B_{r,k}
\}
			\subset
\{
	v_s
\mid
	p < s < q
\} .
\end{equation}
We can easily check
\begin{equation}
\label{eq:isoCaseFinal}
\{
	v_s
\mid
	p < s < q
\} 
		=
\bigcup_{
	r \in
	(p, q) \cap \mathbb{Q}
}
	\,
\bigcap_{k = 0}^{2N}
\{
	v_s
\mid
	v_s(\{u_k\})
		\in
	B_{r,k}
\} 
		\in
\Sigma_{
	X_2
} ,
\end{equation}
which completes the proof.

\end{proof}
\fi 

Now let us define the Uncertainty space
$
	\mathfrak{Z}_2
$
by
\begin{equation}
\label{eq:UseqZtwo}
	\mathfrak{Z}_2
=
	(Z_2, \Sigma_{Z_2}, U_{Z_2})
:=
	(U_{Z_1}, \Sigma_{\mathfrak{Z_1}}, \{ \lambda \})
\end{equation}
where
$\lambda$
is the Lebesgue measure on 
$[0,1]$.
This is well-defined by
Note \ref{Note:StwoXiso}.
Since
the Lebesgue measure is additive,
we obtain
for
$
f \in 
L^{\infty}
(
\mathfrak{Z}_2
)
$,
\begin{equation}
\label{eq:LebesqueStowX}
\xi_{
	\mathfrak{Z_2}
}
(f)
(\lambda)
		=
I_{
	Z_2
}^{\lambda}
(f)
		=
\int_0^1
	dp \,
	f(v_p) .
\end{equation}
Then, by 
(\ref{eq:VnInd}),
\begin{equation}
\label{eq:VthreeF}
V_3(f)(\lambda)
			=
\xi_{
	\mathfrak{Z_2}
}\big( 
	V_2(f)
\big)
(\lambda)
			=
\int_0^1
	dp \,
	V_2(f)
	(v_p) .
\end{equation}

Hence, by using
(\ref{eq:VtwoFone}),
(\ref{eq:VtwoFtwo})
(\ref{eq:VtwoFthree})
and
(\ref{eq:VtwoFfour}),
we get
\begin{align}
V_3(f_1)(\lambda)
		&=
\int_0^1
	dp \,
\frac{1}{3}
\mathfrak{u}(1)
		=
\frac{1}{3}
\mathfrak{u}(1) ,
		\label{eq:VtwoFoneVH}
		\\
V_3(f_2)(\lambda)
		&=
\int_0^1
	dp \,
\frac{2}{3}
\mathfrak{u}(1)
\sum_{k=0}^{2N}
	\big(
		\frac{k}{2N}
	\big)^{\alpha}
v_p
(
\{
	u_k
\}
) 
		\nonumber
		\\&=
\frac{2}{3}
\mathfrak{u}(1)
\int_0^1
	dp \,
\sum_{k=0}^{2N}
	\big(
		\frac{k}{2N}
	\big)^{\alpha}
v_p
(
\{
	u_k
\}
) ,
		\label{eq:VtwoFtwoVH}
		\\
V_3(f_3)(\lambda)
		&=
\int_0^1
	dp \,
\frac{2}{3}
\mathfrak{u}(1) 
		=
\frac{2}{3}
\mathfrak{u}(1) ,
		\label{eq:VtwoFthreeVH}
		\\
V_3(f_4)(\lambda)
		&=
\int_0^1
	dp \,
\mathfrak{u}(1)
\Big(
	\frac{1}{3}
		+
	\frac{2}{3}
\sum_{k=0}^{2N}
	\big(
		1 - 
		\frac{k}{2N}
	\big)^{\alpha} 
		\,
v_p
(
\{
	u_k
\}
)
\Big) 
		\\&=
\mathfrak{u}(1)
\Big(
	\frac{1}{3}
			+
	\frac{2}{3}
	\int_0^1
		dp \,
		\sum_{k=0}^{2N}
			\big(
				1 - 
				\frac{k}{2N}
			\big)^{\alpha} 
		\,
		v_p
		(
		\{
			u_k
		\}
		)
\Big) .
		\label{eq:VtwoFfourVH}
\end{align}

When
$
\alpha = 1
$,
i.e.
the additive case, 
since
$
\sum_{k=0}^{2N}
	k
		\,
	v_p(\{u_k\})
		=
2Np
$,
we obtain
\if \MkLonger  1
\begin{align}
V_3(f_2)(\lambda)
		&=
\frac{2}{3}
\mathfrak{u}(1)
\int_0^1 dp \,
p
		\nonumber
		\\&=
\frac{1}{3}
\mathfrak{u}(1) ,
		\\
V_3(f_4)(\lambda)
		\nonumber
		&=
\mathfrak{u}(1)
\Big(
	\frac{1}{3}
			+
	\frac{2}{3}
	\int_0^1
		dp \,
		\sum_{k=0}^{2N}
			\big(
				1 - 
				\frac{k}{2N}
			\big)
		\,
		v_p
		(
		\{
			u_k
		\}
		)
\Big) 
		\nonumber
		\\&=
\mathfrak{u}(1)
\Big(
	1
			-
	\frac{2}{3}
	\int_0^1
		dp \,
		\sum_{k=0}^{2N}
			\frac{k}{2N}
		\,
		v_p
		(
		\{
			u_k
		\}
		)
\Big) 
		\nonumber
		\\&=
\mathfrak{u}(1)
\Big(
	1
			-
	\frac{2}{3}
	\int_0^1
		dp \,
		p
\Big) 
		\nonumber
		\\&=
\frac{2}{3}
\mathfrak{u}(1) .
\end{align}
\else 
\begin{align}
V_3(f_2)(\lambda)
		&=
\frac{2}{3}
\mathfrak{u}(1)
\int_0^1 dp \,
p
		=
\frac{1}{3}
\mathfrak{u}(1) ,
		\\
V_3(f_4)(\lambda)
		\nonumber
		&=
\mathfrak{u}(1)
\Big(
	\frac{1}{3}
			+
	\frac{2}{3}
	\int_0^1
		dp \,
		\sum_{k=0}^{2N}
			\big(
				1 - 
				\frac{k}{2N}
			\big)
		\,
		v_p
		(
		\{
			u_k
		\}
		)
\Big) 
		\nonumber
		\\&=
\mathfrak{u}(1)
\Big(
	1
			-
	\frac{2}{3}
	\int_0^1
		dp \,
		\sum_{k=0}^{2N}
			\frac{k}{2N}
		\,
		v_p
		(
		\{
			u_k
		\}
		)
\Big) 
		=
\mathfrak{u}(1)
\Big(
	1
			-
	\frac{2}{3}
	\int_0^1
		dp \,
		p
\Big) 
		=
\frac{2}{3}
\mathfrak{u}(1) .
\end{align}
\fi 
Hence we fail to represent
(\ref{eq:EllsbExp})
with them.

However, if
$\alpha > 1$,
then
\begin{equation}
\label{eq:VstarAlTwoPlus}
V_3(f_1)(\lambda)
		=
\frac{1}{3} \mathfrak{u}(1)
		>
V_3(f_2)(\lambda),
			\quad \textrm{and} \quad
V_3(f_3)(\lambda)
		=
\frac{2}{3} \mathfrak{u}(1)
		>
V_3(f_4)(\lambda) 
\end{equation}
since
$
\big(\frac{k}{2N}\big)^{\alpha} 
	<
\frac{k}{2N} ,
$
which supports
(\ref{eq:EllsbExp}).

Note that 
in
the U-sequence
$\mathbb{Z}$,
$\mathfrak{Z}_n = \mathfrak{T}$
if
$n \ge 3$ .

\section{Categories of Uncertainty Spaces}
\label{sec:catUnc}

In this section,
we introduce two categories of uncertainty spaces,
one is with arrows based on absolutely continuous relation,
and the other is with measure preserving arrows.

For those who are not so familiar with the category theory,
please refer to \cite{maclane1997}.

\subsection{Categories $\Unc$ and $\mpUnc$}
\label{sec:UncAndMpUnc}

\begin{defn}
\label{defn:UncMap}

Let
$
\mathfrak{X}
	=
(X, \Sigma_X, U_X)
$
and
$
\mathfrak{Y}
	=
(Y, \Sigma_Y, U_Y)
$
be two uncertainty spaces,
and
$
f : 
(X, \Sigma_X)
	\to
(Y, \Sigma_Y)
$
be a measurable map.
\begin{enumerate}
\item
$f$
is called a
\newword{
$\Unc$-map
}
from
$\mathfrak{X}$
to
$\mathfrak{Y}$
if
for all
$u \in U_X$
there exists
$v \in U_Y$
such that
$
u \circ f^{-1}
	\ll
v
$.

\item
$f$
is called an
\newword{
$\mpUnc$-map
}
from
$\mathfrak{X}$
to
$\mathfrak{Y}$
if
for all
$u \in U_X$,
$
u \circ f^{-1}
	\in
U_Y
$.

\end{enumerate}
\end{defn}

Note that every 
$\mpUnc$-map
is a
$\Unc$-map.

\begin{prop}
\label{prop:compUnc}
Let
$
\mathfrak{X}
	=
(X, \Sigma_X, U_X)
$,
$
\mathfrak{Y}
	=
(Y, \Sigma_Y, U_Y)
$
and
$
\mathfrak{Z}
	=
(Z, \Sigma_Z, U_Z)
$
be uncertainty spaces,
and
$
f : 
(X, \Sigma_X)
	\to
(Y, \Sigma_Y)
$
and
$
g : 
(Y, \Sigma_Y)
	\to
(Z, \Sigma_Z)
$
be measurable maps.
\begin{enumerate}
\item
If $f$ and $g$ are 
$\Unc$-maps,
from
$\mathfrak{X}$
to
$\mathfrak{Y}$
and
from
$\mathfrak{Y}$
to
$\mathfrak{Z}$,
respectively,
then
$g \circ f$ is 
a $\Unc$-map from
$\mathfrak{X}$
to
$\mathfrak{Z}$
.

\item
If $f$ and $g$ are 
$\mpUnc$-maps,
from
$\mathfrak{X}$
to
$\mathfrak{Y}$
and
from
$\mathfrak{Y}$
to
$\mathfrak{Z}$,
respectively,
then
$g \circ f$ is 
an $\mpUnc$-map from
$\mathfrak{X}$
to
$\mathfrak{Z}$
.

\end{enumerate}

\end{prop}

\begin{defn}
\label{defn:catUnc}
\begin{enumerate}
\item
$\Unc$
is the category whose objects are
all uncertainty spaces
and
arrows between objects are 
$\Unc$-maps.

\item
$\mpUnc$
is the category whose objects are
all uncertainty spaces
and
arrows between objects are 
$\mpUnc$-maps.

\end{enumerate}

\end{defn}

Since
an $\mpUnc$-map is always a
$\Unc$-map,
the category $\mpUnc$
is a full subcategory of
$\Unc$.

In case both 
$U_X
	=\{ u\}
$
and
$U_Y
	=\{ v\}
$
are singleton sets,
and $u$ and $v$ are probability measures.
Then, a $\Unc$-map $f$ satisfies
$
u \circ f^{-1} \ll v
$,
which means that 
$f$ is a null-preserving map defined in
\cite{ANR_2020b}
and
\cite{AR_2019}.
Therefore the category $\Prob$ defined there is a subcategory of $\Unc$.

Similarly,
an $\mpUnc$-map $f$ satisfies
$
u \circ f^{-1} = v
$,
which means that 
$f$ is a measure-preserving function.
Therefore, the category $\mpProb$ defined in 
\cite{ANR_2020b}
is a subcategory of $\mpUnc$.

\begin{prop}
\label{prop:terminalObj}
The terminal uncertainty space
$
\mathfrak{T}
	=
(
	\{*\},
	\{ \emptyset, \{*\}\},
	\{ \bar{*}\}
)
$
is a terminal object in both
$\Unc$
and
$\mpUnc$.
\end{prop}
\begin{proof}
All we need to show is that the unique measurable map
$
! :
(X, \Sigma_X, U_X)
	\to
\mathfrak{T}
$
is an $\mpUnc$-map.
However, since
$
!^{-1}(\emptyset) = \emptyset
$
and
$
!^{-1}(\{*\}) = X
$,
we have
for any
$u \in U_X$,
\[
(u \, \circ \, !^{-1})(\emptyset)
	=
u(\emptyset)
	=
0
	=
\bar{*}(\emptyset)
		\quad \textrm{and} \quad
(u \, \circ \, !^{-1})(\{*\})
	=
u(X)
	=
1
	=
\bar{*}(\{*\}).
\]
Therefore
$
u \, \circ \, !^{-1}
	=
\bar{*} ,
$
which means that
$!$
is an $\mpUnc$-map.

\end{proof}

\subsection{U-sequences and $\Unc$-maps}
\label{sec:UseqAndUncMap}

In this section,
we will examines if the U-sequences appeared in 
the examples examining Ellsberg's paradox in
Section \ref{sec:hieraUncertain}
can be considered as a sequence whose components are connected by $\Unc$-maps.

Before going into individual cases, we will prepare the following proposition.

\begin{prop}
\label{prop:UncMapLemma}
Let
$
\mathfrak{X} = (X, \Sigma_{X}, U_{X})
$
and
$
\mathfrak{Y} = (Y, \Sigma_Y, U_{Y})
$
be two uncertainty spaces.
If there exists
$v \in U_Y$
such that
for every
$C \in \Sigma_{Y}$,
$v(C) = 0$
implies
$C = \emptyset$,
then
every measurable map
$
f :
\mathfrak{X}
	\to
\mathfrak{Y}
$
is a $\Unc$-map.
\end{prop}
\begin{proof}
For
$C \in \Sigma_{Y}$,
assume that
$v(C) = 0$.
Then,
$C = \emptyset$.
Therefore,
for any
$u \in U_X$,
\[
(u \circ f^{-1})(C)
	=
u (f^{-1}(C))
	=
u (f^{-1}(\emptyset))
	=
u (\emptyset)
	=
0,
\]
which means
$
u \circ f^{-1}
	\ll
v .
$
\end{proof}

First let us consider maps between 
$
\mathfrak{X}_0 = (X_0, \Sigma_{X_0}, U_{X_0})
$
and
$
\mathfrak{X}_1 = (X_1, \Sigma_{X_1}, U_{X_1})
$
where
$
X_0 = \{ R, B, Y\}
$
and
$
X_1 = \{
	u_k
\mid
	k = 0, \cdots, 2N
\}
$.

Define a map
$
f : X_0 \to X_1
$
by
\if \MkLonger 1
\[
f(R) := u_N, \;
f(B) := u_{2N}, \;
f(Y) := u_{0} .
\]
\else 
$
f(R) := u_N,
$
$
f(B) := u_{2N}
$
and
$
f(Y) := u_{0} .
$
\fi 
Since
$
U_{X_1}
	=
\{
v^{\mathrm{u}}
\}
$
where
$
v^{\mathrm{u}}
$
is a uniform distribution over a finite universe,
every non-empty subset
$C$
of
$X_1$
has non-zero mass
$
v^{\mathrm{u}}(C) 
$.
Therefore by 
Proposition \ref{prop:UncMapLemma},
$f$ becomes a $\Unc$-map.

Similarly,
any measurable map 
from
$
\mathfrak{X}_0
$
to
$
\mathfrak{Y}_1
$
or
from
$
\mathfrak{X}_0
$
to
$
\mathfrak{Z}_1
$
are $\Unc$-maps
since
both
$
U_{Y_1}
	=
\{
v^{\mathrm{b}}
\}
$
and 
$
U_{Z_1}
	=
\{
v_p
	\mid
p \in [0, 1]
\}
$
consists of
the binomial probability measures
that
satisfy the assumption of 
Proposition \ref{prop:UncMapLemma}
as long as
$p \in (0,1)$.

However, if we consider the measurable maps from 
$\mathfrak{Z}_1$
to
$
\mathfrak{Z}_2
$,
the situation has been changed.

\begin{prop}
\label{prop:UncMapLemmaLambda}
Let
$
\mathfrak{X} = (X, \Sigma_{X}, U_{X})
$
and
$
\mathfrak{Y} = (Y, \Sigma_Y, U_{Y})
$
be two uncertainty spaces
and
$
f :
\mathfrak{X}
	\to
\mathfrak{Y}
$
be a 1-1 measurable map.
If
$\{y\} \in U_Y$
and
$v(\{y\}) = 0$
for
every 
$v \in U_Y$
and
$y \in Y$,
and 
for every
$u \in U_X$
there exists
$x \in X$
such that
$
\{x\} \in U_X
$
and
$
u(\{x\})
	> 0
$,
then
$f$
is not a $\Unc$-map.
\end{prop}
\begin{proof}
For a given
$u \in U_X$,
let
$x_0 \in U_X$
with 
$
\{x_0\} \in U_X
$
and
$
u(\{x_0\})
	> 0
$, 
and define
$
y_0 := f(x_0).
$
Then 
$
f^{-1}(\{y_0\})
	=
\{ x_0\}
$
since
$f$
is 1-1.
Hence, we obtain
\[
(u \circ f^{-1})(\{y_0\})
	=
u(\{x_0\})
	> 0
\]
while
$v(\{y_0\}) = 0$.
Therefore
$
u \circ f^{-1}
$
is not absolutely continuous to
$v$,
which means
$f$
is not a $\Unc$-map.

\end{proof}

By Proposition \ref{prop:UncMapLemmaLambda},
any measurable 1-1 map
$
g : Z_1 \to Z_2
$
including the map
defined by
$
g(u_k) :=
v_{\frac{k}{2N}}
$
cannot be a $\Unc$-map.

\subsection{Embedding with a Dirac measure operator}
\label{sec:embEtaX}

Let
$
\mathfrak{X}
	=
(X, \Sigma_X, U_X)
$
be a given uncertainty space
throughout this section.
We will investigate a possible embedding of $X$ into $U_X$ by a $\Unc$-map.
As a candidate of an embedding operator, we will introduce a so-called Dirac measure.

\begin{defn}
\label{defn:xDot}
For
$x \in X$,
$
\dot{x} : \mathfrak{T} \to (X, \Sigma_X, U_X)
$
is a map defined by
$
\dot{x}(*) := x .
$
\end{defn}

Note that the map $\dot{x}$ is measurable since the $\sigma$-algebra of 
$\mathfrak{T}$ is the powerset.

\begin{prop}
\label{prop:xDot}
Let
$x \in X$,
and
$
\dot{x}
 : \mathfrak{T} \to (X, \Sigma_X, U_X)
$
be the function defined in
Definition \ref{defn:xDot}.
\begin{enumerate}
\item
$\dot{x}$
is an 
$\mpUnc$-map
if and only if
\;
$
\bar{*} \circ (\dot{x})^{-1}
	\in
U_X .
$

\item
$\dot{x}$
is a 
$\Unc$-map
if and only if
there exists
$
u \in U_X
$
such that
for any
$
A \in \Sigma_X
$
containing $x$,
$u(A)$
is strictly positive.
\end{enumerate}

\end{prop}
\begin{proof}
\begin{enumerate}
\item
Obvious from the definition of 
$\mpUnc$-maps.

\item
\begin{align*}
		&
\exists u \in U_X \,.\,
\bar{*} \circ (\dot{x})^{-1}
			 \ll u
		\\\Longleftrightarrow \; &
\exists u \in U_X \,.\,
\forall A \in \Sigma_X \,.\,
\big(
	u(A) = 0 \;\to\; 
	(\bar{*} \circ (\dot{x})^{-1})(A) = 0
\big)
		\\\Longleftrightarrow \; &
\exists u \in U_X \,.\,
\forall A \in \Sigma_X \,.\,
\big(
	u(A) = 0 \;\to\; 
	x \notin A
\big)
		\\\Longleftrightarrow \; &
\exists u \in U_X \,.\,
\forall A \in \Sigma_X \,.\,
\big(
	x \in A
		\; \to\; 
	u(A) > 0 
\big) .
\end{align*}

\end{enumerate}
\end{proof}

\begin{defn}
\label{defn:etaX}
Let
$
\mathfrak{X}
	=
(X, \Sigma_X)
$
be a measurable space.
For
$x \in X$,
define a map
$
\eta_{\mathfrak{X}}(x)
	:
\Sigma_X
	\to
[0,1]
$
by
\begin{equation}
\label{eq:etaX}
\eta_{\mathfrak{X}}(x)
	:=
\bar{*} \circ (\dot{x})^{-1} .
\end{equation}

\end{defn}

Note that
by Proposition \ref{prop:xDot} (1),
$
\dot{x}
$
is an $\mpUnc$-map
if and only if
$
\eta_{\mathfrak{X}}(x)
	\in
U_X
$ .

\begin{prop}
\label{prop:etaXxX}
For
a measurable space
$\mathfrak{X} = (X, \Sigma_X)$,
$x \in X$
and
$A \in \Sigma_X$,
we have
\begin{equation}
\label{eq:etaXxX}
\eta_{\mathfrak{X}}(x)(A)
	=
\mathbb{1}_X(A)(x) .
\end{equation}
In other words,
\begin{equation}
\label{eq:etaXxY}
\varepsilon_{\mathfrak{X}}(A) \circ \eta_{\mathfrak{X}}
	=
\mathbb{1}_X(A) .
\end{equation}
\end{prop}
\begin{proof}
\if \MkLonger  1
\begin{align*}
\eta_{\mathfrak{X}}(x)(A)
		&=
(\bar{*} \circ \dot{x}^{-1})(A)
		\\&=
\bar{*} ( \dot{x}^{-1}(A))
		\\&=
\begin{cases}
	\bar{*} (\{*\})
			& (\textrm{if } x \in A) ,
		\\
	\bar{*} (\emptyset)
			& (\textrm{if } x \notin A) 
\end{cases}
		\\&=
\begin{cases}
	1
			& (\textrm{if } x \in A) ,
		\\
	0
			& (\textrm{if } x \notin A) 
\end{cases}
		\\&=
\mathbb{1}_X(A)(x) .
\end{align*}
\else 
\begin{align*}
\eta_{\mathfrak{X}}(x)(A)
		&=
(\bar{*} \circ \dot{x}^{-1})(A)
		=
\bar{*} ( \dot{x}^{-1}(A))
		\\&=
\begin{cases}
	\bar{*} (\{*\})
			& (\textrm{if } x \in A)
		\\
	\bar{*} (\emptyset)
			& (\textrm{if } x \notin A) 
\end{cases}
		=
\begin{cases}
	1
			& (\textrm{if } x \in A) 
		\\
	0
			& (\textrm{if } x \notin A) 
\end{cases}
		=
\mathbb{1}_X(A)(x) .
\end{align*}
\fi 
\end{proof}

\begin{prop}
\label{prop:mmuXmeasurable}
For
an uncertainty space
$
\mathfrak{X} =
(X, \Sigma_X, U_X)
$,
$
\eta_{\mathfrak{X}}
	:
(X, \Sigma_X)
	\to
(U_X, \Sigma_{\mathfrak{X}})
$
is measurable.
\end{prop}
\begin{proof}
For
$
A \in \Sigma_X
$
and
$
Z \in \Sigma_{\mathbb{I}}
$,
\if \MkLonger  1
\begin{align*}
			&
\eta_{\mathfrak{X}}^{-1}
\big(
	\{
		u
			\in
		\mathfrak{S} X
	\mid
		u
		(A)
			\in
		Z
	\}
\big)
			\\&=
\big\{
	x \in X
\mid
	\eta_{\mathfrak{X}}(x)
		\in
	\{
		u
			\in
		\mathfrak{S} X
	\mid
		u
		(A)
			\in
		Z
	\}
\big\}
			\\&=
\big\{
	x \in X
\mid
	\eta_{\mathfrak{X}}
	(x)(A)
		\in
	Z
\big\}
			\\&=
\big\{
	x \in X
\mid
	\mathbb{1}_X(A)(x)
		\in
	Z
\big\}
					&(\textrm{
		by Proposition \ref{prop:etaXxX}
					})
			\\&=
(
	\mathbb{1}_X(A)
)^{-1}
(Z)
			\in
\Sigma_X .
					&(\textrm{
		by Proposition \ref{prop:charFunMeas}
					})
\end{align*}
\else 
by Proposition \ref{prop:etaXxX}
and
Proposition \ref{prop:charFunMeas},
we have
\begin{align*}
			&
\eta_{\mathfrak{X}}^{-1}
\big(
	\{
		u
			\in
		\mathfrak{S} X
	\mid
		u
		(A)
			\in
		Z
	\}
\big)
			=
\big\{
	x \in X
\mid
	\eta_{\mathfrak{X}}(x)
		\in
	\{
		u
			\in
		\mathfrak{S} X
	\mid
		u
		(A)
			\in
		Z
	\}
\big\}
			\\&=
\big\{
	x \in X
\mid
	\eta_{\mathfrak{X}}(x)(A)
		\in
	Z
\big\}
			=
\big\{
	x \in X
\mid
	\mathbb{1}_X(A)(x)
		\in
	Z
\big\}
			=
(
	\mathbb{1}_X(A)
)^{-1}
(Z)
			\in
\Sigma_X .
\end{align*}
\fi 
\end{proof}

Due to the following proposition,
we call
$
\eta_{\mathfrak{X}}(x)
$
a
\newword{Dirac measure}.

\begin{prop}
\label{prop:DiracMesX}
Let
$
\mathfrak{X}
	=
(X, \Sigma_X)
$
be a measurable space,
and
$
f \in
L^{\infty}(X) .
$
Then, for
$x \in X$,
we have
\begin{equation}
\label{eq:DiracMesEqX}
I_X^{
	\eta_{\mathfrak{X}}(x)
}
(f)
		=
f(x).
\end{equation}
\end{prop}
\begin{proof}
First we show
(\ref{eq:DiracMesEqX})
in the case when
$
f := 
\sum_{i=1}^n
	a_i
	\,
	\mathbb{1}_X
	(A_i) 
$
where
$
a_1 \ge a_2 \ge \cdots \ge a_n
$
are decreasing real numbers
and
$A_i$
are mutually disjoint elements of
$\Sigma_X$.
\if \MkLonger  1
\begin{align*}
I_X^{
	\eta_{\mathfrak{X}}(x)
}
(f)
		&=
I_X^{
	\eta_{\mathfrak{X}}(x)
}
\Big(
\sum_{i=1}^n
	a_i
	\,
	\mathbb{1}_X
	(A_i)
\Big)
		\\&=
\sum_{i=1}^n
	(a_i - a_{i+1})
	\,
	\eta_X(x)
	\Big(
	\bigcup_{j=1}^i
		A_j
	\Big)
		&(\textrm{
			by Lemma \ref{lem:stepFunChoquetInt}
		})
		\\&=
\sum_{i=1}^n
	(a_i - a_{i+1})
	\,
	\mathbb{1}_X
	\Big(
	\bigcup_{j=1}^i
		A_j
	\Big)
	(x)
		&(\textrm{
			by Proposition \ref{prop:etaXxX}
		})
		\\&=
\sum_{i=1}^n
	(a_i - a_{i+1})
	\sum_{j=1}^i
		\mathbb{1}_X
		(A_j)
		(x)
		\\&=
\sum_{j=1}^n
	\mathbb{1}_X
	(A_j)
	(x)
	\sum_{i=j}^n
		(a_i - a_{i+1})
		\\&=
\sum_{j=1}^n
	a_j
	\mathbb{1}_X
	(A_j)
	(x)
		=
f(x) .
\end{align*}
\else 
By 
			Lemma \ref{lem:stepFunChoquetInt}
and
			Proposition \ref{prop:etaXxX},
and by assuming
$
a_{n+1} = 0
$,
we obtain
\begin{align*}
I_X^{
	\eta_{\mathfrak{X}}(x)
}
(f)
		&=
I_X^{
	\eta_{\mathfrak{X}}(x)
}
\Big(
\sum_{i=1}^n
	a_i
	\,
	\mathbb{1}_X
	(A_i)
\Big)
		=
\sum_{i=1}^n
	(a_i - a_{i+1})
	\,
	\eta_{\mathfrak{X}}(x)
	\Big(
	\bigcup_{j=1}^i
		A_j
	\Big)
		\\&=
\sum_{i=1}^n
	(a_i - a_{i+1})
	\,
	\mathbb{1}_X
	\Big(
	\bigcup_{j=1}^i
		A_j
	\Big)
	(x)
		=
\sum_{i=1}^n
	(a_i - a_{i+1})
	\sum_{j=1}^i
		\mathbb{1}_X
		(A_j)
		(x)
		\\&=
\sum_{j=1}^n
	\mathbb{1}_X
	(A_j)
	(x)
	\sum_{i=j}^n
		(a_i - a_{i+1})
		=
\sum_{j=1}^n
	a_j
	\mathbb{1}_X
	(A_j)
	(x)
		=
f(x) .
\end{align*}
\fi 

Next,
for general
$
f \in L^{\infty}(Y),
$
we create 
$\underline{f}_n$
and
$\bar{f}_n$
by the method described in
Note \ref{note:discretizeU}.
By the monotonicity of
$
I_X^u
$
and the previous result about step functions, 
we have
\begin{equation}
\underline{f}_n
(x)
	=
I_X^{
	\eta_{\mathfrak{X}}(x)
}(
\underline{f}_n
)
	\le
I_X^{
	\eta_{\mathfrak{X}}(x)
}(
f
)
	\le
I_X^{
	\eta_{\mathfrak{X}}(x)
}(
\bar{f}_n
) 
	=
\bar{f}_n
(x) .
\end{equation}
But,
since
both
$
\underline{f}_n
(x)
$
and
$
\bar{f}_n
(x)
$
converge to
$
f(x),
$
we get
(\ref{eq:DiracMesEqX}).

\end{proof}

Now, in order to adopt the Dirac measure operator 
$
\eta_{\mathfrak{X}}
	:
X \to U_X
$
as an  embedding operator, we need at least the following assumption.

\begin{asm}
\label{asm:etaX}
An uncertainty space
$
\mathfrak{X}
	=
(X, \Sigma_X, U_X)
$
is said to satisfy 
the
\newword{embedding condition}
if for every
$x \in X$,
$
\eta_{\mathfrak{X}}(x)
	\in
U_X
$
holds.

\end{asm}

Here is a necessary condition of making  
$
\eta_{\mathfrak{X}}
$
be a $\Unc$-map.

\begin{prop}
\label{prop:embDirac}
Let
$
\mathfrak{X}
	=
(X, \Sigma_X, U_X)
$
and
$
\mathfrak{Y}
	=
(Y, \Sigma_Y, U_Y)
	:=
(U_X, \Sigma_{\mathfrak{X}}, U_Y)
$
be uncertainty spaces
where
$
\mathfrak{X}
$
satisfies the embedding condition.
If
$
\eta_{\mathfrak{X}}
	:
\mathfrak{X}
	\to
\mathfrak{Y}
$
is a $\Unc$-map, then it implies
for every
$u \in U_X$
there exists
$v \in U_Y$
such that
for all
$A \in \Sigma_X$,
the following three statements hold.
\begin{enumerate}
\item
$
v\big(\{
	w \in U_X
		\mid
	w(A) \in \{0,1\}
\}\big)
	> 0
$,
\item
$
u(X \setminus A) > 0
$
implies
$
v\big(\{
	w \in U_X
		\mid
	w(A) = 0
\}\big)
	> 0
$,
\item
$
u(A) > 0
$
implies
$
v\big(\{
	w \in U_X
		\mid
	w(A) = 1
\}\big)
	> 0
$.
\end{enumerate}

\end{prop}
\begin{proof}
Note that 
$\Sigma_{\mathfrak{X}}$
is generated by the sets of the form
\[
(\varepsilon_{\mathfrak{X}}(A))^{-1}(Z)
	=
\{
	w \in U_X
\mid
	w(A) \in Z
\}
\]
where
$A \in \Sigma_X$
and
$Z \in \beta([0,1])$.
Since
$
\eta_{\mathfrak{X}}
$
is a $\Unc$-map, it must satisfy
$
u \circ
\eta_{\mathfrak{X}}^{-1}
	\ll
v
$. 
Especially,
\if \MkLonger 1
\[
v(
(\varepsilon_{\mathfrak{X}}(A))^{-1}(Z)
)
 = 0
	\; \textrm{implies} \;
(
u \circ
\eta_{\mathfrak{X}}^{-1}
)(
(\varepsilon_{\mathfrak{X}}(A))^{-1}(Z)
)  = 0.
\]
\else 
$
v(
(\varepsilon_{\mathfrak{X}}(A))^{-1}(Z)
)
 = 0
$
implies
$
(
u \circ
\eta_{\mathfrak{X}}^{-1}
)(
(\varepsilon_{\mathfrak{X}}(A))^{-1}(Z)
)  = 0.
$
\fi 
Here we have by (\ref{eq:etaXxY}), 
\[
\eta_{\mathfrak{X}}^{-1}
(
(
\varepsilon_{\mathfrak{X}}(A))^{-1}(Z)
		=
(
\eta_{\mathfrak{X}}^{-1}
	\circ
\varepsilon_{\mathfrak{X}}(A))^{-1}
)
(Z)
		=
(
\varepsilon_{\mathfrak{X}}(A))
	\circ
\eta_{\mathfrak{X}}
)^{-1}
(Z)
		=
(
	\mathbb{1}_X(A)
)^{-1}
(Z)
\]
Therefore,
\[
(
u \circ
\eta_{\mathfrak{X}}^{-1}
)(
(\varepsilon_{\mathfrak{X}}(A))^{-1}(Z)
)
	=
\begin{cases}
0
	&\quad \textrm{if} \quad
0 \notin Z \; \land \;  1 \notin Z
	\\
u(X \setminus A)
	&\quad \textrm{if} \quad
0 \in Z \; \land \;  1 \notin Z
	\\
u(A)
	&\quad \textrm{if} \quad
0 \notin Z \; \land \;  1 \in Z
	\\
1
	&\quad \textrm{if} \quad
0 \in Z \; \land \;  1 \in Z ,
\end{cases}
\]
from which the result comes immediately.
\end{proof}

\section{Categories of U-sequences}
\label{sec:uCat}

In this section,
we introduce maps between U-sequences,
called $U^G$-maps, and they commute with Choquet expectation maps in a sense.
We will see that
U-sequences and $U^G$-maps forms a category.


\begin{defn}
\label{defn:UGmap}
Let
$
G :
\mathbb{R}
\to
\mathbb{R}
$
be a strictly increasing continuous measurable function,
and
$
	\mathfrak{X}
		=
	(X, \Sigma_{X}, U_{X})
$
be a uncertainty space.
Then, the map
$
\xi^G_{\mathfrak{X}}
	:
L^{\infty}(X, \Sigma_X)
	\to
L^{\infty}(U_X, \Sigma_{\mathfrak{X}})
$
is defined by
for
$
f
	\in
L^{\infty}(X, \Sigma_X)
$,
\begin{equation}
\label{eq:UGmap}
\xi^G_{\mathfrak{X}}(f)
	:=
G^{-1}
	\circ
\xi_{\mathfrak{X}}
(G \circ f) .
\end{equation}
\end{defn}

\begin{Note}
\label{Note:UGmap}
Let
$
\mathfrak{X}
	=
(X, \Sigma_X, U_X)
$
be a uncertainty space.
\begin{enumerate}
\item
If
$
G(x) := cx
$
with some positive constant
$c \in \mathbb{R}$,
then
$
	\xi^G_{\mathfrak{X}}
=
	\xi_{\mathfrak{X}}
$.

\item
If
$
G(x)
	:=
e^{\lambda x}
$
for
$\lambda > 0$,
then
for
$
f
	\in
L^{\infty}(\mathfrak{X})
$
and
$
u \in U_{X}
$,
\begin{equation}
\label{eq:entVMtwo}
\xi^G_{\mathfrak{X}}
(f)
(u)
	=
\frac{1}{\lambda}
\log
	I_X^u(
		e^{\lambda f}
	) .
\end{equation}
This is
an entropic value measure,
which is one of the major tools in monetary risk measure theory.
See Section 2 of \cite{adachi_2014crm} for the introduction to monetary risk measure theory.

\end{enumerate}
\end{Note}

Let
$
G :
\mathbb{R}
\to
\mathbb{R}
$
be a strictly increasing continuous measurable function
throughout this section.

\begin{defn}
\label{defn:Uarrow}
Let
$
\mathbb{X}
	=
\{
	\mathfrak{X}_n
		=
	(X_n, \Sigma_{X_n}, U_{X_n})
\}_{n \in \mathbb{N}}
$
and
$
\mathbb{Y}
	=
\{
	\mathfrak{Y}_n
		=
	(Y_n, \Sigma_{Y_n}, U_{Y_n})
\}_{n \in \mathbb{N}}
$
be two U-sequences.
A
\newword{$U^G$-map}
from
$\mathbb{X}$
to
$\mathbb{Y}$
is a set of measurable functions
$
\varphi
	=
\{
	\varphi_n
		:
	(X_n, \Sigma_{X_n}) 
		\to 
	(Y_n, \Sigma_{Y_n}) 
\}_{n \in \mathbb{N}}
$
such that
the following diagram commutes
for every
$n \in \mathbb{N}$.
\begin{equation*}
\xymatrix@C=30 pt@R=30 pt{
	L^{\infty}(\mathfrak{X}_n)
		\ar @{->}^{
			\xi^G_{
				\mathfrak{X}_n
			}
		} [r]
&
	L^{\infty}(\mathfrak{X}_{n+1})
\\
	L^{\infty}(\mathfrak{Y}_n)
		\ar @{->}^{
			\xi^G_{
				\mathfrak{Y}_n
			}
		} [r]
		\ar @{->}^{
			L^{\infty}
			(\varphi_n)
		} [u]
&
	L^{\infty}(\mathfrak{Y}_{n+1})
		\ar @{->}_{
			L^{\infty}
			(\varphi_{n+1})
		} [u]
}
\end{equation*}
\end{defn}

\begin{prop}
\label{prop:UarrowEquiv}
Let
$
\mathbb{X}
	=
\{
	\mathfrak{X}_n
		=
	(X_n, \Sigma_{X_n}, U_{X_n})
\}_{n \in \mathbb{N}}
$
and
$
\mathbb{Y}
	=
\{
	\mathfrak{Y}_n
		=
	(Y_n, \Sigma_{Y_n}, U_{Y_n})
\}_{n \in \mathbb{N}}
$
be U-sequences
and
$
\varphi
	=
\{
	\varphi_n
		:
	(X_n, \Sigma_{X_n}) 
		\to 
	(Y_n, \Sigma_{Y_n}) 
\}_{n \in \mathbb{N}}
$
be a
set of measurable functions.
Then the followings are equivalent.
\begin{enumerate}
\item
$\varphi$
is a 
$U^G$-map
from
$\mathbb{X}$
to
$\mathbb{Y}$.

\item
For every
$n \in \mathbb{N}$
and
$
f \in 
L^{\infty}(\mathfrak{Y}_n)
$,
\begin{equation}
\label{eq:UarrowXi}
\xi_{\mathfrak{X}_n}
(G \circ f \circ \varphi_n)
		=
\xi_{\mathfrak{Y}_n}
(G \circ f)
	\circ
\varphi_{n+1} .
\end{equation}

\item
For every
$n \in \mathbb{N}$,
$
f \in 
L^{\infty}(\mathfrak{Y}_n)
$
and
$
u \in U_{X_n}
$,
\begin{equation}
\label{eq:UarrowI}
I_{
	X_n
}^u
(G \circ f \circ \varphi_n)
		=
I_{
	Y_n
}^{
	\varphi_{n+1}(u)
}
(G \circ f) .
\end{equation}

\end{enumerate}
\end{prop}
\if \WithProof 1
\begin{proof}
\begin{equation*}
\xymatrix@C=15 pt@R=15 pt{
	f \circ \varphi_n
		\ar @{|->} [rrrr]
		\ar @{}_-{\mathrel{\rotatebox[origin=c]{-45}{$\in$} }} @<+6pt> [rd]
&&&&
\xi^G_{\mathfrak{X}_n}
(f \circ \varphi_n)
		=
\xi^G_{\mathfrak{Y}_n}
(f)
	\circ
\varphi_{n+1} 
		\ar @{}_-{\mathrel{\rotatebox[origin=c]{210}{$\in$} }} @<+6pt> [ld]
\\
&
	L^{\infty}(\mathfrak{X}_n)
		\ar @{->}^{
			\xi^G_{
				\mathfrak{X}_n
			}
		} [rr]
&&
	L^{\infty}(\mathfrak{X}_{n+1})
\\
\\
&
	L^{\infty}(\mathfrak{Y}_n)
		\ar @{->}_{
			\xi^G_{
				\mathfrak{Y}_n
			}
		} [rr]
		\ar @{->}^{
			L^{\infty}
			(\varphi_n)
		} [uu]
&&
	L^{\infty}(\mathfrak{Y}_{n+1})
		\ar @{->}_{
			L^{\infty}
			(\varphi_{n+1})
		} [uu]
\\
	f
		\ar @{|->} [uuuu]
		\ar @{|->} [rrrr]
		\ar @{}_-{\mathrel{\rotatebox[origin=c]{45}{$\in$} }} @<+6pt> [ru]
&&&&
	\xi^G_{
		\mathfrak{Y}_n
	}
	(f)
		\ar @{|->} [uuuu]
		\ar @{}_-{\mathrel{\rotatebox[origin=c]{150}{$\in$} }} @<+6pt> [lu]
}
\end{equation*}
\end{proof}
\fi 

If we write
\if \MkLonger  1
\begin{equation}
\label{eq:adjLphi}
\langle f, u \rangle^G_{\mathfrak{X}}
	:=
\xi^G_{\mathfrak{X}}(f)(u)
\end{equation}
\else 
$
\label{eq:adjLphi}
\langle f, u \rangle^G_{\mathfrak{X}}
	:=
\xi^G_{\mathfrak{X}}(f)(u)
$
\fi 
for
$f \in 
L^{\infty}(\mathfrak{X})
$
and
$
u \in
U_{X}
$,
then 
(\ref{eq:UarrowI})
will be represented as
\begin{equation}
\label{eq:adjXiG}
\langle 
	L^{\infty}(\varphi_n)(f),
	u 
\rangle^G_{\mathfrak{X}_n}
		=
\langle 
	f,
	\varphi_{n+1}(u)
\rangle^G_{\mathfrak{Y}_n} ,
\end{equation}
which means that 
$
	L^{\infty}(\varphi_n)
$
is a left adjoint of
$
	\varphi_{n+1}
$.

\begin{prop}
\label{prop:compUarrows}
Let
$
\mathbb{X}
	=
\{
	\mathfrak{X}_n
		=
	(X_n, \Sigma_{X_n}, U_{X_n})
\}_{n \in \mathbb{N}}
$,
$
\mathbb{Y}
	=
\{
	\mathfrak{Y}_n
		=
	(Y_n, \Sigma_{Y_n}, U_{Y_n})
\}_{n \in \mathbb{N}}
$
and
$
\mathbb{Z}
	=
\{
	\mathfrak{Z}_n
		=
	(Z_n, \Sigma_{Z_n}, U_{Z_n})
\}_{n \in \mathbb{N}}
$
be U-sequences.
Let
$
\varphi
	=
\{
	\varphi_n
		:
	(X_n, \Sigma_{X_n}) 
		\to 
	(Y_n, \Sigma_{Y_n}) 
\}_{n \in \mathbb{N}}
$
and
$
\psi
	=
\{
	\psi_n
		:
	(Y_n, \Sigma_{Y_n}) 
		\to 
	(Z_n, \Sigma_{Z_n}) 
\}_{n \in \mathbb{N}}
$
be $U^G$-maps.
Then,
\if \MkLonger  1
\begin{equation}
\label{eq:compUarrows}
\psi \circ \varphi
		:=
\{
	\psi_n \circ \varphi_n
\}_{n \in \mathbb{N}}
\end{equation}
\else 
$
\psi \circ \varphi
		:=
\{
	\psi_n \circ \varphi_n
\}_{n \in \mathbb{N}}
$
\fi 
is a $U^G$-map
from
$\mathbb{X}$
to
$\mathbb{Z}$.
\end{prop}
\if \WithProof 1
\begin{proof}
By
Proposition \ref{prop:UarrowEquiv},
it is enough to show that
\begin{equation}
\label{eq:UarrowXiComp}
\xi_{\mathfrak{X}_n}
(G \circ f \circ (\psi_n \circ \varphi_n))
		=
\xi_{\mathfrak{Z}_n}
(G \circ f)
	\circ
(
\psi_{n+1} 
	\circ
\varphi_{n+1} 
) 
\end{equation}
for
every
$n \in \mathbb{N}$
and
$
f \in 
L^{\infty}(\mathfrak{Z}_n)
$.
But by using
(\ref{eq:UarrowXi})
twice,
we have
\if \MkLonger  1
\begin{align*}
\xi_{\mathfrak{X}_n}
(G \circ f \circ (\psi_n \circ \varphi_n))
			&=
\xi_{\mathfrak{X}_n}
\big(
	G \circ (f \circ \psi_n) \circ \varphi_n)
\big)
			\\&=
\xi_{\mathfrak{Y}_n}
\big(
	G \circ (f \circ \psi_n)
\big)
	\circ
\varphi_{n+1}
			\\&=
\big(
	\xi_{\mathfrak{Z}_n}
	(G \circ f)
		\circ
	\psi_{n+1} 
\big)
	\circ
\varphi_{n+1} 
			\\&=
\xi_{\mathfrak{Z}_n}
(G \circ f)
	\circ
(
\psi_{n+1} 
	\circ
\varphi_{n+1} 
) .
\end{align*}
\else 
\begin{align*}
\xi_{\mathfrak{X}_n}
(G \circ f \circ (\psi_n \circ \varphi_n))
			&=
\xi_{\mathfrak{X}_n}
\big(
	G \circ (f \circ \psi_n) \circ \varphi_n)
\big)
			=
\xi_{\mathfrak{Y}_n}
\big(
	G \circ (f \circ \psi_n)
\big)
	\circ
\varphi_{n+1}
			\\&=
\big(
	\xi_{\mathfrak{Z}_n}
	(G \circ f)
		\circ
	\psi_{n+1} 
\big)
	\circ
\varphi_{n+1} 
			=
\xi_{\mathfrak{Z}_n}
(G \circ f)
	\circ
(
\psi_{n+1} 
	\circ
\varphi_{n+1} 
) .
\end{align*}
\fi 
\end{proof}
\fi 

\begin{defn}
\label{defn:uCategory}
The 
\newword{$U^G$-category}
is a category 
$\USeq^G$
whose objects are
all U-sequences
and arrows are all $U^G$-maps.
\end{defn}

\begin{prop}
\label{prop:inclUseq}
Let
$
\mathbb{X}
	=
\{
	\mathfrak{X}_n
		=
	(X_n, \Sigma_{X_n}, U_{X_n})
\}_{n \in \mathbb{N}}
$
and
$
\mathbb{Y}
	=
\{
	\mathfrak{Y}_n
		=
	(Y_n, \Sigma_{Y_n}, U_{Y_n})
\}_{n \in \mathbb{N}}
$
be two U-sequences.
Assume that
for some 
$m \in \mathbb{N}$,
\if \MkLonger  1
\[
X_m = Y_m,
	\quad
\Sigma_{X_m}
	=
\Sigma_{Y_m},
	\quad
U_{X_m}
	\subset
U_{Y_m} .
\]
\else 
$
X_m = Y_m
$,
$
\Sigma_{X_m}
	=
\Sigma_{Y_m}
$
and
$
U_{X_m}
	\subset
U_{Y_m} .
$
\fi 
Then,
by defining
$
\varphi_m := \mathrm{Id}_{X_m} : X_m \to Y_m
$
and
$
\varphi_{m+1} : U_{X_m} \to U_{Y_m}
$
be an inclusion map,
we have
\begin{equation}
\label{eq:inclUseq}
\xi^G_{\mathfrak{X}_m}
	\circ
L^{\infty}(\varphi_m)
		=
L^{\infty}(\varphi_{m+1})
	\circ
\xi^G_{\mathfrak{Y}_m} .
\end{equation}

\end{prop}
\if \WithProof 1
\begin{proof}
By
Proposition \ref{prop:UarrowEquiv},
it is enough to show that
for any
$
f \in
L^{\infty}(
	\mathfrak{Y}_m
)
$,
\[
\xi_{\mathfrak{X}_m}
(G \circ f \circ \varphi_m)
		=
\xi_{\mathfrak{Y}_m}
(G \circ f)
	\circ
\varphi_{m+1} .
\]
But,
\if \MkLonger  0
by noting that
$\varphi_m = \textrm{Id}_{X_m}$,
$
\xi_{\mathfrak{X}_m}
(f)
	=
\xi_{\mathfrak{Y}_m}
(f)
|_{U_{X_m}}
$
and
$
\varphi_{m+1}(u) = u
$,
\fi 
for every
$
u \in U_{X_m}
= X_{m+1}
$,
we have
\if \MkLonger  1
\begin{align*}
\xi_{\mathfrak{X}_m}
(G \circ f \circ \varphi_m)
(u)
		&=
\xi_{\mathfrak{X}_m}
(G \circ f)
(u)
			&(\textrm{since $\varphi_m = \textrm{Id}_{X_m}$})
		\\&=
\xi_{\mathfrak{Y}_m}
(G \circ f)
(u)
			&(\textrm{since $
\xi_{\mathfrak{X}_m}
(f)
	=
\xi_{\mathfrak{Y}_m}
(f)
|_{U_{X_m}}
$})
		\\&=
\xi_{\mathfrak{Y}_m}
(G \circ f)
(\varphi_{m+1}(u))
			&(\textrm{since $
\varphi_{m+1}(u) = u
$})
		\\&=
\big(
\xi_{\mathfrak{Y}_m}
(G \circ f)
	\circ
\varphi_{m+1} 
\big)
(u) .
\end{align*}
\else 
\begin{align*}
\xi_{\mathfrak{X}_m}
(G \circ f \circ \varphi_m)
(u)
		&=
\xi_{\mathfrak{X}_m}
(G \circ f)
(u)
		=
\xi_{\mathfrak{Y}_m}
(G \circ f)
(u)
		\\&=
\xi_{\mathfrak{Y}_m}
(G \circ f)
(\varphi_{m+1}(u))
		=
\big(
\xi_{\mathfrak{Y}_m}
(G \circ f)
	\circ
\varphi_{m+1} 
\big)
(u) .
\end{align*}
\fi 

\end{proof}
\fi 

\begin{exmp}
In 
Section \ref{sec:secLayerAna}
and
Section \ref{sec:thiLayerAna},
we demonstrate three concrete U-sequences
$\mathbb{X}$,
$\mathbb{Y}$
and
$\mathbb{Z}$.
Among them,
we have a
$U^G$-map
$
\varphi
	=
\{
	\varphi_n
\}_{n \in \mathbb{N}}
$
from 
$\mathbb{Y}$
to
$\mathbb{Z}$,
where
all
$\varphi_n$
is the identity map
$
\mathrm{Id}_{Y_n}
$
except the case
$n = 2$
in which
\if \MkLonger  1
\[
\varphi_2
	:
Y_2 = \{*\}
	\to
Z_2 = U_{Z_1}
	=
\{
	v_p
\mid
	p \in [0,1]
\}
\]
\else 
$
\varphi_2
	:
Y_2 = \{*\}
	\to
Z_2 = U_{Z_1}
	=
\{
	v_p
\mid
	p \in [0,1]
\}
$
\fi 
is the inclusion map defined by
$
\varphi_2(*) := v_{\frac{1}{2}}
$.

\end{exmp}

\if \InclCMFun 1 
\section{CM-functors}
\label{sec:functorS}

In this section,
we will introduce 
a concept of C-functor
which is 
an endofunctor
$\mathfrak{S}$
of
$\Mble$,
the category of measurable spaces and measurable maps between them,
defined through
the lift-up functor
$
\mathfrak{L}
	:
\mpUnc
	\to
\Mble
$.
We also introduce two natural transformations 
$\eta^{\mathfrak{S}}$
and
$\mu^{\mathfrak{S}}$
over 
$\mathfrak{S}$.
In the final subsection,
We introduce
a concept of
CM-functors
which is an endofunctor of
$\Mble$,
and then
prove
the existence of the CM-functor that encompasses a given set of uncertainty spaces,
especially a given U-sequence.

\subsection{C-functors}
\label{sec:conditionsCM}

\begin{thm}
\label{thm:functorL}
For a uncertainty space
$
\mathfrak{X}
	=
(X, \Sigma_X, U_X)
$,
define a measurable space
$
\mathfrak{L}\mathfrak{X}
$
by
\begin{equation}
\label{eq:defLobj}
\mathfrak{L}\mathfrak{X}
	:=
(U_X, \Sigma_{\mathfrak{X}}) ,
\end{equation}
and
for an $\mpUnc$-map
$
h : 
\mathfrak{X}
	\to
\mathfrak{Y}
	= (Y, \Sigma_Y, U_Y)
$,
define a map
$
\mathfrak{L}h
	:
\mathfrak{L}\mathfrak{X}
	\to
\mathfrak{L}\mathfrak{Y}
$
by
\begin{equation}
\label{eq:defLmap}
\mathfrak{L}h(u)
	:=
u \circ h^{-1}
\end{equation}
for 
$u \in U_X$.
Then,
the correspondence
$\mathfrak{L}$
becomes a functor
from
$\mpUnc$
to
$\Mble$.

\end{thm}
\begin{proof}
Note that
(\ref{eq:defLmap})
is well-defined because
$h$ is an $\mpUnc$-map.

First, we will show that the map
$
\mathfrak{L}h
$
is a measurable map.
But to this end,
all we need to show is
by Proposition \ref{prop:sigmaGX},
$
(\mathfrak{L}h)^{-1}
(
	\{
		v
			\in
		U_Y
	\mid
		v(B)
			\in
		Z
	\}
)
	\in
\Sigma_{\mathfrak{X}}
$
for all
$B \in \Sigma_Y$
and
$
Z \in \Sigma_{\mathbb{I}}
$.
However, 
\if \MkLonger  1
we have
\begin{align*}
		&
(\mathfrak{L}h)^{-1}
(
	\{
		v
			\in
		U_Y
	\mid
		v(B)
			\in
		Z
	\}
)
		\\&=
\{
	u
\in
	U_X
\mid
	\mathfrak{L}h
	(u)
\in
	\{
		v
			\in
		U_Y
	\mid
		v(B)
			\in
		Z
	\}
\}
		\\&=
\{
	u
\in
	U_X
\mid
	u
		\circ
	h^{-1}
\in
	\{
		v
			\in
		U_Y
	\mid
		v(B)
			\in
		Z
	\}
\}
		\\&=
\{
	u
\in
	U_X
\mid
	(u
		\circ
	h^{-1})
	(B)
		\in
	Z
\}
		\\&=
\{
	u
\in
	U_X
\mid
	u(
	h^{-1}
	(B))
		\in
	Z
\}
		\in
\Sigma_{\mathfrak{X}} .
			&(\textrm{
by Proposition \ref{prop:sigmaGX}
			})
\end{align*}
\else 
we have
again
by Proposition \ref{prop:sigmaGX},
\begin{align*}
		&
(\mathfrak{L}h)^{-1}
(
	\{
		v
			\in
		U_Y
	\mid
		v(B)
			\in
		Z
	\}
)
		=
\{
	u
\in
	U_X
\mid
	\mathfrak{L}h
	(u)
\in
	\{
		v
			\in
		U_Y
	\mid
		v(B)
			\in
		Z
	\}
\}
		\\&=
\{
	u
\in
	U_X
\mid
	u
		\circ
	h^{-1}
\in
	\{
		v
			\in
		U_Y
	\mid
		v(B)
			\in
		Z
	\}
\}
		=
\{
	u
\in
	U_X
\mid
	(u
		\circ
	h^{-1})
	(B)
		\in
	Z
\}
		\\&=
\{
	u
\in
	U_X
\mid
	u(
	h^{-1}
	(B))
		\in
	Z
\}
		\in
\Sigma_{\mathfrak{X}} .
\end{align*}
\fi 

Next,
we will show that
$
\mathfrak{L}(j \circ h)
	=
(\mathfrak{L}j)
	\circ
(\mathfrak{L}h)
$
for all
pair of $\mpUnc$-maps

\noindent
$
\xymatrix{
	\mathfrak{X}
		\ar @{->}^h [r]
&
	\mathfrak{Y}
		\ar @{->}^-j [r]
&
	\mathfrak{Z}
		=
	(Z, \Sigma_Z, U_Z)
}
$
(
See Diagram \ref{diag:functorL}
).
But for
$
u \in U_X
$,
we have
\if \MkLonger  1
\begin{align*}
(\mathfrak{L}(j \circ h))
u
		&=
u
 \circ (j \circ h)^{-1}
		\\&=
u
 \circ h^{-1} \circ j^{-1}
		\\&=
(\mathfrak{L} h)(
u
) \circ j^{-1}
		\\&=
(\mathfrak{L}j)
(
(\mathfrak{L}h)
u
) 
		\\&=
(
(\mathfrak{L}j)
	\circ
(\mathfrak{L}h)
)
u .
\end{align*}
\else 
\begin{align*}
(\mathfrak{L}(j \circ h))
u
		&=
u
 \circ (j \circ h)^{-1}
		=
u
 \circ h^{-1} \circ j^{-1}
		\\&=
(\mathfrak{L} h)(
u
) \circ j^{-1}
		=
(\mathfrak{L}j)
(
(\mathfrak{L}h)
u
) 
		=
(
(\mathfrak{L}j)
	\circ
(\mathfrak{L}h)
)
u .
\end{align*}
\fi 

\end{proof}

\begin{diagram}[htb]
\begin{equation*}
\xymatrix{
   \mpUnc
      \ar @{->}^{\mathfrak{L}} [rrr]
&&&
   \Mble
\\
   \mathfrak{X}
      \ar @{->}^{h} [d]
      \ar @(ur,dr)[]^{1_{\mathfrak{X}}}
      \ar @/_2pc/_{
         j \circ h
      } [dd]
&&&
   \mathfrak{L}X
      \ar @(ul,dl)[]_{
         \mathfrak{L}1_{\mathfrak{X}} = 1_{\mathfrak{L}\mathfrak{X}}
      }
      \ar @{->}_{\mathfrak{L}h} [d]
      \ar @/^2pc/^{
         \mathfrak{L}(j \circ h) = (\mathfrak{L}j) \circ (\mathfrak{L}h)
      } [dd]
\\
   \mathfrak{Y}
      \ar @{->}^j [d]
      \ar @(ur,dr)[]^{1_{\mathfrak{Y}}}
&&&
   \mathfrak{L}Y
      \ar @(ul,dl)[]_{
         \mathfrak{L}1_{\mathfrak{Y}} = 1_{\mathfrak{L}\mathfrak{Y}}
      }
      \ar @{->}_{\mathfrak{L}j} [d]
\\
   \mathfrak{Z}
      \ar @(ur,dr)[]^{1_{\mathfrak{Z}}}
&&&
   \mathfrak{L}Z
      \ar @(ul,dl)[]_{
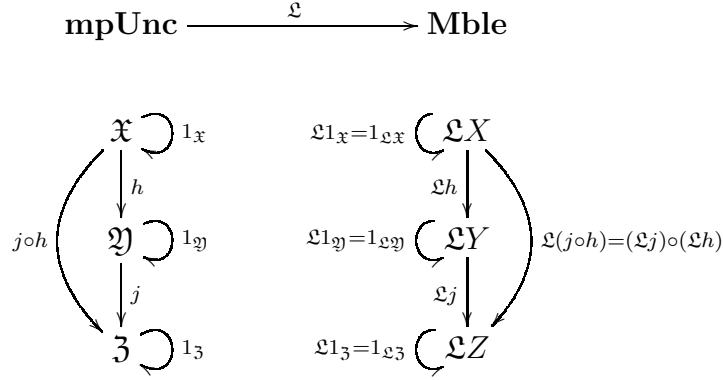
         \mathfrak{L}1_{\mathfrak{Z}} = 1_{\mathfrak{L}\mathfrak{Z}}
      }
}
\end{equation*}
\caption{Lift-up functor $\mathfrak{L}$}
\label{diag:functorL}
\end{diagram}

\begin{defn}
\label{defn:cFunctor}
A functor
$
\mathfrak{S}
	:
\Mble
	\to
\Mble
$
is called a 
\newword{C-functor}
if there exists a functor
$
\rho
	:
\Mble
	\to
\mpUnc
$
(called a 
\newword{choice functor}
)
such that
\begin{equation}
\label{eq:cfuntorDef}
\mathfrak{U}
	\circ
\rho
	=
1_{\Mble}
		\; \textrm{and} \;
\mathfrak{S}
	=
\mathfrak{L}
	\circ
\rho
\end{equation}
where
$
\mathfrak{U}
	:
\mpUnc \to \Mble
$
is the forgetful functor
that maps
$(X, \Sigma_X, U_X)$
to
$(X, \Sigma_X)$.
\end{defn}

\begin{diagram}[htb]
\begin{equation*}
\xymatrix{
	\Mble
		\ar @{->}^{\mathfrak{S}} [r]
		\ar @{->}_{1_{\Mble}} [d]
		\ar @{.>}^{\rho} [rd]
&
	\Mble
\\
	\Mble
&
	\mpUnc
		\ar @{->}_{\mathfrak{L}} [u]
		\ar @{->}^{\mathfrak{U}} [l]
}
\end{equation*}
\caption{C-functor $\mathfrak{S}$}
\label{diag:cFunctor}
\end{diagram}
We sometimes write
$
\mathfrak{S}_{\rho}
$
for this $\mathfrak{S}$
in order to clarify the choice functor $\rho$.

Here are two trivial examples of C-functors.

\begin{prop}
\label{prop:cFunctorsCandP}
Define two functors
$\mathfrak{C}$
and
$\mathfrak{P}$
from
$\Mble$
to
$\mpUnc$
by
for every measurable space 
$
\mathfrak{X}
	=
(X, \Sigma_X)
$,
\begin{enumerate}
\item
$
\mathfrak{C}
\mathfrak{X}
	:=
(X, \Sigma_X, C_{X})
$,
where
$C_X$
is the set of all capacities defined on 
$\mathfrak{X}$,

\item
$
\mathfrak{P}
\mathfrak{X}
	:=
(X, \Sigma_X, P_{X})
$,
where
$P_X$
is the set of all probability measures defined on 
$\mathfrak{X}$.
\end{enumerate}
Then, both
$
\mathfrak{S}_{\mathfrak{C}}
$
and
$
\mathfrak{S}_{\mathfrak{P}}
$
are C-functors.
\end{prop}

We frequently use 
$
\mathfrak{S}
(X, \Sigma_X) 
$
for denoting 
the underlying set
of
$
\mathfrak{S}
(X, \Sigma_X) 
$
below.

\begin{defn}
\label{defn:epsilonXiDetail}
Let
$\mathfrak{S}$
be a C-functor
and
$
\mathfrak{X}
	=
(X, \Sigma_X)
$
be a measurable space.
We use the following abbreviations.
\begin{align}
\label{eq:epsilonDetail}
\varepsilon^{\mathfrak{S}}_{\mathfrak{X}}
		&:=
\varepsilon_{(X, \Sigma_X, \mathfrak{S}{\mathfrak{X}})},
		\\
\label{eq:xiDetail}
\xi^{\mathfrak{S}}_{\mathfrak{X}}
		&:=
\xi_{(X, \Sigma_X, \mathfrak{S}{\mathfrak{X}})}.
\end{align}
\end{defn}


\begin{prop}
\label{prop:epsBcircGf}
Let
$\mathfrak{S}$
be a C-functor,
$\mathfrak{X} = (X, \Sigma_X)$
and
$\mathfrak{Y} = (Y, \Sigma_Y)$
be measurable spaces,
$
h : \mathfrak{X} \to \mathfrak{Y}
$
be a measurable map,
and
$
B \in \Sigma_Y
$.
Then, we have
\begin{equation}
\label{eq:epsBcircGf}
\varepsilon^{\mathfrak{S}}_{\mathfrak{Y}}(B)
	\circ
\mathfrak{S}
h
		=
\varepsilon^{\mathfrak{S}}_{\mathfrak{X}}
({h^{-1}(B)}) .
\end{equation}
\end{prop}
\begin{proof}
For
$
u
	\in
\mathfrak{S}
X
$,
\if \MkLonger  1
\begin{align*}
(
\varepsilon^{\mathfrak{S}}_{\mathfrak{Y}}
(B)
	\circ
\mathfrak{S}
h
)
(u)
			&=
\varepsilon^{\mathfrak{S}}_{\mathfrak{Y}}
(B)
(
	\mathfrak{S}
	h
	(u)
)
			\\&=
\varepsilon^{\mathfrak{S}}_{\mathfrak{Y}}
(B)
(
	u
		\circ
	h^{-1}
)
			\\&=
(
	u
		\circ
	h^{-1}
)
(B)
			\\&=
u
(
	h^{-1}
	(B)
)
			\\&=
\varepsilon^{\mathfrak{S}}_{\mathfrak{X}}
({h^{-1}(B)}) 
(u) .
\end{align*}
\else 
\begin{align*}
(
\varepsilon^{\mathfrak{S}}_{\mathfrak{Y}}
(B)
	\circ
\mathfrak{S}
h
)
(u)
			&=
\varepsilon^{\mathfrak{S}}_{\mathfrak{Y}}
(B)
(
	\mathfrak{S}
	h
	(u)
)
			=
\varepsilon^{\mathfrak{S}}_{\mathfrak{Y}}
(B)
(
	u
		\circ
	h^{-1}
)
			\\&=
(
	u
		\circ
	h^{-1}
)
(B)
			=
u
(
	h^{-1}
	(B)
)
			=
\varepsilon^{\mathfrak{S}}_{\mathfrak{X}}
({h^{-1}(B)}) 
(u) .
\end{align*}
\fi 
\end{proof}


\begin{defn}
\label{defn:DiracEta}
Let
$\mathfrak{S}$
be a C-functor
and
$
\mathfrak{X}
	=
(X, \Sigma_X)
$
be a measurable space.
Define a map
$
\eta^{\mathfrak{S}}_{\mathfrak{X}}
	:
\mathfrak{X}
	\to
\mathfrak{S}
\mathfrak{X}
$
by
\begin{equation}
\label{eq:DiracEta}
\eta^{\mathfrak{S}}_{\mathfrak{X}}
	:=
\eta_{(X, \Sigma_X, \mathfrak{S}\mathfrak{X})} .
\end{equation}
\end{defn}


\begin{prop}
\label{prop:etaNT}
For a C-functor
$\mathfrak{S}$,
the correspondence
$\eta^{\mathfrak{S}}$
is a natural transformation
$
\eta^{\mathfrak{S}}
	:
1_{\Mble}
	\dot{\to}
\mathfrak{S}
$.
i.e.
the diagram
\ref{diag:natTransEta}
commutes.

\begin{diagram}[htb]
\begin{equation*}
\xymatrix{
&
	1_{\Mble}
		\ar @{->}^{
\eta^{\mathfrak{S}}
} [r]
&
	\mathfrak{S}
\\
    \mathfrak{X}
      \ar @{->}_{h} [d]
  &
    \mathfrak{X}
      \ar @{->}_{h} [d]
      \ar @{->}^{
\eta^{\mathfrak{S}}_{\mathfrak{X}}
} [r]
  &
    \mathfrak{S}\mathfrak{X}
      \ar @{->}^{\mathfrak{S}h} [d]
 \\
    \mathfrak{Y}
  &
    \mathfrak{Y}
      \ar @{->}_{
\eta^{\mathfrak{S}}_{\mathfrak{Y}}
} [r]
  &
    \mathfrak{S}\mathfrak{Y}
}
\end{equation*}
\caption{natural transformation $
\eta^{\mathfrak{S}}
	:
1_{\Mble}
	\dot{\to}
\mathfrak{S}
$}
\label{diag:natTransEta}
\end{diagram}
\end{prop}
\begin{proof}
Let
$\mathfrak{X} = (X, \Sigma_X)$
and
$\mathfrak{Y} = (Y, \Sigma_Y)$
be two measurable spaces,
$h :
\mathfrak{X}
	\to
\mathfrak{Y}
$
be a measurable map,
$x \in X$
and
$
B \in
\Sigma_Y
$.
Then, 
\if \MkLonger  1
we have
\begin{align*}
(
\mathfrak{S} h
	\circ
\eta^{\mathfrak{S}}_{\mathfrak{X}}
)
(x)(B)
		&=
(
\mathfrak{S} h
)
(
\eta^{\mathfrak{S}}_{\mathfrak{X}}
 (x))
(B)
		\\&=
(
\eta^{\mathfrak{S}}_{\mathfrak{X}}
(x)
	\circ
h^{-1}
)
(B)
		\\&=
\eta^{\mathfrak{S}}_{\mathfrak{X}}
 (x)
(
h^{-1}
(B)
)
		\\&=
\mathbb{1}_X(h^{-1}(B))(x)
		&(\textrm{
			by Proposition \ref{prop:etaXxX}
		})
		\\&=
(\mathbb{1}_Y(B) \circ h)(x)
		&(\textrm{
by Proposition \ref{prop:lemOne}
		})
		\\&=
\mathbb{1}_Y(B) (h(x))
		\\&=
\eta_Y
(
h
(x)
)(B) 
		&(\textrm{
			by Proposition \ref{prop:etaXxX}
		})
		\\&=
(
\eta^{\mathfrak{S}}_{\mathfrak{Y}}
	\circ
h)
(x)(B) .
\end{align*}
\else 
we have
by
Proposition \ref{prop:etaXxX},
\ref{prop:lemOne}
and
\ref{prop:etaXxX},
\begin{align*}
(
\mathfrak{S} h
	\circ
\eta^{\mathfrak{S}}_{\mathfrak{X}}
)
(x)(B)
		&=
(
\mathfrak{S} h
)
(
\eta^{\mathfrak{S}}_{\mathfrak{X}}
(x))
(B)
		=
(
\eta^{\mathfrak{S}}_{\mathfrak{X}}
(x)
	\circ
h^{-1}
)
(B)
		=
\eta^{\mathfrak{S}}_{\mathfrak{X}}
(x)
(
h^{-1}
(B)
)
		\\&=
\mathbb{1}_X(h^{-1}(B))(x)
		=
(\mathbb{1}_Y(B) \circ h)(x)
		=
\mathbb{1}_Y(B) (h(x))
		\\&=
\eta^{\mathfrak{S}}_{\mathfrak{Y}}
(
h
(x)
)(B) 
		=
(
\eta^{\mathfrak{S}}_{\mathfrak{Y}}
	\circ
h)
(x)(B) .
\end{align*}
\fi 
\end{proof}

\subsection{CM-functors}
\label{sec:ntMu}

\begin{defn}
\label{defn:condM}
A C-functor
$\mathfrak{S}$
is called a
\newword{CM-functor}
if it satisfies
for every
measurable space
$
\mathfrak{X} = (X, \Sigma_X)
$
and
$
v
	\in
\mathfrak{S}
\mathfrak{S}
\mathfrak{X} ,
$
\begin{equation}
\label{eq:condM}
I_{
	\mathfrak{S}
	\mathfrak{X}
}^{
	v
}
	\circ
\varepsilon^{\mathfrak{S}}_{\mathfrak{X}}
		\in
\mathfrak{S}
\mathfrak{X} .
\end{equation}

\end{defn}
Here we have used the following convention.
\begin{equation}
\mathfrak{S}
\mathfrak{S}
\mathfrak{X}
	:=
\mathfrak{S}
(
\mathfrak{S}
\mathfrak{X}
)
	=
(
\mathfrak{S}
	\circ
\mathfrak{S}
)
\mathfrak{X}
	=:
\mathfrak{S}^2
\mathfrak{X} .
\end{equation}

\begin{exmp}
\label{exmp:corresCM}
Both
$
\mathfrak{S}_{\mathfrak{C}}
$
and
$
\mathfrak{S}_{\mathfrak{P}}
$
are CM-functors.
\end{exmp}

Let
$\mathfrak{S}$
be a CM-functor throughout the rest of this section.

\begin{defn}
\label{defn:muX}
For a measurable space
$
\mathfrak{X}
$,
define a map
$
\mu^{\mathfrak{S}}_{\mathfrak{X}}
	:
\mathfrak{S}^2
\mathfrak{X}
	\to
\mathfrak{S}
\mathfrak{X}
$
by
for
$
v
	\in
\mathfrak{S}^2
\mathfrak{X}
$,
\begin{equation}
\label{eq:muX}
\mu^{\mathfrak{S}}_{\mathfrak{X}}
(
	v
)
	:=
I_{
	\mathfrak{S}
	\mathfrak{X}
}^{
	v
}
	\circ
\varepsilon^{\mathfrak{S}}_{\mathfrak{X}} .
\end{equation}
\end{defn}

\begin{lem}
\label{lem:alterMu}
For 
a measurable space
$
\mathfrak{X}
	=
(X, \Sigma_X)
$,
$
v
	\in
\mathfrak{S}^2
\mathfrak{X}
$
and
$
A
	\in
\Sigma_X
$,
\begin{equation}
\label{eq:alterMu}
\mu^{\mathfrak{S}}_\mathfrak{X}
(v)
(A)
		=
I_{
	\mathfrak{S}
	\mathfrak{X}
}^{
	v
}
(
	\varepsilon^{\mathfrak{S}}_{\mathfrak{X}}
	(A)
) 
		=
\xi^{\mathfrak{S}}_{
	\mathfrak{S}
	\mathfrak{X}
}
(
	\varepsilon^{\mathfrak{S}}_{\mathfrak{X}}
	(A)
)
(v)
		=
\int_0^1
	v\big(
		\{
			u \in
			\mathfrak{S}
			\mathfrak{X}
		\mid
			u(A) \ge r
		\}
	\big)
\,
dr .
\end{equation}
\end{lem}
\begin{proof}
It is immediate from the fact that
\[
(
	\varepsilon^{\mathfrak{S}}_{\mathfrak{X}}
	(A)
		\ge
	r
)
		=
\{
	u \in 
	\mathfrak{S}
	\mathfrak{X}
\mid
	\varepsilon^{\mathfrak{S}}_{\mathfrak{X}}
	(A)(u)
	\ge r
\}
		=
\{
	u \in
	\mathfrak{S}
	\mathfrak{X}
\mid
	u(A)
	\ge r
\}
\]
and that
$
0 \le u(A) \le 1 .
$
\end{proof}

\begin{prop}
For 
a measurable space
$
\mathfrak{X}
$,
the map
$
\mu^{\mathfrak{S}}_{\mathfrak{X}}
	:
\mathfrak{S}^2
\mathfrak{X}
	\to
\mathfrak{S}
\mathfrak{X}
$
is measurable.
\end{prop}
\begin{proof}
Let
$
\mathfrak{X}
	= (X, \Sigma_X)
$
be a measurable space.
For
$
A
	\in
\Sigma_X
$
and
$
Z
	\in
\Sigma_{\mathbf{I}},
$
we have
\if \MkLonger  1
\begin{align*}
&
(\mu^{\mathfrak{S}}_{\mathfrak{X}})^{-1}
\big(
	\{
		u
			\in
		\mathfrak{S}
		\mathfrak{X}
	\mid
		u
		(A)
			\in
		Z
	\}
\big)
			\\&=
\big\{
	v
		\in
	\mathfrak{S}^2
	\mathfrak{X}
\mid
	\mu^{\mathfrak{S}}_{\mathfrak{X}}
	(v)
		\in
	\{
		u
			\in
		\mathfrak{S}
		\mathfrak{X}
	\mid
		u
		(A)
			\in
		Z
	\}
\big\}
			\\&=
\big\{
	v
		\in
	\mathfrak{S}^2
	\mathfrak{X}
\mid
	\mu^{\mathfrak{S}}_{\mathfrak{X}}
	(v)
	(A)
		\in
	Z
\big\}
			\\&=
\big\{
	v
		\in
	\mathfrak{S}^2
	\mathfrak{X}
\mid
\xi^{\mathfrak{S}}_{
	\mathfrak{S}
	\mathfrak{X}
}
(
	\varepsilon^{\mathfrak{S}}_{\mathfrak{X}}
	(A)
)
(v)
		\in
	Z
\big\}
			&(\textrm{
				by Lemma \ref{eq:alterMu}
			})
			\\&=
\big(
\xi^{\mathfrak{S}}_{
	\mathfrak{S}
	\mathfrak{X}
}
(
	\varepsilon^{\mathfrak{S}}_{\mathfrak{X}}
	(A)
)
\big)^{-1}
(Z)
			\in
\Sigma_{
	\mathfrak{S}^2
	\mathfrak{X}
} .
			&(\textrm{
				by Proposition \ref{eq:defXiX}
			})
\end{align*}
\else 
by
				Lemma \ref{eq:alterMu}
and
				Proposition \ref{eq:defXiX},
\begin{align*}
&
(\mu^{\mathfrak{S}}_{\mathfrak{X}})^{-1}
\big(
	\{
		u
			\in
		\mathfrak{S}
		\mathfrak{X}
	\mid
		u
		(A)
			\in
		Z
	\}
\big)
			=
\big\{
	v
		\in
	\mathfrak{S}^2
	\mathfrak{X}
\mid
	\mu^{\mathfrak{S}}_{\mathfrak{X}}
	(v)
		\in
	\{
		u
			\in
		\mathfrak{S}
		\mathfrak{X}
	\mid
		u
		(A)
			\in
		Z
	\}
\big\}
			\\&=
\big\{
	v
		\in
	\mathfrak{S}^2
	\mathfrak{X}
\mid
	\mu^{\mathfrak{S}}_{\mathfrak{X}}
	(v)
	(A)
		\in
	Z
\big\}
			=
\big\{
	v
		\in
	\mathfrak{S}^2
	\mathfrak{X}
\mid
\xi^{\mathfrak{S}}_{
	\mathfrak{S}
	\mathfrak{X}
}
(
	\varepsilon^{\mathfrak{S}}_{\mathfrak{X}}
	(A)
)
(v)
		\in
	Z
\big\}
			\\&=
\big(
\xi^{\mathfrak{S}}_{
	\mathfrak{S}
	\mathfrak{X}
}
(
	\varepsilon^{\mathfrak{S}}_{\mathfrak{X}}
	(A)
)
\big)^{-1}
(Z)
			\in
\Sigma_{
	\mathfrak{S}^2
	\mathfrak{X}
} .
\end{align*}
\fi 
\end{proof}

We will show that
$\mu^{\mathfrak{S}}$
is a natural transformation.
But before doing that,
we prepare the following lemma, showing the principle of the substitution integral.

\begin{lem}
\label{lem:lemTwo}
Let
$
\mathfrak{X}
$
and
$
\mathfrak{Y}
$
be measurable spaces,
$
h :
\mathfrak{X}
	\to
\mathfrak{Y}
$
be a measurable function,
$
f 
	\in
L^{\infty}(
\mathfrak{Y}
) ,
$
and
$
u
	\in
\mathfrak{S}
\mathfrak{X}
$.
Then, we have
\begin{equation}
\label{eq:lemTwo}
I_{
	\mathfrak{X}
}^u
(f \circ h)
		=
I_{
	\mathfrak{Y}
}^{u \circ h^{-1}}
(f)
\end{equation}
if the both sides are defined.
\end{lem}

Note that
$
u \circ h^{-1}
	\in
\mathfrak{S}
\mathfrak{Y}
$
since 
$
\mathfrak{S}
$
is a C-functor.
\begin{proof}
Let
$
\mathfrak{X} = (X, \Sigma_X)
$
and
$
\mathfrak{Y} = (Y, \Sigma_Y)
$
be two measurable spaces.
First,
we prove 
(\ref{eq:lemTwo})
when $f$ is a finite step function such as
\if \MkLonger  1
\[
f
	:=
\sum_{i=1}^n
a_i
\mathbb{1}_Y(B_i)
\]
\else 
$
f
	:=
\sum_{i=1}^n
a_i
\mathbb{1}_Y(B_i)
$
\fi 
where
$
a_1 \ge a_2 \ge \cdots \ge a_n
$
are decreasing real numbers
and
$B_i$
are mutually disjoint elements of
$\Sigma_Y$.
Then, 
by a simple consideration,
we have
\begin{equation}
f \circ h
	:=
\sum_{i=1}^n
a_i
\mathbb{1}_X(
	h^{-1}(B_i)
) ,
\end{equation}
which means 
$
f \circ h
	\in
L^{\infty}(
	\mathfrak{X}
)
$
is a finite step function
since
$h$
is measurable.
Hence,
\if \MkLonger  1
\begin{align*}
I_{\mathfrak{X}}^u
(f \circ h)
			&=
\sum_{i=1}^n
	(a_i - a_{i+1})
	\,
	u\Big(
		\bigcup_{j=1}^i
			h^{-1}(B_j)
	\Big)
			&(\textrm{by Lemma \ref{lem:stepFunChoquetInt}})
			\\&=
\sum_{i=1}^n
	(a_i - a_{i+1})
	\,
	u\Big(
		h^{-1} \big(
			\bigcup_{j=1}^i
				B_j
		\big)
	\Big)
			\\&=
\sum_{i=1}^n
	(a_i - a_{i+1})
	\,
	(u \circ h^{-1})
	\Big(
		\bigcup_{j=1}^i
			B_j
	\Big)
			\\&=
I_{\mathfrak{Y}}^{
	u \circ h^{-1}
}
	(f) .
			&(\textrm{by Lemma \ref{lem:stepFunChoquetInt}})
\end{align*}
\else 
by Lemma \ref{lem:stepFunChoquetInt},
\begin{align*}
I_{\mathfrak{X}}^u
(f \circ h)
			&=
\sum_{i=1}^n
	(a_i - a_{i+1})
	\,
	u\Big(
		\bigcup_{j=1}^i
			h^{-1}(B_j)
	\Big)
			=
\sum_{i=1}^n
	(a_i - a_{i+1})
	\,
	u\Big(
		h^{-1} \big(
			\bigcup_{j=1}^i
				B_j
		\big)
	\Big)
			\\&=
\sum_{i=1}^n
	(a_i - a_{i+1})
	\,
	(u \circ h^{-1})
	\Big(
		\bigcup_{j=1}^i
			B_j
	\Big)
			=
I_{\mathfrak{Y}}^{
	u \circ h^{-1}
}
	(f) .
\end{align*}
\fi 

For general
$
f \in L^{\infty}(\mathfrak{Y}),
$
we create 
$\underline{f}_n$
and
$\bar{f}_n$
by the method described in
Note \ref{note:discretizeU}.
By the monotonicity of
$
I_{\mathfrak{X}}^u
$
and the previous result about step functions, 
we have
\begin{equation}
I_{\mathfrak{Y}}^{ u \circ h^{-1} }
(
	\underline{f}_n
)
	=
I_{\mathfrak{X}}^u
(
	\underline{f}_n
		\circ
	h
)
	\le
I_{\mathfrak{X}}^u
(
f
	\circ
h
)
	\le
I_{\mathfrak{X}}^u
(
\bar{f}_n
	\circ
h
) 
	=
I_{\mathfrak{Y}}^{ u \circ h^{-1} }
(
\bar{f}_n
) .
\end{equation}
But,
both
$
I_{\mathfrak{Y}}^{ u \circ h^{-1} }
(
\underline{f}_n
)
$
and
$
I_{\mathfrak{Y}}^{ u \circ h^{-1} }
(
\bar{f}_n
)
$
converge to
$
I_{\mathfrak{Y}}^{ u \circ h^{-1} }
(
f
) .
$
Therefore, we get
(\ref{eq:lemTwo}).

\end{proof}

\begin{prop}
\label{prop:MuNatTrans}
The correspondence 
$\mu^{\mathfrak{S}}$
is a natural transformation
$
\mu^{\mathfrak{S}}
	:
\mathfrak{S}^2
	\dot{\to}
\mathfrak{S}
$.
i.e.
the diagram
\ref{diag:natTransMu}
commutes.

\begin{diagram}[htb]
\begin{equation*}
\xymatrix{
&
	\mathfrak{S}^2
		\ar @{->}^{\mu^{\mathfrak{S}}} [r]
&
	\mathfrak{S}
\\
    \mathfrak{X}
      \ar @{->}_{h} [d]
  &
    \mathfrak{S}^2
    \mathfrak{X}
      \ar @{->}_{\mathfrak{S}^2 h} [d]
      \ar @{->}^{\mu^{\mathfrak{S}}_{
    		\mathfrak{X}
		}} [r]
  &
    \mathfrak{S}
    \mathfrak{X}
      \ar @{->}^{\mathfrak{S}h} [d]
 \\
    \mathfrak{Y}
  &
    \mathfrak{S}^2
    \mathfrak{Y}
      \ar @{->}_{\mu^{\mathfrak{S}}_{
    		\mathfrak{Y}
		}} [r]
  &
    \mathfrak{S}
    \mathfrak{Y}
}
\end{equation*}
\caption{natural transformation $
\mu^{\mathfrak{S}}
	:
\mathfrak{S}^2
	\dot{\to}
\mathfrak{S}
$}
\label{diag:natTransMu}
\end{diagram}
\end{prop}
\begin{proof}
Let
$
\mathfrak{X} = (X, \Sigma_X)
$
and
$
\mathfrak{Y} = (Y, \Sigma_Y)
$
be two measurable spaces,
$
v
	\in
\mathfrak{S}^2
\mathfrak{X}
$
and
$
B \in \Sigma_Y
$.
Then, we have
\if \MkLonger  1
\begin{align*}
(
\mu^{\mathfrak{S}}_{\mathfrak{Y}}
	\circ
\mathfrak{S}^2
h
)
(v)(B)
		&=
\mu^{\mathfrak{S}}_{\mathfrak{Y}}
(
	\mathfrak{S}
		(
		\mathfrak{S}
		h
		)
		(v)
)
(B)
		\\&=
\mu^{\mathfrak{S}}_{\mathfrak{Y}}
(
	v
		\circ
	(
	\mathfrak{S}
	h
	)^{-1}
)
(B)
		\\&=
I_{
	\mathfrak{S}
	\mathfrak{Y}
}^{
	v
		\circ
	(
		\mathfrak{S}
		h
	)^{-1}
}
(
	\varepsilon_Y(B)
)
		\\&=
I_{
	\mathfrak{S}
	\mathfrak{X}
}^{
	v
}
(
	\varepsilon_Y(B)
		\circ
	\mathfrak{S}
	h
)
		&(\textrm{
by Lemma \ref{lem:lemTwo}
		})
		\\&=
I_{
	\mathfrak{S}
	\mathfrak{X}
}^{
	v
}
(
	\varepsilon_X
	(h^{-1}(B))
)
			&(\textrm{
				by Proposition \ref{prop:epsBcircGf}
			})
		\\&=
\mu^{\mathfrak{S}}_{\mathfrak{X}}
(
	v
)
(
	h^{-1}
	(B)
)
		\\&=
(
	\mu^{\mathfrak{S}}_{\mathfrak{X}}
	(v)
		\circ
	h^{-1}
)
(B)
		\\&=
\mathfrak{S}
h
(
	\mu^{\mathfrak{S}}_{\mathfrak{X}}
	(v)
)
(B)
		\\&=
(
\mathfrak{S}
h
	\circ
\mu^{\mathfrak{S}}_{\mathfrak{X}}
)
(v)(B) .
\end{align*}
\else 
by
Lemma \ref{lem:lemTwo}
and
Proposition \ref{prop:epsBcircGf},
\begin{align*}
		&
(
\mu^{\mathfrak{S}}_{\mathfrak{Y}}
	\circ
\mathfrak{S}^2
h
)
(v)(B)
		=
\mu^{\mathfrak{S}}_{\mathfrak{Y}}
(
	\mathfrak{S}
		(
		\mathfrak{S}
		h
		)
		(v)
)
(B)
		=
\mu^{\mathfrak{S}}_{\mathfrak{Y}}
(
	v
		\circ
	(
	\mathfrak{S}
	h
	)^{-1}
)
(B)
		=
I_{
	\mathfrak{S}
	\mathfrak{Y}
}^{
	v
		\circ
	(
		\mathfrak{S}
		h
	)^{-1}
}
(
	\varepsilon_Y(B)
)
		\\&=
I_{
	\mathfrak{S}
	\mathfrak{X}
}^{
	v
}
(
	\varepsilon_Y(B)
		\circ
	\mathfrak{S}
	h
)
		=
I_{
	\mathfrak{S}
	\mathfrak{X}
}^{
	v
}
(
	\varepsilon_X
	(h^{-1}(B))
)
\mu^{\mathfrak{S}}_{\mathfrak{X}}
(
	v
)
(
	h^{-1}
	(B)
)
		=
(
	\mu^{\mathfrak{S}}_{\mathfrak{X}}
	(v)
		\circ
	h^{-1}
)
(B)
		\\&=
\mathfrak{S}
h
(
	\mu^{\mathfrak{S}}_{\mathfrak{X}}
	(v)
)
(B)
		=
(
\mathfrak{S}
h
	\circ
\mu^{\mathfrak{S}}_{\mathfrak{X}}
)
(v)(B) .
\end{align*}
\fi 

\end{proof}

\begin{prop}
\label{prop:GiryUnit}
The diagram in
Diagram \ref{diag:defGiryUnit}
commutes.

\begin{diagram}[htb]
\begin{equation*}
\xymatrix{
	1_{\Mble}
	\mathfrak{S}
		\ar @{}^-{\rotatebox{90}{=}} @<-4pt> [d]
		\ar @{->}^-{\eta^{\mathfrak{S}} \mathfrak{S}} [r]
&
    \mathfrak{S}^2
		\ar @{->}^{\mu^{\mathfrak{S}}} [d]
&
	\mathfrak{S}
	1_{\Mble}
		\ar @{}^-{\rotatebox{90}{=}} @<-4pt> [d]
		\ar @{->}_-{\mathfrak{S} \eta^{\mathfrak{S}}} [l]
 \\
    \mathfrak{S}
		\ar @{}^-{=} @<-4pt> [r]
&
    \mathfrak{S}
		\ar @{}^-{=} @<-4pt> [r]
&
    \mathfrak{S}
}
\end{equation*}
\caption{Giry unit}
\label{diag:defGiryUnit}
\end{diagram}
\end{prop}
\begin{proof}
Let
$
\mathfrak{X}
	=
(X, \Sigma_X)
$
be a measurable space.
The natural transformations described in
Diagram \ref{diag:defGiryUnit}
are defined as
$
(
\eta^{\mathfrak{S}}
\mathfrak{S}
)_{
	\mathfrak{X}
}
	:=
\eta^{\mathfrak{S}}_{
	\mathfrak{S}
	\mathfrak{X}
}
$
and
$
(
\mathfrak{S}
\eta^{\mathfrak{S}}
)_{
	\mathfrak{X}
}
	:=
\mathfrak{S}(
	\eta^{\mathfrak{S}}_{
		\mathfrak{X}
	}
) .
$

In order to show the diagram commutes, we need to show that
$
\mu^{\mathfrak{S}}
	\circ
\eta^{\mathfrak{S}}
\mathfrak{S}
		=
1_{
	\mathfrak{S}
	\mathfrak{X}
}
$
and
$
\mu^{\mathfrak{S}}
	\circ
\mathfrak{S}
\eta^{\mathfrak{S}}
		=
1_{
	\mathfrak{S}
	\mathfrak{X}
}
$.
Let
$
u
	\in
\mathfrak{S}
\mathfrak{X}
$
and
$
A \in \Sigma_X
$.

For the first equation,
we have
\if \MkLonger  1
\begin{align*}
(
\mu^{\mathfrak{S}}
	\circ
\eta^{\mathfrak{S}}
\mathfrak{S}
)_{
	\mathfrak{X}
}
(u)(A)
		&=
(
\mu^{\mathfrak{S}}_{\mathfrak{X}}
	\circ
\eta^{\mathfrak{S}}_{
	\mathfrak{S}
	X
}
)
(u)(A)
		\\&=
\mu^{\mathfrak{S}}_{\mathfrak{X}}
(
	\eta^{\mathfrak{S}}_{
		\mathfrak{S}
		\mathfrak{X}
	}
	(u)
)
(A)
		\\&=
I_{
	\mathfrak{S}
	\mathfrak{X}
}^{
	\eta^{\mathfrak{S}}_{
		\mathfrak{S}
		\mathfrak{X}
	}
	(u)
}
\big(
	\varepsilon^{\mathfrak{S}}_{\mathfrak{X}}
	(A)
\big)
		\\&=
\varepsilon^{\mathfrak{S}}_{\mathfrak{X}}
(A)
(u)
			&(\textrm{
by Proposition \ref{prop:DiracMesX}
			})
		\\&=
u(A)
		\\&=
1_{
	\mathfrak{S}
	\mathfrak{X}
}
(u)(A) .
\end{align*}
\else 
by Proposition \ref{prop:DiracMesX},
\begin{align*}
		&
(
\mu^{\mathfrak{S}}
	\circ
\eta^{\mathfrak{S}}
\mathfrak{S}
)_{
	\mathfrak{X}
}
(u)(A)
		=
(
\mu^{\mathfrak{S}}_{\mathfrak{X}}
	\circ
\eta^{\mathfrak{S}}_{
	\mathfrak{S}
	\mathfrak{X}
}
)
(u)(A)
		=
\mu^{\mathfrak{S}}_{\mathfrak{X}}
(
	\eta^{\mathfrak{S}}_{
		\mathfrak{S}
		\mathfrak{X}
	}
	(u)
)
(A)
		\\&=
I_{
	\mathfrak{S}
	\mathfrak{X}
}^{
	\eta^{\mathfrak{S}}_{
		\mathfrak{S}
		\mathfrak{X}
	}
	(u)
}
\big(
	\varepsilon^{\mathfrak{S}}_{\mathfrak{X}}
	(A)
\big)
		=
\varepsilon^{\mathfrak{S}}_{\mathfrak{X}}
(A)
(u)
		=
u(A)
		=
1_{
	\mathfrak{S}
	\mathfrak{X}
}
(u)(A) .
\end{align*}
\fi 

Before proving 
the second equation,
we show
\begin{equation}
\label{eq:epsAEtaId}
\varepsilon^{\mathfrak{S}}_{\mathfrak{X}}
(A)
	\circ
\eta^{\mathfrak{S}}_{\mathfrak{X}}
		=
\mathbb{1}_X(A) .
\end{equation}
But this comes from
the fact that
for
$x \in X$,
\if \MkLonger  1
\begin{align*}
(
\varepsilon^{\mathfrak{S}}_{\mathfrak{X}}
(A)
	\circ
\eta^{\mathfrak{S}}_{\mathfrak{X}}
)
(x)
		&=
\varepsilon^{\mathfrak{S}}_{\mathfrak{X}}
(A)
(
	\eta^{\mathfrak{S}}_{\mathfrak{X}}
	(x)
)
		\\&=
\eta^{\mathfrak{S}}_{\mathfrak{X}}
(x)
(A)
		\\&=
\mathbb{1}_X(A)
(x) .
			&(\textrm{
by Proposition \ref{prop:etaXxX}
			})
\end{align*}
\else 
by Proposition \ref{prop:etaXxX},
\[
(
\varepsilon^{\mathfrak{S}}_{\mathfrak{X}}
(A)
	\circ
\eta^{\mathfrak{S}}_{\mathfrak{X}}
)
(x)
		=
\varepsilon^{\mathfrak{S}}_{\mathfrak{X}}(A)
(
	\eta^{\mathfrak{S}}_{\mathfrak{X}}
	(x)
)
		=
\eta^{\mathfrak{S}}_{\mathfrak{X}}
(x)
(A)
		=
\mathbb{1}_X(A)
(x) .
\]
\fi 
Then, we obtain
\if \MkLonger  1
\begin{align*}
(
\mu^{\mathfrak{S}}
	\circ
\mathfrak{S}
\eta^{\mathfrak{S}}
)_{
	\mathfrak{X}
}
(u)(A)
		&=
(
\mu^{\mathfrak{S}}_{\mathfrak{X}}
	\circ
\mathfrak{S}
(
	\eta^{\mathfrak{S}}_{\mathfrak{X}}
)
)
(u)(A)
		\\&=
\mu^{\mathfrak{S}}_{\mathfrak{X}}
(
	\mathfrak{S}
	(
		\eta^{\mathfrak{S}}_{\mathfrak{X}}
	)
	(u)
)
(A)
		\\&=
\mu^{\mathfrak{S}}_{\mathfrak{X}}
(
	u
		\circ
	(\eta^{\mathfrak{S}}_{\mathfrak{X}})^{-1}
)
(A)
		\\&=
I_{
	\mathfrak{S}
	\mathfrak{X}
}^{
	u
		\circ
	(\eta^{\mathfrak{S}}_{\mathfrak{X}})^{-1}
}
\big(
	\varepsilon^{\mathfrak{S}}_{\mathfrak{X}}(A)
\big)
		\\&=
I_{
	\mathfrak{X}
}^{
	u
}
\big(
	\varepsilon^{\mathfrak{S}}_{\mathfrak{X}}(A)
		\circ
	\eta^{\mathfrak{S}}_{\mathfrak{X}}
\big)
			&(\textrm{
by Lemma \ref{lem:lemTwo}
			})
		\\&=
I_{
	\mathfrak{X}
}^{
	u
}
\big(
	\mathbb{1}_X(A)
\big)
			&(\textrm{
by (\ref{eq:epsAEtaId})
			})
		\\&=
u(A)
			&(\textrm{
by (\ref{eq:ChoquetOneA})
			})
		\\&=
1_{
	\mathfrak{S}
	\mathfrak{X}
}
(u)(A) .
\end{align*}
\else 
by Lemma \ref{lem:lemTwo}
and equations
(\ref{eq:epsAEtaId}), 
(\ref{eq:ChoquetOneA}),
\begin{align*}
		&
(
\mu^{\mathfrak{S}}
	\circ
\mathfrak{S}
\eta^{\mathfrak{S}}
)_{
	\mathfrak{X}
}
(u)(A)
		=
(
\mu^{\mathfrak{S}}_{\mathfrak{X}}
	\circ
\mathfrak{S}
(
	\eta^{\mathfrak{S}}_{\mathfrak{X}}
)
)
(u)(A)
		=
\mu^{\mathfrak{S}}_{\mathfrak{X}}
(
	\mathfrak{S}
	(
		\eta^{\mathfrak{S}}_{\mathfrak{X}}
	)
	(u)
)
(A)
		=
\mu^{\mathfrak{S}}_{\mathfrak{X}}
(
	u
		\circ
	(\eta^{\mathfrak{S}}_{\mathfrak{X}})^{-1}
)
(A)
		\\&=
I_{
	\mathfrak{S}
	\mathfrak{X}
}^{
	u
		\circ
	(\eta^{\mathfrak{S}}_{\mathfrak{X}})^{-1}
}
\big(
	\varepsilon^{\mathfrak{S}}_{\mathfrak{X}}(A)
\big)
		=
I_{
	\mathfrak{X}
}^{
	u
}
\big(
	\varepsilon^{\mathfrak{S}}_{\mathfrak{X}}(A)
		\circ
	\eta^{\mathfrak{S}}_{\mathfrak{X}}
\big)
		=
I_{
	\mathfrak{X}
}^{
	u
}
\big(
	\mathbb{1}_X(A)
\big)
		=
u(A)
		=
1_{
	\mathfrak{S}
	\mathfrak{X}
}
(u)(A) .
\end{align*}
\fi 

\end{proof}


\begin{defn}
\label{defn:interSecS}
Let
$
\{
	\mathfrak{S}_{\alpha}
\}_{\alpha}
$
be a non-empty set of CM-functors.
Then,
$
\mathfrak{S}
	:=
\bigcap_{\alpha}
	\mathfrak{S}_{\alpha}
		:
\Mble \to \Mble
$
is a functor defined by for any
measurable space
$
\mathfrak{X}
$,
\begin{equation}
\label{eq:interSecS}
\Big(
\bigcap_{\alpha}
	\mathfrak{S}_{\alpha}
\Big)
\mathfrak{X}
		:=
\Big(
\bigcap_{\alpha}
\big(
	\mathfrak{S}_{\alpha}
	\mathfrak{X}
\big),
	\,
\Sigma_{
	\mathfrak{S}
	\mathfrak{X}
}
\Big),
\end{equation}
where
$
\Sigma_{
	\mathfrak{S}
	\mathfrak{X}
}
$
is the smallest $\sigma$-algebra
making all inclusion maps
$
i^{\alpha}_{
	\mathfrak{X}
}
	:
\mathfrak{S}
\mathfrak{X}
	\to
\mathfrak{S}_{\alpha}
\mathfrak{X}
$
measurable.

\end{defn}

\begin{prop}
\label{prop:existCapS}
The functor
$
\mathfrak{S}
	:=
\bigcap_{\alpha}
	\mathfrak{S}_{\alpha}
$
defined by
(\ref{eq:interSecS})
is a CM-functor.
\end{prop}
\begin{proof}
Let
$
\mathfrak{X} = (X, \Sigma_X)
$
and
$
\mathfrak{Y} = (Y, \Sigma_Y)
$
be two measurable spaces.
First, let us check if
$\mathfrak{S}$
is a C-functor.
Let 
$h :
\mathfrak{X}
	\to
\mathfrak{Y}
$
be a measurable map and
$
u
	\in
\mathfrak{S}
\mathfrak{X}
$.
Then, for every
$\alpha$,
\[
u 
		\in 
\mathfrak{S}
\mathfrak{X}
		=
\bigcap_{\alpha}
\big(
	\mathfrak{S}_{\alpha}
	\mathfrak{X}
\big)
		\subset
\mathfrak{S}_{\alpha}
\mathfrak{X} .
\]
Since
$
\mathfrak{S}_{\alpha}
$
is a C-functor,
we have
$
u \circ h^{-1}
	\in
\mathfrak{S}_{\alpha}
\mathfrak{Y} .
$
Therefore,
$
u \circ h^{-1}
	\in
\bigcap_{\alpha}
\big(
	\mathfrak{S}_{\alpha}
	\mathfrak{Y} 
\big)
		=
\mathfrak{S} 
\mathfrak{Y} .
$

Next,
we will check if the equation
(\ref{eq:condM}) holds.

For
$
v \in
\mathfrak{S}^2
\mathfrak{X}
$
and
any $\alpha$,
define a capacity
$
v_{\alpha}
$
on
$
\big(
\mathfrak{S}_{\alpha}
\mathfrak{X},
	\,
\Sigma_{
	\mathfrak{S}_{\alpha} 
	\mathfrak{X}
}
\big)
$
by
\[
v_{\alpha}
	:=
v \circ
\big(
	i^{\alpha}_{
		\mathfrak{X}
	}
\big)^{-1} .
\]
Then,
since
$
\mathfrak{S}
$
is a C-functor,
we have
\[
v_{\alpha}
		\in
\mathfrak{S}(
	\mathfrak{S}_{\alpha}
	\mathfrak{X}
)
		\subset
\mathfrak{S}_{\alpha}^2
\mathfrak{X} .
\]
Now, for any
$
A \in \Sigma_X
$,
\if \MkLonger  1
\begin{align*}
\big(
I_{
	\mathfrak{S}_{\alpha}
	\mathfrak{X}
}^{
	v_{\alpha}
}
	\circ
\varepsilon^{\mathfrak{S}_{\alpha}}_{\mathfrak{X}}
\big)
(A)
			&=
\int_0^1
	v_{\alpha}\big(\{
		u \in
		\mathfrak{S}_{\alpha}
		\mathfrak{X} 
	\mid
		u(A) \ge r
	\}\big)
dr
			&(\textrm{by Lemma \ref{lem:alterMu}})
			\\&=
\int_0^1
	v\Big(
	\big(i^{\alpha}_X\big)^{-1}
	\big(\{
		u \in
		\mathfrak{S}_{\alpha}
		\mathfrak{X} 
	\mid
		u(A) \ge r
	\}\big)
	\Big)
dr
			\\&=
\int_0^1
	v\Big(
	\mathfrak{S} 
	\mathfrak{X} 
		\cap
	\big(\{
		u \in
		\mathfrak{S}_{\alpha} 
		\mathfrak{X} 
	\mid
		u(A) \ge r
	\}\big)
	\Big)
dr
			\\&=
\int_0^1
	v\Big(
	\big(\{
		u \in
		\mathfrak{S}
		\mathfrak{X} 
	\mid
		u(A) \ge r
	\}\big)
	\Big)
dr
			\\&=
\big(
I_{
	\mathfrak{S}
	\mathfrak{X} 
}^{
	v
}
	\circ
\varepsilon^{\mathfrak{S}}_{\mathfrak{X}}
\big)
(A) .
			&(\textrm{by Lemma \ref{lem:alterMu}})
\end{align*}
\else 
by Lemma \ref{lem:alterMu},
\begin{align*}
			&
\big(
I_{
	\mathfrak{S}_{\alpha}
	\mathfrak{X} 
}^{
	v_{\alpha}
}
	\circ
\varepsilon^{\mathfrak{S}_{\alpha}}_{\mathfrak{X}}
\big)
(A)
			=
\int_0^1
	v_{\alpha}\big(\{
		u \in
		\mathfrak{S}_{\alpha}
		\mathfrak{X} 
	\mid
		u(A) \ge r
	\}\big)
dr
			\\&=
\int_0^1
	v\Big(
	\big(i^{\alpha}_X\big)^{-1}
	\big(\{
		u \in
		\mathfrak{S}_{\alpha}
		\mathfrak{X} 
	\mid
		u(A) \ge r
	\}\big)
	\Big)
dr
			=
\int_0^1
	v\Big(
	\mathfrak{S}
	\mathfrak{X} 
		\cap
	\big(\{
		u \in
		\mathfrak{S}_{\alpha}
		\mathfrak{X} 
	\mid
		u(A) \ge r
	\}\big)
	\Big)
dr
			\\&=
\int_0^1
	v\Big(
	\big(\{
		u \in
		\mathfrak{S}
		\mathfrak{X} 
	\mid
		u(A) \ge r
	\}\big)
	\Big)
dr
			=
\big(
I_{
	\mathfrak{S}
	\mathfrak{X} 
}^{
	v
}
	\circ
\varepsilon^{\mathfrak{S}}_{\mathfrak{X}}
\big)
(A) .
\end{align*}
\fi 
Since
$
\mathfrak{S}_{\alpha}
$
satisfies 
(\ref{eq:condM}),
we obtain
\[
I_{
	\mathfrak{S}
	\mathfrak{X} 
}^{
	v
}
	\circ
\varepsilon^{\mathfrak{S}}_{\mathfrak{X}}
		=
I_{
	\mathfrak{S}_{\alpha}
	\mathfrak{X} 
}^{
	v_{\alpha}
}
	\circ
\varepsilon^{\mathfrak{S}_{\alpha}}_{\mathfrak{X}}
		\in
\mathfrak{S}_{\alpha}
\mathfrak{X} .
\]
Hence,
\[
I_{
	\mathfrak{S}
	\mathfrak{X}
}^{
	v
}
	\circ
\varepsilon^{\mathfrak{S}}_{\mathfrak{X}}
		\in
\bigcap_{\alpha}
	\mathfrak{S}_{\alpha} 
	\mathfrak{X}
		=
\mathfrak{S} X ,
\]
which completes the proof.

\end{proof}

\begin{defn}
\label{defn:admitRel}
Let
$
\mathfrak{X}
	=
(X, \Sigma_X, U_X)
$
be an uncertainty space.
A CM-functor
$
\mathfrak{S}
$
is said to \newword{admit}
$
\mathfrak{X}
$
if
there exists a monic
$
m :
\mathfrak{L}
\mathfrak{X}
	\to
\mathfrak{S}
(X, \Sigma_X)
$
in
$\Mble$.

$
\mathfrak{S}
$
is said to \newword{admit}
a family of uncertainty spaces
$
\mathbb{X}
		=
\{
	\mathfrak{X}_{\alpha}
\}_{\alpha}
$
if
$
\mathfrak{S}
$
admits
$
	\mathfrak{X}_{\alpha}
$
for every
$\alpha$.

\end{defn}

Since
$\mathcal{C}$
admits every
uncertainty space,
and by
Proposition \ref{prop:existCapS},
the following definition is well-defined.

\begin{defn}
\label{defn:smallestS}
Let
$
\mathbb{X}
$
be
a non-empty family of uncertainty spaces.
The CM-functor 
generated by
$
\mathbb{X}
$
is defined by
\begin{equation}
\label{eq:smallestS}
\mathfrak{S}\big[
\mathbb{X}
\big]
		:=
\bigcap \Big\{
	\mathfrak{S}
\mid
	\mathfrak{S}
\textrm{ admits }	
\mathbb{X}
\Big\} .
\end{equation}

\end{defn}

We can think
$
\mathfrak{S}[
\mathbb{X}
]
$
as the
``minimum''
CM-functor admitting 
$
\mathbb{X}
$.

\section{Universal uncertainty spaces}
\label{sec:uniUncSp}

In Section \ref{sec:hieraUncertain},
we introduced $n$-layer uncertainty
showing as 
U-sequences
which are
infinite sequences of uncertain spaces.
We gave a concrete example of the 3rd layer uncertainty.
But who knows we need to consider higher level uncertainty in future.
If we get a kind of limits of the sequence, 
we may be able to use it as a universal domain for analyzing uncertainty in any levels.


Actually,
for a U-sequence $
\mathbb{X}
	=
\{
	\mathfrak{X}_n
		=
	(
		X_n, \Sigma_{X_n}, U_{X_n}
	)
\}_{n \in \mathbb{N}}
$,
we constructed its envelop CM-functor
$\mathfrak{S} := \mathfrak{S}[\mathbb{X}]$
admitting 
$\mathbb{X}$
by Definition \ref{defn:smallestS}.
Then,  
by applying the functor 
$
\mathfrak{S}
$
repeatedly,
we can generate an enlarged U-sequence
$
\mathbb{X}'
	=
\{
	(X'_n, \Sigma_{X'_n}, U_{X'_n})
\}_{n \in \mathbb{N}},
$
where 
for $n \in \mathbb{N}$,
\begin{equation}
\label{eq:enlargedUseq}
(X'_0, \Sigma_{X'_0})
 := 
(X_0, \Sigma_{X_0}),
		\quad
(X'_{n+1}, \Sigma_{X'_{n+1}})
	:=
\mathfrak{S}
(X'_{n}, \Sigma_{X'_n}),
		\quad
U_{X'_n} := X'_{n+1}.
\end{equation}

In this section,
we shall get its limit, which is called the universal uncertainty space.
Let $\mathfrak{S}$ be
a CM-functor
throughout this section.

\subsection{Retraction sequence}
\label{sec:retraction}

Let
$
\mathfrak{X}
	=
(X, \Sigma_X)
$
be a fixed measurable space.
Then,
we have a following 
infinite sequence
in $\Mble$,
which shows
the hierarchy of uncertainty on 
$\mathfrak{X}$.

\begin{equation}
\label{eq:infUnceSeq}
\xymatrix@C=40 pt@R=20 pt{
	\mathfrak{S}
	\mathfrak{X}
		\ar @{->}_{
			\eta^{\mathfrak{S}}_{
			\mathfrak{S}
			\mathfrak{X}
			}
		} @<-2pt> [r]
&
	\mathfrak{S}^2
	\mathfrak{X}
		\ar @{->}_{
			\eta^{\mathfrak{S}}_{
			\mathfrak{S}^2
			\mathfrak{X}
			}
		} @<-2pt> [r]
		\ar @{->}_{
			\mu^{\mathfrak{S}}_{
			\mathfrak{X}
			}
		} @<-2pt> [l]
&
	\mathfrak{S}^3
	\mathfrak{X}
		\ar @{->}_{
			\eta^{\mathfrak{S}}_{
			\mathfrak{S}^3
			\mathfrak{X}
			}
		} @<-2pt> [r]
		\ar @{->}_{
			\mu^{\mathfrak{S}}_{
			\mathfrak{S}
			\mathfrak{X}
			}
		} @<-2pt> [l]
&
	\cdots
		\ar @{->}_{
			\mu^{\mathfrak{S}}_{
			\mathfrak{S}^2
			\mathfrak{X}
			}
		} @<-2pt> [l]
}
\end{equation}

Here, we already know
for
$n = 1, 2, \cdots$
\begin{equation}
\mu^{\mathfrak{S}}_{
	\mathfrak{S}^{n-1}
	\mathfrak{X}
}
		\circ
\eta^{\mathfrak{S}}_{
	\mathfrak{S}^n
	\mathfrak{X}
}
		=
1_{
	\mathfrak{S}^n
	\mathfrak{X}
} 
\end{equation}
by 
Proposition \ref{prop:GiryUnit}.

Moreover,
if
$
v
	\in
\mathfrak{S}^{n+1}
\mathfrak{X}
$
is represented as
$
v
	=
\eta^{\mathfrak{S}}_{
	\mathfrak{S}^{n}
	\mathfrak{X}
}
(u)
$
with some
$
u
	\in
\mathfrak{S}^{n}
\mathfrak{X},
$
then we have
\begin{equation}
\big(
\eta^{\mathfrak{S}}_{
	\mathfrak{S}^n
	\mathfrak{X}
}
	\circ
\mu^{\mathfrak{S}}_{
	\mathfrak{S}^{n-1}
	\mathfrak{X}
}
\big)
(v)
		=
v .
\end{equation}
So, this is a retraction sequence.
In this sense, we can think that the map
$
\mu^{\mathfrak{S}}_{
	\mathfrak{S}^n
	\mathfrak{X}
}
$
gives an approximation at the lower hierarchy.

In this section, we will provide the limit of the sequence.

\subsection{The inverse limit}
\label{sec:invLimit}

\begin{defn}
\label{defn:candInvLim}
Let 
$
\mathfrak{X}
	= (X, \Sigma_X)
$
be a measurable space.
\begin{enumerate}
\item
Let
\begin{equation}
\label{eq:prodX}
\mathcal{P}_{
	\mathfrak{X}
}
	:=
\prod_{n=1}^{\infty}
	\mathfrak{S}^n
	\mathfrak{X}
\end{equation}
be the product space with
projections
$
j_n : 
\mathcal{P}_{
	\mathfrak{X}
}
	\to
	\mathfrak{S}^n
	\mathfrak{X}
$
for
$n = 1, 2, \cdots$,
equipped with the 
$\sigma$-algebra
\begin{equation}
\label{eq:prodXsigmaA}
\Sigma_{
	\mathcal{P}_{
		\mathfrak{X}
	}
}
	:=
\sigma\big(
	j_n
;
	n = 1, 2, \cdots
\big) .
\end{equation}

\item
Let 
\begin{equation}
\label{eq:SinftyX}
\mathfrak{S}^{\infty}
\mathfrak{X}
		:=
\{
	u = (u_1, u_2, \cdots)
		\in
	\mathcal{P}_{\mathfrak{X}}
\mid
	\forall n
		\,.\,
	u_n
		=
	\mu^{\mathfrak{S}}_{
		\mathfrak{S}^{n-1}
		\mathfrak{X}
	}
	(u_{n+1})
\} 
\end{equation}
be a sub-measurable space of
$
	\mathcal{P}_{\mathfrak{X}}
$ in which we use the same notations for its projections
$
j_n : 
\mathcal{P}_{\mathfrak{X}}
	\to
	\mathfrak{S}^n \mathfrak{X}
$ and its associated $\sigma$-algebra is defined as:
\begin{equation}
\label{eq:SinftyXsigmaAl}
\Sigma_{
	\mathfrak{S}^{\infty}
	\mathfrak{X}
}
	:=
\{
	A \cap
	\mathfrak{S}^{\infty}
	\mathfrak{X}
\mid
	A \in
	\Sigma_{
		\mathcal{P}_{\mathfrak{X}}
	}
\} .
\end{equation}

\end{enumerate}
\end{defn}

%

\begin{thm}
\label{ithm:InvLimit}
The measurable space
$
(
	\mathfrak{S}^{\infty}
	\mathfrak{X},
\Sigma_{
	\mathfrak{S}^{\infty}
	\mathfrak{X}
}
)
$
is the inverse limit of the sequence
(\ref{eq:infMonoUnceSeq})
in $\Mble$.
\begin{equation}
\label{eq:infMonoUnceSeq}
\xymatrix@C=40 pt@R=20 pt{
	\mathfrak{S}
	\mathfrak{X}
&
	\mathfrak{S}^2
	\mathfrak{X}
		\ar @{->}_{
			\mu^{\mathfrak{S}}_{
			\mathfrak{X}
			}
		} [l]
&
	\mathfrak{S}^3
	\mathfrak{X}
		\ar @{->}_{
			\mu^{\mathfrak{S}}_{
			\mathfrak{S}
			\mathfrak{X}
			}
		} [l]
&
	\mathfrak{S}^4
	\mathfrak{X}
		\ar @{->}_{
			\mu^{\mathfrak{S}}_{
			\mathfrak{S}^2
			\mathfrak{X}
			}
		} [l]
&
	\cdots
		\ar @{->}_{
			\mu^{\mathfrak{S}}_{
			\mathfrak{S}^3
			\mathfrak{X}
			}
		} [l]
}
\end{equation}
\end{thm}
\begin{proof}
We will show the Diagram
\ref{diag:invLimit}
commutes.
\begin{diagram}[htb]
\begin{equation*}
\xymatrix @C=30 pt @R=25 pt {
	\mathfrak{S}^n 
	\mathfrak{X}
&&
	\mathfrak{S}^{n+1}
	\mathfrak{X}
		\ar @{->}_{
			\mu^{\mathfrak{S}}_{
				\mathfrak{S}^{n-1} 
				\mathfrak{X}
			}
		} [ll]
\\
&
	\mathfrak{S}^{\infty} 
	\mathfrak{X}
		\ar @{->}^{
			j_n
		} [lu]
		\ar @{->}_{
			j_{n+1}
		} [ru]
\\
&
	\mathfrak{Y}
		\ar@/^/[uul]^{
			{k_n}
		}
		\ar@/_/[uur]_{
			{k_{n+1}}
		}
		\ar@{.>}[u]|-{
			{\Prefix^{\exists !}{h}}
		}
}
\end{equation*}
\caption{inverse limit $
\mathfrak{S}^{\infty}
\mathfrak{X}
$}
\label{diag:invLimit}
\end{diagram}

By the definition of
$
	\mathfrak{S}^{\infty}
	\mathfrak{X}
$,
we have
for
$
u = (u_1, u_2, \cdots)
	\in
\mathfrak{S}^{\infty}
\mathfrak{X}
$,
\if \MkLonger  1
\[
	u_n
		=
	\mu^{\mathfrak{S}}_{
		\mathfrak{S}^{n-1}
		\mathfrak{X}
	}
	(u_{n+1}) .
\]
\else 
$
	u_n
		=
	\mu^{\mathfrak{S}}_{
		\mathfrak{S}^{n-1}
		\mathfrak{X}
	}
	(u_{n+1}) .
$
\fi 
Then,
\if \MkLonger  1
\[
	j_n(u)
		=
	\mu^{\mathfrak{S}}_{
		\mathfrak{S}^{n-1}
		\mathfrak{X}
	}
	(j_{n+1}(u)) ,
\]
\else 
$
	j_n(u)
		=
	\mu^{\mathfrak{S}}_{
		\mathfrak{S}^{n-1}
		\mathfrak{X}
	}
	(j_{n+1}(u)) ,
$
\fi 
which means
\begin{equation}
\label{eq:jNmuJnP}
	j_n
		=
	\mu^{\mathfrak{S}}_{
		\mathfrak{S}^{n-1}
		\mathfrak{X}
	}
		\circ
	j_{n+1} .
\end{equation}

Next,
let us assume that
a measurable map
$
k_n : 
		\mathfrak{Y}
	\to
		\mathfrak{S}^{n}
		\mathfrak{X}
$
satisfies
$
	k_n
		=
	\mu^{\mathfrak{S}}_{
		\mathfrak{S}^{n-1}
		\mathfrak{X}
	}
		\circ
	k_{n+1} 
$
for all $n$
with a certain measurable space
$
\mathfrak{Y}
	=
(Y, \Sigma_Y)
$.
Then,
we will show that
there exists
a measurable map
$
h :
\mathfrak{Y} \to
		\mathfrak{S}^{\infty}
		\mathfrak{X}
$
such that
$
k_n = j_n \circ h 
$
for all $n$.

Define $h$ by
for all $y \in Y$,
\if \MkLonger  1
\begin{equation}
\label{eq:proofExistH}
h(y)
		:=
(k_1(y), k_2(y), \cdots) .
\end{equation}
\else 
$
h(y)
		:=
(k_1(y), k_2(y), \cdots) .
$
\fi 
Then, by the assumption,
\if \MkLonger  1
\[
	k_n(y)
		=
	\mu^{\mathfrak{S}}_{
		\mathfrak{S}^{n-1}
		\mathfrak{X}
	}(
		k_{n+1}(y)
	) .
\]
\else 
$
	k_n(y)
		=
	\mu^{\mathfrak{S}}_{
		\mathfrak{S}^{n-1}
		\mathfrak{X}
	}(
		k_{n+1}(y)
	) .
$
\fi 
Hence,
$
h(y)
	\in 
\mathfrak{S}^{\infty} 
\mathfrak{X}
$ .
We have also
\if \MkLonger  1
\[
(j_n \circ h)(y)
		=
j_n( h(y))
		=
k_n(y).
\]
\else 
$
(j_n \circ h)(y)
		=
j_n( h(y))
		=
k_n(y).
$
\fi 
So, the remaining is to show that
$h$
is measurable.

Let
$
A \in
\Sigma_{
	\mathfrak{S}^{m} 
	\mathfrak{X}
}.
$
Then,
\begin{equation}
B
			:=
j_m^{-1}(A)
			=
\big(
	\mathfrak{S}
	\mathfrak{X}
		\times
	\mathfrak{S}^2
	\mathfrak{X}
		\times
	\cdots
		\times
	\mathfrak{S}^{m-1} 
	\mathfrak{X}
		\times
	A
		\times
	\mathfrak{S}^{m+1}
	\mathfrak{X}
		\times
	\cdots
\big)
		\cap
\mathfrak{S}^{\infty} 
\mathfrak{X}
			\label{eq:jmInvA} .
\end{equation}
Hence,
\if \MkLonger  1
\begin{align*}
h^{-1}(B)
		&=
\{
	y \in Y
\mid
	(k_1(y), k_2(y), \cdots)
\in B
\}
		\\&=
\{
	y \in Y
\mid
	k_m(y) \in A
\; \land \;
	(k_1(y), k_2(y), \cdots)
\in 
	\mathfrak{S}^{\infty}
	\mathfrak{X}
\}
			&(\textrm{by (\ref{eq:jmInvA}) })
		\\&=
k_m^{-1}(A) 
		\in
\Sigma_Y
\end{align*}
\else 
by (\ref{eq:jmInvA}),
\begin{align*}
h^{-1}(B)
		&=
\{
	y \in Y
\mid
	(k_1(y), k_2(y), \cdots)
\in B
\}
		\\&=
\{
	y \in Y
\mid
	k_m(y) \in A
\; \land \;
	(k_1(y), k_2(y), \cdots)
\in 
	\mathfrak{S}^{\infty}
	\mathfrak{X}
\}
		=
k_m^{-1}(A) 
		\in
\Sigma_Y
\end{align*}
\fi 
since
$k_m$
is measurable.
Therefore,
$h$ is measurable.

Now, suppose there exists another
$
h' : \mathfrak{Y} \to
\mathfrak{S}^{\infty}
\mathfrak{X}
$
satisfying
$
k_n
	=
j_n \circ h'
$
for all
$n$.
Then, for every
$y \in Y$,
we have
\if \MkLonger  1
\[
j_n(h'(y))
	=
k_n(y)
	=
j_n(h(y)).
\]
\else 
$
j_n(h'(y))
	=
k_n(y)
	=
j_n(h(y)).
$
\fi 
Therefore
\if \MkLonger  1
\[
h'(y)
		=
(j_1(h'(y)), j_2(h'(y)), \cdots)
		=
(j_1(h(y)), j_2(h(y)), \cdots)
		=
h(y),
\]
\else 
$
h'(y)
		=
(j_1(h'(y)), j_2(h'(y)), \cdots)
		=
(j_1(h(y)), j_2(h(y)), \cdots)
		=
h(y),
$
\fi 
which means
$h'=h$.
Hence, $h$ is uniquely determined.

\end{proof}

\subsection{Projection system}
\label{sec:projSys}

Here is  a unified system of 
projection (or approximation) maps,
which is well defined thanks to
Theorem \ref{ithm:InvLimit}.

\begin{defn}
\label{defn:unifiedProj}
Let
$\mathbb{N}^+{}$
be the set
$\{1, 2, \cdots, \infty\}$
and
$m, n
	\in
\mathbb{N}^+{}
$.
Define a map
$
\iota_{\mathfrak{X}}^{m,n}
	:
\mathfrak{S}^m
\mathfrak{X}
	\to
\mathfrak{S}^n
\mathfrak{X}
$
by
\begin{equation}
\label{eq:unifiedProj}
\iota_{\mathfrak{X}}^{m,n}
		:=
\begin{cases}
	1_{
		\mathfrak{S}^m 
		\mathfrak{X}
	}
		&\quad \textrm{if} \;
	m = n,
\\
	\iota_{\mathfrak{X}}^{m-1,n}
		\circ
	\mu^{\mathfrak{S}}_{
		\mathfrak{S}^{m-2}
		\mathfrak{X}
	}
		&\quad \textrm{if} \;
	\infty > m > n ,
\\
	j_n
		&\quad \textrm{if} \;
	\infty = m > n,
\\
	\eta^{\mathfrak{S}}_{\mathfrak{S}^{n-1} \mathfrak{X}}
		\circ
	\iota_{\mathfrak{X}}^{m,n-1}
		&\quad \textrm{if} \;
	m < n < \infty ,
\\
	\langle
		\iota_{\mathfrak{X}}^{m,1},
		\iota_{\mathfrak{X}}^{m,2},
		\cdots
	\rangle
		&\quad \textrm{if} \;
	m < n = \infty .
\end{cases}
\end{equation}

\end{defn}

\begin{prop}
\label{prop:iotaX}
Let
$
\ell, m, n \in \mathbb{N}^+
$.
\begin{enumerate}
\item
If $m \ge n$,
we have
\begin{equation}
\label{eq:iotaXretr}
\iota_{\mathfrak{X}}^{m,n}
	\circ
\iota_{\mathfrak{X}}^{n,m}
		=
1_{\mathfrak{S}^n \mathfrak{X}} .
\end{equation}

\item
If $\ell \ge m \ge n$
or $\ell \le m \le n$,
we have
\begin{equation}
\label{eq:itoaXcomp}
\iota_{\mathfrak{X}}^{m,n}
	\circ
\iota_{\mathfrak{X}}^{\ell,m}
		=
\iota_{\mathfrak{X}}^{\ell,n}.
\end{equation}
\end{enumerate}
\end{prop}
\begin{proof}
\begin{enumerate}
\item
When $n = \infty$,
we have $m = \infty$ as well,
Hence,
\if \MkLonger  1
\[
\iota_{\mathfrak{X}}^{\infty,\infty}
	\circ
\iota_{\mathfrak{X}}^{\infty,\infty}
		=
1_{\mathfrak{S}^{\infty} \mathfrak{X}}
	\circ
1_{\mathfrak{S}^{\infty} \mathfrak{X}}
		=
1_{\mathfrak{S}^{\infty} \mathfrak{X}}  .
\]
\else 
$
\iota_{\mathfrak{X}}^{\infty,\infty}
	\circ
\iota_{\mathfrak{X}}^{\infty,\infty}
		=
1_{\mathfrak{S}^{\infty} \mathfrak{X}} 
	\circ
1_{\mathfrak{S}^{\infty} \mathfrak{X}} 
		=
1_{\mathfrak{S}^{\infty} \mathfrak{X}} .
$
\fi 

When $m=\infty$ and $n < \infty$,
\if \MkLonger  1
\[
\iota_{\mathfrak{X}}^{m,n}
	\circ
\iota_{\mathfrak{X}}^{n,m}
		=
j_n
	\circ
\langle
	\iota_{\mathfrak{X}}^{n,1},
	\iota_{\mathfrak{X}}^{n,2},
	\cdots
\rangle
		=
\iota_{\mathfrak{X}}^{n,n}
		=
1_{\mathfrak{S}^n \mathfrak{X}}
\]
\else 
$
\iota_{\mathfrak{X}}^{m,n}
	\circ
\iota_{\mathfrak{X}}^{n,m}
		=
j_n
	\circ
\langle
	\iota_{\mathfrak{X}}^{n,1},
	\iota_{\mathfrak{X}}^{n,2},
	\cdots
\rangle
		=
\iota_{\mathfrak{X}}^{n,n}
		=
1_{\mathfrak{S}^n \mathfrak{X}}
$
\fi 

We prove the case
when 
$
\infty > m \ge n,
$
by induction on $m$.
If $m=n$,
\if \MkLonger  1
\[
\iota_{\mathfrak{X}}^{m,n}
	\circ
\iota_{\mathfrak{X}}^{n,m}
		=
1_{\mathfrak{S}^n \mathfrak{X}}
	\circ
1_{\mathfrak{S}^n \mathfrak{X}}
		=
1_{\mathfrak{S}^n \mathfrak{X}} .
\]
\else 
$
\iota_{\mathfrak{X}}^{m,n}
	\circ
\iota_{\mathfrak{X}}^{n,m}
		=
1_{\mathfrak{S}^n \mathfrak{X}} 
	\circ
1_{\mathfrak{S}^n \mathfrak{X}}
		=
1_{\mathfrak{S}^n \mathfrak{X}} .
$
\fi 
For induction part, we have
\if \MkLonger  1
\begin{align*}
\iota_{\mathfrak{X}}^{m+1,n}
	\circ
\iota_{\mathfrak{X}}^{n,m+1}
		&=
(
\iota_{\mathfrak{X}}^{m,n}
	\circ
\mu^{\mathfrak{S}}_{
	\mathfrak{S}^{m-1}
	\mathfrak{X}
}
)
	\circ
(
\eta^{\mathfrak{S}}_{\mathfrak{S}^{m} \mathfrak{X}}
	\circ
\iota_{\mathfrak{X}}^{n,m}
)
		\\&=
\iota_{\mathfrak{X}}^{m,n}
	\circ
(
\mu^{\mathfrak{S}}_{\mathfrak{S}^{m-1} \mathfrak{X}}
	\circ
\eta^{\mathfrak{S}}_{\mathfrak{S}^{m} \mathfrak{X}}
)
	\circ
\iota_{\mathfrak{X}}^{n,m}
		\\&=
\iota_{\mathfrak{X}}^{m,n}
	\circ
1_{\mathfrak{S}^{m} \mathfrak{X}}
	\circ
\iota_{\mathfrak{X}}^{n,m}
		\\&=
\iota_{\mathfrak{X}}^{m,n}
	\circ
\iota_{\mathfrak{X}}^{n,m}
		=
1_{\mathfrak{S}^{n} \mathfrak{X}}
\end{align*}
\else 
\begin{align*}
\iota_{\mathfrak{X}}^{m+1,n}
	\circ
\iota_{\mathfrak{X}}^{n,m+1}
		&=
(
\iota_{\mathfrak{X}}^{m,n}
	\circ
\mu^{\mathfrak{S}}_{\mathfrak{S}^{m-1} \mathfrak{X}}
)
	\circ
(
\eta^{\mathfrak{S}}_{\mathfrak{S}^{m} \mathfrak{X}}
	\circ
\iota_{\mathfrak{X}}^{n,m}
)
		=
\iota_{\mathfrak{X}}^{m,n}
	\circ
(
\mu^{\mathfrak{S}}_{\mathfrak{S}^{m-1} \mathfrak{X}}
	\circ
\eta^{\mathfrak{S}}_{\mathfrak{S}^{m} \mathfrak{X}}
)
	\circ
\iota_{\mathfrak{X}}^{n,m}
		\\&=
\iota_{\mathfrak{X}}^{m,n}
	\circ
1_{\mathfrak{S}^{m} \mathfrak{X}}
	\circ
\iota_{\mathfrak{X}}^{n,m}
		=
\iota_{\mathfrak{X}}^{m,n}
	\circ
\iota_{\mathfrak{X}}^{n,m}
		=
1_{\mathfrak{S}^{n} \mathfrak{X}}
\end{align*}
\fi 
in which the last equality comes from the assumption.

\item
If $m = n$, it is clear.
So we assume $m \ne n$.
Note that
for $\infty > m > n$,
we have
\begin{equation}
\label{eq:iterateMu}
\iota_{\mathfrak{X}}^{m,n}
		=
\mu^{\mathfrak{S}}_{\mathfrak{S}^{n-1} \mathfrak{X}}
	\circ
\mu^{\mathfrak{S}}_{\mathfrak{S}^{n} \mathfrak{X}}
	\circ
\mu^{\mathfrak{S}}_{\mathfrak{S}^{n+1} \mathfrak{X}}
	\circ
\cdots
	\circ
\mu^{\mathfrak{S}}_{\mathfrak{S}^{m-2} \mathfrak{X}} .
\end{equation}
When $\infty = \ell > m > n$,
\if \MkLonger  1
\begin{align*}
\iota_{\mathfrak{X}}^{m, n}
	\circ
\iota_{\mathfrak{X}}^{\infty,m}
			&=
(
\mu^{\mathfrak{S}}_{\mathfrak{S}^{n-1} \mathfrak{X}}
	\circ
\mu^{\mathfrak{S}}_{\mathfrak{S}^{n} \mathfrak{X}}
	\circ
\mu^{\mathfrak{S}}_{\mathfrak{S}^{n+1} \mathfrak{X}}
	\circ
\cdots
	\circ
\mu^{\mathfrak{S}}_{\mathfrak{S}^{m-2} \mathfrak{X}} .
)
	\circ
j_m
			&(\textrm{by (\ref{eq:iterateMu})})
			\\&=
\mu^{\mathfrak{S}}_{\mathfrak{S}^{n-1} \mathfrak{X}}
	\circ
\mu^{\mathfrak{S}}_{\mathfrak{S}^{n} \mathfrak{X}}
	\circ
\mu^{\mathfrak{S}}_{\mathfrak{S}^{n+1} \mathfrak{X}}
	\circ
\cdots
	\circ
(
\mu^{\mathfrak{S}}_{\mathfrak{S}^{m-2} \mathfrak{X}} .
	\circ
j_m
)
			\\&=
\mu^{\mathfrak{S}}_{\mathfrak{S}^{n-1} \mathfrak{X}}
	\circ
\mu^{\mathfrak{S}}_{\mathfrak{S}^{n} \mathfrak{X}}
	\circ
\mu^{\mathfrak{S}}_{\mathfrak{S}^{n+1} \mathfrak{X}}
	\circ
\cdots
	\circ
(
\mu^{\mathfrak{S}}_{\mathfrak{S}^{m-3} \mathfrak{X}} .
	\circ
j_{m-1}
)
			&(\textrm{by (\ref{eq:jNmuJnP})})
			\\&=
\cdots
			\\&=
\mu^{\mathfrak{S}}_{\mathfrak{S}^{n-1} \mathfrak{X}}
	\circ
j_{n+1}
			\\&=
j_{n}
			&(\textrm{by (\ref{eq:jNmuJnP})})
			\\&=
\iota_{\mathfrak{X}}^{\infty,n} .
\end{align*}
\else 
by equations
			(\ref{eq:jNmuJnP})
and
			(\ref{eq:jNmuJnP}),
\begin{align*}
\iota_{\mathfrak{X}}^{m, n}
	\circ
\iota_{\mathfrak{X}}^{\infty,m}
			&=
(
\mu^{\mathfrak{S}}_{\mathfrak{S}^{n-1} \mathfrak{X}}
	\circ
\mu^{\mathfrak{S}}_{\mathfrak{S}^{n} \mathfrak{X}}
	\circ
\mu^{\mathfrak{S}}_{\mathfrak{S}^{n+1} \mathfrak{X}}
	\circ
\cdots
	\circ
\mu^{\mathfrak{S}}_{\mathfrak{S}^{m-2} \mathfrak{X}} .
)
	\circ
j_m
			\\&=
\mu^{\mathfrak{S}}_{\mathfrak{S}^{n-1} \mathfrak{X}}
	\circ
\mu^{\mathfrak{S}}_{\mathfrak{S}^{n} \mathfrak{X}}
	\circ
\mu^{\mathfrak{S}}_{\mathfrak{S}^{n+1} \mathfrak{X}}
	\circ
\cdots
	\circ
(
\mu^{\mathfrak{S}}_{\mathfrak{S}^{m-2} \mathfrak{X}} .
	\circ
j_m
)
			\\&=
\mu^{\mathfrak{S}}_{\mathfrak{S}^{n-1} \mathfrak{X}}
	\circ
\mu^{\mathfrak{S}}_{\mathfrak{S}^{n} \mathfrak{X}}
	\circ
\mu^{\mathfrak{S}}_{\mathfrak{S}^{n+1} \mathfrak{X}}
	\circ
\cdots
	\circ
(
\mu^{\mathfrak{S}}_{\mathfrak{S}^{m-3} \mathfrak{X}} .
	\circ
j_{m-1}
)
			\\&=
\cdots
			\\&=
\mu^{\mathfrak{S}}_{\mathfrak{S}^{n-1} \mathfrak{X}}
	\circ
j_{n+1}
			=
j_{n}
			=
\iota_{\mathfrak{X}}^{\infty,n} .
\end{align*}
\fi 

On the other hand,
note that
for $m < n < \infty$,
we have
\begin{equation}
\label{eq:iterateEta}
\iota_{\mathfrak{X}}^{m,n}
		=
\eta^{\mathfrak{S}}_{\mathfrak{S}^{n-1} \mathfrak{X}}
	\circ
\eta^{\mathfrak{S}}_{\mathfrak{S}^{n-2} \mathfrak{X}}
	\circ
\cdots
	\circ
\eta^{\mathfrak{S}}_{\mathfrak{S}^{m} \mathfrak{X}} .
\end{equation}
Then, when
$
l < m < n = \infty
$,
\if \MkLonger  1
\begin{align*}
\iota_{\mathfrak{X}}^{m,\infty}
	\circ
\iota_{\mathfrak{X}}^{\ell, m}
			&=
\langle
	\iota_{\mathfrak{X}}^{m,1},
	\iota_{\mathfrak{X}}^{m,2},
	\cdots
\rangle
	\circ
\iota_{\mathfrak{X}}^{\ell, m}
			\\&=
\langle
	\iota_{\mathfrak{X}}^{m,1}
		\circ
	\iota_{\mathfrak{X}}^{\ell, m},
		\,
	\iota_{\mathfrak{X}}^{m,2}
		\circ
	\iota_{\mathfrak{X}}^{\ell, m},
	\cdots
\rangle
			\\&=
\langle
	\iota_{\mathfrak{X}}^{\ell, 1},
	\iota_{\mathfrak{X}}^{\ell, 2},
	\cdots
\rangle
			=
\iota_{\mathfrak{X}}^{\ell,\infty} .
			&(\textrm{by (\ref{eq:iterateEta}) and (\ref{eq:iotaXretr})})
\end{align*}
\else 
by (\ref{eq:iterateEta}) and (\ref{eq:iotaXretr},
we have
\begin{align*}
\iota_{\mathfrak{X}}^{m,\infty}
	\circ
\iota_{\mathfrak{X}}^{\ell, m}
			&=
\langle
	\iota_{\mathfrak{X}}^{m,1},
	\iota_{\mathfrak{X}}^{m,2},
	\cdots
\rangle
	\circ
\iota_{\mathfrak{X}}^{\ell, m}
			=
\langle
	\iota_{\mathfrak{X}}^{m,1}
		\circ
	\iota_{\mathfrak{X}}^{\ell, m},
		\,
	\iota_{\mathfrak{X}}^{m,2}
		\circ
	\iota_{\mathfrak{X}}^{\ell, m},
	\cdots
\rangle
			=
\langle
	\iota_{\mathfrak{X}}^{\ell, 1},
	\iota_{\mathfrak{X}}^{\ell, 2},
	\cdots
\rangle
			=
\iota_{\mathfrak{X}}^{\ell,\infty} .
\end{align*}
\fi 

\end{enumerate}
\end{proof}

\begin{diagram}[htb]
\begin{equation*}
\xymatrix @C=50 pt @R=35 pt {
	\mathfrak{S}^{\ell} \mathfrak{X}
&
	\mathfrak{S}^m \mathfrak{X}
		\ar @{->}_{
			\iota_{\mathfrak{X}}^{m,\ell}
		} [l]
&
	\mathfrak{S}^{n} \mathfrak{X}
		\ar @{->}_{
			\iota_{\mathfrak{X}}^{n,m}
		} [l]
		\ar @/_2pc/_{
			\iota_{\mathfrak{X}}^{n,\ell}
		} [ll]
\\
&
	\mathfrak{S}^{\infty} \mathfrak{X}
		\ar @{->}^{
			\iota_{\mathfrak{X}}^{\infty,\ell}
		} [lu]
		\ar @{->}^{
			\iota_{\mathfrak{X}}^{\infty,m}
		} [u]
		\ar @{->}_{
			\iota_{\mathfrak{X}}^{\infty,n}
		} [ru]
}
\end{equation*}
\caption{Approximation maps for $
1 \le \ell \le m \le n < \infty
$}
\label{diag:approxMaps}
\end{diagram}

The measurable space
$
	\mathfrak{S}^{\infty} \mathfrak{X}
$
is called the 
\newword{universal uncertainty space}
starting from
$\mathfrak{X}$,
which contains all levels of uncertainty hierarchy over 
the given measurable space
$\mathfrak{X}$.
You can get an approximated capacity of any level by projecting with approximation maps
specified in the commutative diagram \ref{diag:approxMaps}.

\if \InclPM  1
\section{Giry monad}
\label{sec:GiryMonad}

\cite{lawvere_1962}
and
\cite{giry_1982}
introduced the concept of 
\newword{probability monad} .
Actually, 
maps
$\eta^{\mathfrak{S}}_{\mathfrak{X}}$
and
$\mu^{\mathfrak{S}}_{\mathfrak{X}}$
introduced in
Section \ref{sec:functorS}
would be components of the probability monad.

In this section,
we investigate a sufficient condition of our triple
$
(
\mathfrak{S},
\eta^{\mathfrak{S}},
\mu^{\mathfrak{S}}
)
$
becomes a probability monad.


Let
$
\mathfrak{S}
$
be a CM-functor
throughout this section.

\begin{lem}
\label{lem:LemThree}
Let
$
\mathfrak{X}
	=
(X, \Sigma_X)
$
be a measurable space
and
$
v
	\in
\mathfrak{S}^2
\mathfrak{X}
$
be an additive capacity.
Then, 
we have,
\begin{equation}
\label{eq:LemThreeOne}
I_{\mathfrak{X}}^{
	\mu^{\mathfrak{S}}_{\mathfrak{X}}(v)
}
		=
I_{
	\mathfrak{S}
	\mathfrak{X}
}^{
	v
}
	\circ
\xi^{\mathfrak{S}}_{\mathfrak{X}} .
\end{equation}
\end{lem}
\begin{proof}
We will show
for
$
f 
\in L^{\infty}(\mathfrak{X})
$,
\begin{equation}
\label{eq:LemThree}
I_{\mathfrak{X}}^{
	\mu^{\mathfrak{S}}_{\mathfrak{X}}(v)
}
(f)
		=
I_{
	\mathfrak{S}
	\mathfrak{X}
}^{
	v
}
(
	\xi^{\mathfrak{S}}_{\mathfrak{X}}(f)
) .
\end{equation}

First, 
we will prove for the case when
$f$ is a finite step function such as
\[
f
	:=
\sum_{i=1}^n
a_i
\mathbb{1}_X(A_i)
\]
where
$
a_1 \ge a_2 \ge \cdots \ge a_n
$
are decreasing real numbers
and
$A_i$
are mutually disjoint elements of
$\Sigma_X$.
But, we have
\if \MkLonger  1
\begin{align*}
I_{\mathfrak{X}}^{
	\mu^{\mathfrak{S}}_{\mathfrak{X}}(v)
}
(f)
		&=
I_{\mathfrak{X}}^{
	\mu^{\mathfrak{S}}_{\mathfrak{X}}(v)
}
(
	\sum_{i=1}^n
		a_i
		\mathbb{1}_X(A_i)
)
		\\&=
\sum_{i=1}^n
	(a_i - a_{i+1})
	\,
	\mu^{\mathfrak{S}}_{\mathfrak{X}}(v)
	\Big(
		\bigcup_{j=1}^i
			A_j
	\Big)
		&(\textrm{
			by Lemma \ref{lem:stepFunChoquetInt}
		})
		\\&=
\sum_{i=1}^n
	(a_i - a_{i+1})
	\,
	I_{
		\mathfrak{S}
		\mathfrak{X}
	}^{
		v
	}
	\Big(
		\varepsilon^{\mathfrak{S}}_{\mathfrak{X}}
		\big(
			\bigcup_{j=1}^i
				A_j
		\big)
	\Big)
		\\&=
I_{
	\mathfrak{S}
	\mathfrak{X}
}^{
	v
}
\Big(
\sum_{i=1}^n
	(a_i - a_{i+1})
	\,
		\varepsilon^{\mathfrak{S}}_{\mathfrak{X}}
		\big(
			\bigcup_{j=1}^i
				A_j
		\big)
\Big)
		&(\textrm{
			by Proposition \ref{prop:additiveNu}
		})
		\\&=
I_{
	\mathfrak{S}
	\mathfrak{X}
}^{
	v
}
(
	\xi^{\mathfrak{S}}_{\mathfrak{X}}(f)
) .
		&(\textrm{
			by 
Proposition \ref{prop:xiXstepF}
		})
\end{align*}
\else 
by Lemma \ref{lem:stepFunChoquetInt},
Proposition \ref{prop:additiveNu}
and
Proposition \ref{prop:xiXstepF},
\begin{align*}
I_{\mathfrak{X}}^{
	\mu^{\mathfrak{S}}_{\mathfrak{X}}(v)
}
(f)
		&=
I_{\mathfrak{X}}^{
	\mu^{\mathfrak{S}}_{\mathfrak{X}}(v)
}
(
	\sum_{i=1}^n
		a_i
		\mathbb{1}_X(A_i)
)
		=
\sum_{i=1}^n
	(a_i - a_{i+1})
	\,
	\mu^{\mathfrak{S}}_{\mathfrak{X}}(v)
	\Big(
		\bigcup_{j=1}^i
			A_j
	\Big)
		\\&=
\sum_{i=1}^n
	(a_i - a_{i+1})
	\,
	I_{
		\mathfrak{S}
		\mathfrak{X}
	}^{
		v
	}
	\Big(
		\varepsilon^{\mathfrak{S}}_{\mathfrak{X}}
		\big(
			\bigcup_{j=1}^i
				A_j
		\big)
	\Big)
		=
I_{
	\mathfrak{S}
	\mathfrak{X}
}^{
	v
}
\Big(
\sum_{i=1}^n
	(a_i - a_{i+1})
	\,
		\varepsilon^{\mathfrak{S}}_{\mathfrak{X}}
		\big(
			\bigcup_{j=1}^i
				A_j
		\big)
\Big)
		=
I_{
	\mathfrak{S}
	\mathfrak{X}
}^{
	v
}
(
	\xi^{\mathfrak{S}}_{\mathfrak{X}}(f)
) .
\end{align*}
\fi 

Next,
for general
$
f \in L^{\infty}(Y),
$
we create 
finite step functions
$\underline{f}_n$
and
$\bar{f}_n$
by the method described in
Note \ref{note:discretizeU}.
Then, by the monotonicity of the Choquet integration and 
$\xi^{\mathfrak{S}}_{\mathfrak{X}}$,
we have
\[
I_{\mathfrak{X}}^{\mu^{\mathfrak{S}}_{\mathfrak{X}}(v)}(
	\underline{f}_n
)
	\le
I_{\mathfrak{X}}^{\mu^{\mathfrak{S}}_{\mathfrak{X}}(v)}(
	f
)
	\le
I_{\mathfrak{X}}^{\mu^{\mathfrak{S}}_{\mathfrak{X}}(v)}(
	\bar{f}_n
)
\]
and
\[
I_{\mathfrak{S} \mathfrak{X}}^{v}\big(\xi^{\mathfrak{S}}_{\mathfrak{X}}(
	\underline{f}_n
)\big)
	\le
I_{\mathfrak{S} \mathfrak{X}}^{v}\big(\xi^{\mathfrak{S}}_{\mathfrak{X}}(
	f
)\big)
	\le
I_{\mathfrak{S} \mathfrak{X}}^{v}\big(\xi^{\mathfrak{S}}_{\mathfrak{X}}(
	\bar{f}_n
)\big) .
\]
Taking limits of both sides by letting $n$ go infinity,
we have
(\ref{eq:LemThreeOne})
for this general
$f$.

\end{proof}

The following example shows that 
we cannot extend
Lemma \ref{lem:LemThree}
to the cases when
$v$ is non-additive.

\begin{exmp}
\label{exmp:CEprobMon}
Let
$
\mathfrak{X}
	=
(X, \Sigma_X)
$
be a measurable space such that
$
X
	:=
\{ R, B, Y\}
$,
$
\Sigma_X := 2^X  
$
and
$U$
be a set of additive capacities 
over $\mathfrak{X}$
defined by
\begin{equation}
\label{eq:exmpSX}
U
	:=
\{
	u_{i,j}
\mid
	i,j \in \mathbb{N},
\,
	i + j \le N
\},
\end{equation}
where
$N$
is a fixed positive integer greater than $3$
and
\begin{equation}
\label{eq:exmpUij}
u_{i,j}(\{R\}) :=
\frac{i}{N},
		\quad
u_{i,j}(\{B\}) :=
\frac{j}{N},
		\quad
u_{i,j}(\{Y\}) :=
\frac{N-(i+j)}{N} .
\end{equation}
Then, we have
$
\#U
	=
\sum_{i=0}^N
\#\{
	j
\mid
	0 \le j \le N-i
\}
	=
\frac{1}{2}
(N-1)N
$.

Now we have
$
\mathfrak{L}
(X, \Sigma_X, U)
	=
(U, 2^U)
$
since
for 
every
$
i, j \ge 0
$
such that
$i + j \le N$,
\[
\{ u_{ij} \} 
		=
\big(
	\varepsilon^{\mathfrak{S}}_{\mathfrak{X}}
	(\{R\})
\big)^{-1}
\big(
	\{\frac{i}{N}\}
\big)
		\cap
\big(
	\varepsilon^{\mathfrak{S}}_{\mathfrak{X}}
	(\{B\})
\big)^{-1}
\big(
	\{\frac{j}{N}\}
\big)
		\in
\Sigma_{(X, \Sigma_X, U)} .
\]
So we can consider a CM-functor
$\mathfrak{S}$ such that 
$
\mathfrak{S}
\mathfrak{X}
	=
(U, 2^U)
$ .

Define a non-additive capacity $v$ on
$
\mathfrak{S}
\mathfrak{X}
$
by
with some fixed
$
\beta \ge 1
$,
\begin{equation}
\label{eq:exmpNAv}
v(A)
	:=
\Big(
\frac{
	\# A
}{
	\# U
}
\Big)^{\beta} ,
	\quad
(
A \in
\Sigma_{U}
) .
\end{equation}

Let
$f
\in
L^{\infty}(\mathfrak{X})
$
be an act defined by
\begin{equation}
\label{eq:exmpActF}
f
		:=
a
\mathbb{1}_X(\{R\})
	+
b
\mathbb{1}_X(\{R, B\})
		=
(a+b)
\mathbb{1}_X(\{R\})
	+
b
\mathbb{1}_X(\{B\}) ,
\end{equation}
where $a$ and $b$ are distinct positive numbers,
and we will check if
(\ref{eq:LemThree})
holds with it.

First, 
let us calculate the LHS of
(\ref{eq:LemThree}).
\begin{align*}
I_{\mathfrak{X}}^{
	\mu^{\mathfrak{S}}_{\mathfrak{X}}(v)
}
(f)
		&=
\{(a+b) - b\}
\mu^{\mathfrak{S}}_{\mathfrak{X}}(v)(\{R\})
	+
\{b - 0\}
\mu^{\mathfrak{S}}_{\mathfrak{X}}(v)(\{R, B\})
		\\&=
a
I_{\mathfrak{S} \mathfrak{X}}^v
\big( \varepsilon^{\mathfrak{S}}_{\mathfrak{X}}(\{R\}) \big)
	+
b
I_{\mathfrak{S} \mathfrak{X}}^v
\big( \varepsilon^{\mathfrak{S}}_{\mathfrak{X}}(\{R, B\}) \big) .
\end{align*}
By considering 
$
\varepsilon^{\mathfrak{S}}_{\mathfrak{X}}(\{R\})(u_{i,j})
	=
\frac{i}{N}
$,
define a subset
$A_k
	\subset
U
$
by
\begin{equation}
\label{eq:exmpAi}
A_k
		:=
\big( \varepsilon^{\mathfrak{S}}_{\mathfrak{X}}(\{R\}) \big)^{-1} 
\big(
	\frac{k}{N}
\big)
		=
\{
	u_{k,j}
\mid
	0 \le j \le N-k
\} .
\end{equation}
Then
we can write
\[
\varepsilon^{\mathfrak{S}}_{\mathfrak{X}}(\{R\})
		=
\frac{N}{N}
\mathbb{1}_{\mathfrak{S} \mathfrak{X}}
(A_N)
		+
\frac{N-1}{N}
\mathbb{1}_{\mathfrak{S} \mathfrak{X}}
(A_{N-1})
		+
\cdots
		+
\frac{k}{N}
\mathbb{1}_{\mathfrak{S} \mathfrak{X}}
(A_{k})
		+
\cdots
		+
\frac{1}{N}
\mathbb{1}_{\mathfrak{S} \mathfrak{X}}
(A_{1})
		+
\frac{0}{N}
\mathbb{1}_{\mathfrak{S} \mathfrak{X}}
(A_{0}) .
\]
Therefore,
\begin{align*}
I_{\mathfrak{S} \mathfrak{X}}^v
\big( \varepsilon^{\mathfrak{S}}_{\mathfrak{X}}(\{R\}) \big)
			&=
\sum_{k=N}^1
	\big(
		\frac{k}{N}
			-
		\frac{k-1}{N}
	\big)
	v\Big(
		\bigcup_{\ell=k}^N
			A_{\ell}
	\Big) 
			\\&=
\frac{1}{N}
\sum_{k=1}^N
	v\Big(
		\bigcup_{\ell=k}^N
			A_{\ell}
	\Big) 
			=
\frac{1}{N}
\sum_{k=1}^N
	\Big(
		\frac{
			(N+1-k)
			(N+2-k)
		}{
			(N-1)
			N
		}
	\Big)^{\beta}
\end{align*}
since
\[
\#\
\Big(
	\bigcup_{\ell=k}^N
		A_{\ell}
\Big)
		=
\sum_{\ell=k}^N
	\# A_{\ell}
		=
\sum_{\ell=k}^N
	(N+1-\ell)
		=
\frac{1}{2}
(N+1-k)
(N+2-k) .
\]
Similarly,
by considering
$
\varepsilon^{\mathfrak{S}}_{\mathfrak{X}}(\{R, B\})(u_{i,j})
	=
\frac{i+j}{N}
$,
for
$k = 0, 1, \cdots, N$,
define a subset
$
B_k
	\subset
U
$
by
\begin{equation}
\label{eq:exmpBk}
B_k
		:=
\big( \varepsilon^{\mathfrak{S}}_{\mathfrak{X}}(\{R, B\}) \big)^{-1} 
\big(
	\frac{k}{N}
\big)
		=
\{
	u_{i,j}
\mid
	i, j \ge 0,
\,
	i+j = k
\} .
\end{equation}
Then, we have
$
\# B_k
	=
k+1
$
and
can write
\[
\varepsilon^{\mathfrak{S}}_{\mathfrak{X}}(\{R, B\})
		=
\frac{N}{N}
\mathbb{1}_{\mathfrak{S} \mathfrak{X}}
(B_N)
		+
\frac{N-1}{N}
\mathbb{1}_{\mathfrak{S} \mathfrak{X}}
(B_{N-1})
		+
\cdots
		+
\frac{k}{N}
\mathbb{1}_{\mathfrak{S} \mathfrak{X}}
(B_{k})
		+
\cdots
		+
\frac{1}{N}
\mathbb{1}_{\mathfrak{S} \mathfrak{X}}
(B_{1})
		+
\frac{0}{N}
\mathbb{1}_{\mathfrak{S} \mathfrak{X}}
(B_{0}) .
\]
Therefore,
\if \MkLonger  1
\begin{align*}
I_{\mathfrak{S} \mathfrak{X}}^v
\big( \varepsilon^{\mathfrak{S}}_{\mathfrak{X}}(\{R, B\}) \big)
			&=
\sum_{k=N}^1
	\big(
		\frac{k}{N}
			-
		\frac{k-1}{N}
	\big)
	v\Big(
		\bigcup_{\ell=k}^N
			B_{\ell}
	\Big) 
			\\&=
\frac{1}{N}
\sum_{k=1}^N
	v\Big(
		\bigcup_{\ell=k}^N
			B_{\ell}
	\Big) 
			\\&=
\frac{1}{N}
\sum_{k=1}^N
	\Big(
		\frac{
			(N+1-k)
			(N+2+k)
		}{
			(N+1)
			(N+2)
		}
	\Big)^{\beta}
\end{align*}
\else 
\begin{align*}
I_{\mathfrak{S} \mathfrak{X}}^v
\big( \varepsilon^{\mathfrak{S}}_{\mathfrak{X}}(\{R, B\}) \big)
			&=
\sum_{k=N}^1
	\big(
		\frac{k}{N}
			-
		\frac{k-1}{N}
	\big)
	v\Big(
		\bigcup_{\ell=k}^N
			B_{\ell}
	\Big) 
			\\&=
\frac{1}{N}
\sum_{k=1}^N
	v\Big(
		\bigcup_{\ell=k}^N
			B_{\ell}
	\Big) 
			=
\frac{1}{N}
\sum_{k=1}^N
	\Big(
		\frac{
			(N+1-k)
			(N+2+k)
		}{
			(N-1)
			N
		}
	\Big)^{\beta}
\end{align*}
\fi 
since
\[
\#\
\Big(
	\bigcup_{\ell=k}^N
		B_{\ell}
\Big)
		=
\sum_{\ell=k}^N
	\# B_{\ell}
		=
\sum_{\ell=k}^N
	(\ell + 1)
		=
\frac{1}{2}
(N+1-k) 
(N+2+k) .
\]
Hence,
\begin{equation}
\label{eq:exmpIsxMuV}
I_{\mathfrak{X}}^{
	\mu^{\mathfrak{S}}_{\mathfrak{X}}
	(v)
}
(f)
		=
\frac{a}{N}
\sum_{k=1}^N
	\Big(
		\frac{
			(N+1-k)
			(N+2-k)
		}{
			(N-1)
			N
		}
	\Big)^{\beta}
			+
\frac{b}{N}
\sum_{k=1}^N
	\Big(
		\frac{
			(N+1-k)
			(N+2+k)
		}{
			(N-1)
			N
		}
	\Big)^{\beta} .
\end{equation}

Next,
we will compute the RHS of
(\ref{eq:LemThree}) .
\begin{align*}
\xi^{\mathfrak{S}}_{\mathfrak{X}}(f)(u_{i,j})
		&=
\{(a+b) - b\}
u_{i,j}(\{R\})
	+
\{b - 0\}
u_{i,j}(\{R,B\})
		\\&=
a
\frac{i}{N}
	+
b
\frac{i+j}{N}
		=
\frac{1}{N}
\big(
	(a+b)i
		+
	bj
\big) .
\end{align*}
Let
\[
K :=
\{
	(a+b)i + bj
\mid
	i, j \ge 0,
\,
	i+j \le N
\} .
\]
Then, since $K$ is a finite set,
we can write
\[
K = \{
	k_0, k_1, \cdots, k_M
\},
\]
where
$
M := \#K - 1
$
and
$
0 = k_0 < k_1 < \cdots < k_M = (a+b)N .
$
Then, by defining
\[
C_{\ell}
		:=
\big(
	\xi^{\mathfrak{S}}_{\mathfrak{X}}(f)
\big)^{-1}
\big(
	\frac{
		k_{\ell}
	}{
		N
	}
\big)
		=
\{
	u_{i,j}
\mid
	(a+b)i + bj = k_{\ell}
\} ,
\]
we have
\[
\xi^{\mathfrak{S}}_{\mathfrak{X}}(f)
	=
\sum_{\ell = M}^1
	\frac{
		k_{\ell}
	}{
		N
	}
	\mathbb{1}_{\mathfrak{S} \mathfrak{X}}
	(C_{\ell}) .
\]
Hence,
\if \MkLonger  1
\begin{align}
I_{
	\mathfrak{S} \mathfrak{X}
}^v
\big(
	\xi^{\mathfrak{S}}_{\mathfrak{X}}(f)
\big)
			&=
\frac{1}{N}
\sum_{\ell = 1}^M
	(
		k_{\ell}
			-
		k_{\ell-1}
	)
\,
	v\Big(
		\bigcup_{m=\ell}^M
			C_{m}
	\Big) 
			\nonumber
			\\&=
\frac{1}{N}
\sum_{\ell = 1}^M
	(
		k_{\ell}
			-
		k_{\ell-1}
	)
	\Big(
		\frac{
			\sum_{m=\ell}^M
				\# C_{m}
		}{
			\frac{1}{2}
			N-1
			N
		}
	\Big)^{\beta}.
\label{eq:empIsXVxiF}
\end{align}
\else 
\begin{equation}
I_{
	\mathfrak{S} \mathfrak{X}
}^v
\big(
	\xi^{\mathfrak{S}}_{\mathfrak{X}}(f)
\big)
			=
\frac{1}{N}
\sum_{\ell = 1}^M
	(
		k_{\ell}
			-
		k_{\ell-1}
	)
\,
	v\Big(
		\bigcup_{m=\ell}^M
			C_{m}
	\Big)
			=
\frac{1}{N}
\sum_{\ell = 1}^M
	(
		k_{\ell}
			-
		k_{\ell-1}
	)
	\Big(
		\frac{
			\sum_{m=\ell}^M
				\# C_{m}
		}{
			\frac{1}{2}
			(N-1)
			N
		}
	\Big)^{\beta}.
\label{eq:empIsXVxiF}
\end{equation}
\fi 

Now let
$N := 3$,
$a := 2$
and
$b := 1$.
Then,
by
(\ref{eq:exmpIsxMuV}),
\begin{align}
I_{\mathfrak{X}}^{\mu(v)}
(f)
			&=
\frac{2}{3}
\Big(
	\big(
		\frac{3 \cdot 4}{6}
	\big)^{\beta}
+
	\big(
		\frac{2 \cdot 3}{6}
	\big)^{\beta}
+
	\big(
		\frac{1 \cdot 2}{6}
	\big)^{\beta}
\Big)
		+
\frac{1}{3}
\Big(
	\big(
		\frac{3 \cdot 6}{6}
	\big)^{\beta}
+
	\big(
		\frac{2 \cdot 7}{6}
	\big)^{\beta}
+
	\big(
		\frac{1 \cdot 8}{6}
	\big)^{\beta}
\Big)
			\nonumber
			\\&=
\frac{1}{3 \cdot 3^{\beta}}
\big(
	2
	(
		6^{\beta} 
	+
		3^{\beta} 
	+
		1^{\beta} 
	)
		+
	(
		9^{\beta} 
	+
		7^{\beta} 
	+
		4^{\beta} 
	)
\big) .
			\label{eq:exmpIsxMuVConcrete}
\end{align}

On the other hand, we have
\[
M = 8, \quad
k_{\ell}
	=
\begin{cases}
	\ell
		\quad& \textrm{if} \quad
	0 \le \ell \le 7
\\
	9
		\quad& \textrm{if} \quad
	\ell = M
\end{cases}
\]
and
\begin{align*}
		&
C_0 = \{ u_{0,0} \}, \;
C_1 = \{ u_{0,1} \}, \;
C_2 = \{ u_{0,2} \}, \;
C_3 = \{ u_{0,3}, u_{1,0} \}, \;
C_4 = \{ u_{1,1} \}, \;
		\\&
C_5 = \{ u_{1,2} \}, \;
C_6 = \{ u_{2,0} \}, \;
C_7 = \{ u_{2,1} \}, \;
C_8 = \{ u_{3,0} \}.
\end{align*}
Therefore,
by
(\ref{eq:empIsXVxiF}),
\begin{align}
			&
I_{
	\mathfrak{S} \mathfrak{X}
}^v
\big(
	\xi^{\mathfrak{S}}_{\mathfrak{X}}(f)
\big)
			\nonumber
			\\=&
\frac{1}{3}
\Big(
	(9-7)
	\big( \frac{1}{3} \big)^{\beta}
		+
	(7-6)
	\big( \frac{2}{3} \big)^{\beta}
		+
	(6-5)
	\big( \frac{3}{3} \big)^{\beta}
		+
	(5-4)
	\big( \frac{4}{3} \big)^{\beta}
		+
	(4-3)
	\big( \frac{5}{3} \big)^{\beta}
			\nonumber
			\\&
		+
	(3-2)
	\big( \frac{7}{3} \big)^{\beta}
		+
	(2-1)
	\big( \frac{8}{3} \big)^{\beta}
		+
	(1-0)
	\big( \frac{9}{3} \big)^{\beta}
\Big)
			\nonumber
			\\=&
\frac{1}{3 \cdot 3^{\beta}}
\big(
	2
	\cdot
	1^{\beta}
		+
	2^{\beta}
		+
	3^{\beta}
		+
	4^{\beta}
		+
	5^{\beta}
		+
	7^{\beta}
		+
	8^{\beta}
		+
	9^{\beta}
\big) .
			\label{eq:empIsXVxiFconcrete}
\end{align}

Then, the difference between
			(\ref{eq:empIsXVxiFconcrete})
and
			(\ref{eq:exmpIsxMuVConcrete})
is
\begin{equation}
\label{eq:exmpIsDiff}
I_{
	\mathfrak{S} \mathfrak{X}
}^v
\big(
	\xi^{\mathfrak{S}}_{\mathfrak{X}}(f)
\big)
		-
I_{\mathfrak{X}}^{\mu(v)}
(f)
			=
\frac{1}{3 \cdot 3^{\beta}}
\big(
	2^{\beta}
		-
	3^{\beta}
		+
	5^{\beta}
		-
	2 \cdot
	6^{\beta}
		+
	8^{\beta}
\big) ,
\end{equation}
which is $0$
when
$\beta = 1$,
i.e.
$v$
is additive.
However, 
(\ref{eq:exmpIsDiff})
may not be
$0$
if 
$\beta > 1$.

That is,
the equation
(\ref{eq:LemThree})
does not hold in general
when
$v$
is non-additive.

\end{exmp}

\begin{defn}
\label{defn:addtiveS}
A CM-functor
$\mathfrak{S}$
is said
to be 
\newword{additive}
if
every
$
u \in \mathfrak{S} \mathfrak{X}
$
is additive
for any measurable space
$\mathfrak{X}$.
\end{defn}

\begin{prop}
\label{prop:GiryAssoc}
If 
$\mathfrak{S}$
is additive,
the diagram in
Diagram \ref{diag:defGiryAssoc}
commutes.

\begin{diagram}[htb]
\begin{equation*}
\xymatrix{
    \mathfrak{S}^3
      \ar @{->}_{\mu^{\mathfrak{S}} \mathfrak{S}} [d]
      \ar @{->}^{\mathfrak{S} \mu^{\mathfrak{S}}} [r]
&
    \mathfrak{S}^2
      \ar @{->}^{\mu^{\mathfrak{S}}} [d]
 \\
    \mathfrak{S}^2
      \ar @{->}^{\mu^{\mathfrak{S}}} [r]
&
    \mathfrak{S}
}
\end{equation*}
\caption{Giry associativity}
\label{diag:defGiryAssoc}
\end{diagram}
\end{prop}
\begin{proof}
Let
$
\mathfrak{X}
	=
(X, \Sigma_X)
$
be a measurable space.
The natural transformations described in
Diagram \ref{diag:defGiryAssoc}
are defined as
$
(
\mu^{\mathfrak{S}}
\mathfrak{S}
)_{\mathfrak{X}}
	:=
\mu^{\mathfrak{S}}_{
	\mathfrak{S}
	\mathfrak{X}
}
$
and
$
(
\mathfrak{S}
\mu^{\mathfrak{S}}
)_{\mathfrak{X}}
	:=
\mathfrak{S}(
	\mu^{\mathfrak{S}}_{\mathfrak{X}}
) .
$

Let
$
A \in \Sigma_X
$.
First, we will show that
\begin{equation}
\label{eq:gmProofOne}
\varepsilon^{\mathfrak{S}}_{\mathfrak{X}}(A)
	\circ
\mu^{\mathfrak{S}}_{\mathfrak{X}}
		=
\xi^{\mathfrak{S}}_{
	\mathfrak{S}
	\mathfrak{X}
}
(
	\varepsilon^{\mathfrak{S}}_{\mathfrak{X}}(A)
) .
\end{equation}
But for 
$
v
	\in
\mathfrak{S}^2
\mathfrak{X}
$,
we have
\if \MkLonger  1
\begin{align*}
(
\varepsilon^{\mathfrak{S}}_{\mathfrak{X}}(A)
	\circ
\mu^{\mathfrak{S}}_{\mathfrak{X}}
)
(v)
		&=
\varepsilon^{\mathfrak{S}}_{\mathfrak{X}}(A)
(
	\mu^{\mathfrak{S}}_{\mathfrak{X}}
	(v)
)
		\\&=
\mu^{\mathfrak{S}}_{\mathfrak{X}}
(v)
(A)
		\\&=
I_{
	\mathfrak{S} \mathfrak{X}
}^{
	v
}
\big(
	\varepsilon^{\mathfrak{S}}_{\mathfrak{X}}(A)
\big)
		\\&=
\xi^{\mathfrak{S}}_{
	\mathfrak{S}
	\mathfrak{X}
}
(
	\varepsilon^{\mathfrak{S}}_{\mathfrak{X}}(A)
)
(v) .
\end{align*}
\else 
\begin{align*}
(
\varepsilon^{\mathfrak{S}}_{\mathfrak{X}}(A)
	\circ
\mu^{\mathfrak{S}}_{\mathfrak{X}}
)
(v)
		&=
\varepsilon^{\mathfrak{S}}_{\mathfrak{X}}(A)
(
	\mu^{\mathfrak{S}}_{\mathfrak{X}}
	(v)
)
		=
\mu^{\mathfrak{S}}_{\mathfrak{X}}
(v)
(A)
		=
I_{
	\mathfrak{S} \mathfrak{X}
}^{
	v
}
\big(
	\varepsilon^{\mathfrak{S}}_{\mathfrak{X}}(A)
\big)
		=
\xi^{\mathfrak{S}}_{
	\mathfrak{S}
	\mathfrak{X}
}
(
	\varepsilon^{\mathfrak{S}}_{\mathfrak{X}}(A)
)
(v) .
\end{align*}
\fi 
Hence, we obtain
(\ref{eq:gmProofOne}).
Now for
$
w
	\in
\mathfrak{S}^3
\mathfrak{X}
$
and
$
A \in \Sigma_X
$,
we have
\if \MkLonger  1
\begin{align*}
(
\mu^{\mathfrak{S}}
	\circ
\mathfrak{S}
\mu^{\mathfrak{S}}
)_{\mathfrak{X}}
(w)(A)
		&=
(
\mu^{\mathfrak{S}}_{\mathfrak{X}}
	\circ
\mathfrak{S}(
	\mu^{\mathfrak{S}}_{\mathfrak{X}}
)
)
(w)(A)
		\\&=
\mu^{\mathfrak{S}}_{\mathfrak{X}}
(
	\mathfrak{S}(
		\mu^{\mathfrak{S}}_{\mathfrak{X}}
	)
	(w)
)
(A)
		\\&=
I_{
	\mathfrak{S} \mathfrak{X}
}^{
	\mathfrak{S} (\mu^{\mathfrak{S}}_{\mathfrak{X}})
	(w)
}
\big(
	\varepsilon^{\mathfrak{S}}_{\mathfrak{X}}(A)
\big)
		\\&=
I_{
	\mathfrak{S} \mathfrak{X}
}^{
	w
		\circ
	(\mu^{\mathfrak{S}}_{\mathfrak{X}})^{-1}
}
\big(
	\varepsilon^{\mathfrak{S}}_{\mathfrak{X}}(A)
\big)
		\\&=
I_{
	\mathfrak{S}^2 \mathfrak{X}
}^w
\big(
	\varepsilon^{\mathfrak{S}}_{\mathfrak{X}}(A)
		\circ
	\mu^{\mathfrak{S}}_{\mathfrak{X}}
\big)
			&(\textrm{
by Lemma \ref{lem:lemTwo}
			})
		\\&=
I_{
	\mathfrak{S}^2 \mathfrak{X}
}^w
\big(
	\xi^{\mathfrak{S}}_{
		\mathfrak{S}
		\mathfrak{X}
	}
	(
		\varepsilon^{\mathfrak{S}}_{\mathfrak{X}}(A)
	)
\big)
			&(\textrm{
by (\ref{eq:gmProofOne})
			})
		\\&=
I_{
	\mathfrak{S} \mathfrak{X}
}^{
	\mu^{\mathfrak{S}}_{
		\mathfrak{S} \mathfrak{X}
	}
	(w)
}
\big(
	\varepsilon^{\mathfrak{S}}_{\mathfrak{X}}(A)
\big)
			&(\textrm{
by Lemma \ref{lem:LemThree}
			})
		\\&=
\mu^{\mathfrak{S}}_{\mathfrak{X}}
(
\mu^{\mathfrak{S}}_{
	\mathfrak{S}
	\mathfrak{X}
}
(w)
)
(A)
		\\&=
(
\mu^{\mathfrak{S}}_{\mathfrak{X}}
	\circ
\mu^{\mathfrak{S}}_{
	\mathfrak{S}
	\mathfrak{X}
}
)
(w)(A)
		\\&=
(
\mu^{\mathfrak{S}}
	\circ
\mu^{\mathfrak{S}}
\mathfrak{S}
)_{\mathfrak{X}}
(w)(A) .
\end{align*}
\else 
by Lemma \ref{lem:lemTwo},
(\ref{eq:gmProofOne})
and
Lemma \ref{lem:LemThree},
\begin{align*}
		&
(
\mu^{\mathfrak{S}}
	\circ
\mathfrak{S}
\mu^{\mathfrak{S}}
)_{\mathfrak{X}}
(w)(A)
		=
(
\mu^{\mathfrak{S}}_{\mathfrak{X}}
	\circ
\mathfrak{S}(
	\mu^{\mathfrak{S}}_{\mathfrak{X}}
)
)
(w)(A)
		=
I_{
	\mathfrak{S} \mathfrak{X}
}^{
	\mathfrak{S} (\mu^{\mathfrak{S}}_{\mathfrak{X}})
	(w)
}
\big(
	\varepsilon^{\mathfrak{S}}_{\mathfrak{X}}(A)
\big)
		=
I_{
	\mathfrak{S} \mathfrak{X}
}^{
	w
		\circ
	(\mu^{\mathfrak{S}}_{\mathfrak{X}})^{-1}
}
\big(
	\varepsilon^{\mathfrak{S}}_{\mathfrak{X}}(A)
\big)
		\\&=
I_{
	\mathfrak{S}^2 \mathfrak{X}
}^w
\big(
	\varepsilon^{\mathfrak{S}}_{\mathfrak{X}}(A)
		\circ
	\mu^{\mathfrak{S}}_{\mathfrak{X}}
\big)
		=
I_{
	\mathfrak{S}^2 \mathfrak{X}
}^w
\big(
	\xi^{\mathfrak{S}}_{
		\mathfrak{S}
		\mathfrak{X}
	}
	(
		\varepsilon^{\mathfrak{S}}_{\mathfrak{X}}(A)
	)
\big)
		=
I_{
	\mathfrak{S} \mathfrak{X}
}^{
	\mu^{\mathfrak{S}}_{
		\mathfrak{S} \mathfrak{X}
	}
	(w)
}
\big(
	\varepsilon^{\mathfrak{S}}_{\mathfrak{X}}(A)
\big)
		=
\mu^{\mathfrak{S}}_{\mathfrak{X}}
(
\mu^{\mathfrak{S}}_{
	\mathfrak{S}
	\mathfrak{X}
}
(w)
)
(A)
		\\&=
(
\mu^{\mathfrak{S}}_{\mathfrak{X}}
	\circ
\mu^{\mathfrak{S}}_{
	\mathfrak{S}
	\mathfrak{X}
}
)
(w)(A)
		=
(
\mu^{\mathfrak{S}}
	\circ
\mu^{\mathfrak{S}}
\mathfrak{S}
)_{\mathfrak{X}}
(w)(A) .
\end{align*}
\fi 
Therefore,
we showed the diagram commutes.

\end{proof}


\begin{thm}{\normalfont{[\cite{giry_1982}]}}
\label{thm:GiryMonad}
If
$
\mathfrak{S}
$
is an additive CM-functor,
the triple
$
(
	\mathfrak{S}, \eta^{\mathfrak{S}}, \mu^{\mathfrak{S}}
)
$
is a monad.
\end{thm}
\begin{proof}
Obvious by
Proposition \ref{prop:GiryUnit}
and
Proposition \ref{prop:GiryAssoc}.

\end{proof}

\fi 
\fi 

\if \InclSC  1
\input{KleisliCat}
\input{SConGiry}
\input{catFilt}
\input{furtherTopics}
\fi 

\section{Concluding remarks}
\label{sec:concl}

The purpose of this paper was 
to provide a framework to systematically handle the hierarchy of uncertainty in a multi-layered manner, 
which previously had only two levels: risk and ambiguity.
To this end, 
we first introduced an uncertainty space as a measurable space equipped with multiple capacities as its measures.
This is a generalized concept of probability space.
We then introduced a sequence of uncertainty spaces, called a U-sequence, 
to represent the hierarchy of uncertainty structures.
Using these concepts, 
we provided a concrete example of the third layer of uncertainty by demonstrating Ellsberg's paradox.

Next, 
we decided to use category theory to overview the structure of the uncertainty spaces, 
and introduced two categories, $\Unc$ and $\mpUnc$.
Of these, $\Unc$ is a natural generalization of $\Prob$, 
the category of probability spaces introduced in \cite{AR_2019}.
That is, 
it is expected to be a framework for computing conditional expectations 
along $\Unc$-maps between uncertainty spaces.
Similarly, we introduced the $U^G$-map as an arrow between U-sequences,
forming 
$\USeq^G$,
the category of U-sequences.
This arrow can be viewed as a generalization of the entropic value measure.
We introduced the lift-up functor
$\mathfrak{L}
	: \mpUnc \to \Mble
$
that maps a given uncertainty space to the measurable space becoming a basis of higher level uncertainty space.
The functor
$\mathfrak{L}$
was used to define CM-functors that is considered as envelopes of U-sequences.
By the iterative application of the CM-functor to a given measurable space,
we constructed the universal uncertainty space
that has a potential to be the basis of multi-level uncertainty theory
because it has the uncertainty spaces of all levels as its projections.
Lastly, we confirmed 
a sufficient condition for the functor
$\mathfrak{S}$
to be a probability monad developed by
\cite{lawvere_1962}
and
\cite{giry_1982}.

\lspace

We would like to find concrete examples where 
$V_{\infty}(f)$
works meaningfully in the real world, 
such as in financial risk management \cite{adachi_2014crm}, which was beyond the scope of this version.
In other words, 
one of our goals is to concretely show the necessity of dealing with higher-order uncertainty structures.

On the other hand, 
it is already difficult for humans to fully follow the thought process of generative AIs such as ChatGPT and Stable Diffusion, 
and their black-box nature is inherent in uncertainty.
However, current generative AIs are still tractable in the sense that they are developed by human beings.
But, what will happen when AI starts to develop its own generative AI by itself in the future?
It seems to me that what mankind will be confronted with at that time is precisely a higher level of uncertainty.

As I understand it, 
we are just at the starting point of handling uncertainty in the engineering sense,
seeing we are still at the 2nd layer.
However, 
if such an iterative mechanism that escalates the layers of uncertainty is implemented in society, 
it would be easy to imagine that the level of uncertainty will go infinity.
Therefore, it may not be a bad idea to prepare for such a situation.








\end{document}